\documentclass[11pt]{article}
\usepackage{amsmath,amssymb,amsfonts,amsthm,epsfig,verbatim, framed}
\usepackage{times}
\usepackage{color}
\newif\ifhyper\IfFileExists{hyperref.sty}{\hypertrue}{\hyperfalse}
\hypertrue
\ifhyper\usepackage{hyperref}\fi

\def\confversion{0}

\newtheorem{assumption}{Assumption}
\newtheorem{theorem}{Theorem}

\newtheorem{lemma}[theorem]{Lemma}
\newtheorem{proposition}[theorem]{Proposition}
\newtheorem{corollary}[theorem]{Corollary}
\newtheorem{claim}[theorem]{Claim}
\newtheorem{fact}[theorem]{Fact}
\newtheorem{obs}[theorem]{Observation}
\newtheorem{definition}[theorem]{Definition}
\newtheorem{remark}[theorem]{Remark}
\newcommand{\eps}{\epsilon}

\newcommand{\etal}{{\em et al.\ }}

\newcommand{\dtv}{d_{\rm TV}}

\newcommand{\sample}{\mathrm{sample}}
\newcommand{\sign}{\mathrm{sign}}

\newcommand{\ignore}[1]{}

\newcommand{\cref}[1]{Corollary~\ref{cor:#1}}

\newcommand{\calU}{{\cal U}}

\newcommand{\bits}{\{-1,1\}}
\newcommand{\bn}{\bits^n}

\newcommand{\R}{{\mathbb{R}}}

\newcommand{\Z}{{\mathbb Z}}
\newcommand{\E}{\operatorname{{\bf E}}}

\newcommand{\littlesum}{\mathop{\textstyle \sum}}

\newcommand{\poly}{\mathrm{poly}}

\newcommand{\eqdef}{\stackrel{\textrm{def}}{=}}

\newcommand{\mB}{\mathop{Ber}}

\renewcommand{\Pr}{\operatorname{{\bf Pr}}}

\newcommand{\U}{{\cal U}}

\newcommand{\den}{\mathrm{den}}
\newcommand{\sq}{\mathrm{SQ}}
\newcommand{\A}{\mathcal{A}}
\newcommand{\C}{\mathcal{C}}
\newcommand{\D}{\mathcal{D}}
\newcommand{\stat}{\mathrm{STAT}}
\newcommand{\sqs}{\mathrm{SQ-SIM}}
\newcommand{\W}{\mathcal{W}}
\newcommand{\eval}{\mathrm{EVAL}}
\newcommand{\co}{\mathrm{count}}
\newcommand{\gen}{\mathrm{gen}}
\newcommand{\inv}{\mathrm{inv}}
\newcommand{\reconstruct}{\mathrm{reconstruct}}
\newcommand{\ltf}{\mathbf{LTF}}
\newcommand{\mt}{\mathrm{MT}}
\newcommand{\dnfns}{{\mathbf{DNF}_{n,s}}}
\newcommand{\dnf}{\mathbf{DNF}}
\newcommand{\disj}{\mathbf{DISJ}}

\newcommand{\new}[1]{{\color{black} {#1}}}
\newcommand{\newer}[1]{{\color{black} {#1}}}

\newcommand{\nnew}[1]{{\color{black} {#1}}}

\makeatletter
\renewcommand{\subsection}{\@startsection{subsection}{2}{0pt}{-6pt}{-5pt}{\normalsize\bf}}
\renewcommand{\subsubsection}{\@startsection{subsubsection}{3}{0pt}{-12pt}{-5pt}{\normalsize\bf}}
\makeatother

\newcommand{\mnote}[1]{ \marginpar{\tiny\bf
            \begin{minipage}[t]{0.5in}
              \raggedright #1
           \end{minipage}}}

\newcommand{\rnote}[1]{\footnote{{\bf [[Rocco: {#1}\bf ]] }}}
\newcommand{\inote}[1]{\footnote{{\bf [[Ilias: {#1}\bf ]] }}}
\newcommand{\anote}[1]{\footnote{{\bf [[Anindya: {#1}\bf ]] }}}

\date{}

\begin{document}

\title{
Inverse problems in approximate uniform generation
}

\author{Anindya De\thanks{{\tt anindya@cs.berkeley.edu}.  Research supported by NSF award CCF-0915929 and Umesh Vazirani's Templeton Foundation Grant 21674.}\\
University of California, Berkeley\\
\and
Ilias Diakonikolas\thanks{{\tt ilias.d@ed.ac.uk}.  Most of this work was done while the author was at UC Berkeley supported by a Simons Postdoctoral Fellowship.}\\
University of Edinburgh\\
\and
Rocco A.\ Servedio\thanks{{\tt rocco@cs.columbia.edu}. Supported by NSF grant CCF-1115703, CCF-0915929.}\\
Columbia University\\
}

\maketitle

\setcounter{page}{0}

\thispagestyle{empty}

~
\vskip -.5in
~

\begin{abstract}

We initiate the study of \emph{inverse} problems in approximate
uniform generation, focusing on uniform generation of satisfying
assignments of various types of Boolean functions.  In such an inverse
problem, the algorithm is 
given uniform random satisfying assignments of an unknown function
$f$ belonging to a class $\C$ of Boolean functions (such as 
linear threshold functions or polynomial-size DNF formulas), and
the goal is to output a probability distribution
$D$ which is $\epsilon$-close, in total variation distance, to the uniform
distribution over $f^{-1}(1)$.  Problems of this sort comprise a
natural type of unsupervised learning problem in which the
unknown distribution to be learned is the 
uniform distribution over satisfying assignments of an unknown
function $f \in \C.$

\medskip

{\bf Positive results:}  
We prove a general positive result establishing sufficient
conditions for efficient inverse approximate uniform generation
for a class $\C$.  We define a new type of algorithm 
called a \emph{densifier} for $\C$, and show (roughly speaking)
how to combine (i) a densifier, (ii) an approximate counting / uniform
generation algorithm, and
(iii) a Statistical Query learning algorithm, to obtain an inverse
approximate uniform generation algorithm.
We apply this general result to obtain a poly$(n,1/\eps)$-time inverse
approximate uniform generation algorithm for the class of $n$-variable
linear threshold functions (halfspaces);
and a quasipoly$(n,1/\eps)$-time inverse approximate
uniform generation algorithm for the class of $\poly(n)$-size DNF 
formulas.

\medskip

{\bf Negative results:}  We prove a general negative result establishing
that the existence of certain types of signature schemes in cryptography 
implies the hardness of certain inverse approximate uniform generation
problems.
We instantiate this negative result with known
signature schemes from the cryptographic literature to prove
(under a plausible cryptographic hardness assumption) that there
are no {subexponential}-time inverse approximate uniform generation algorithms
for 3-CNF formulas; for intersections of two halfspaces; for
degree-2 polynomial threshold functions; and for monotone 2-CNF formulas.

\medskip

Finally, we show that there is no  general relationship between
the complexity of the ``forward'' approximate uniform generation problem and 
the complexity of the inverse problem for a class
$\C$ -- it is possible for either
one to be easy while the other is hard.
In one direction, 
we show that the existence of certain types of 
Message Authentication Codes (MACs) in cryptography implies the hardness
of certain corresponding inverse approximate uniform generation problems, 
and we combine this general result with recent MAC constructions
from the cryptographic literature to show (under a plausible cryptographic 
hardness assumption) that
there is a class $\C$ for which the ``forward'' approximate uniform generation
problem is easy but the inverse approximate uniform generation problem
is computationally hard.  In the other direction, we also show 
(assuming the GRAPH ISOMORPHISM problem is computationally hard) that
there is a problem for which inverse approximate uniform generation
is easy but ``forward'' approximate uniform generation is
computationally hard.

\end{abstract}

\newpage

\section{Introduction}\label{sec:intro}

The generation of (approximately) uniform random combinatorial objects
has been an important research topic in theoretical computer science
for several decades.
In complexity theory, well-known results have related
approximate uniform generation to other fundamental topics
such as approximate counting and the power of nondeterminism
\cite{JS89a,JVV86, JS89b, Sipser83-short, Stockmeyer83-short}.
On the algorithms side, celebrated algorithms have been given for a wide range
of approximate uniform generation problems such as perfect
matchings \cite{JSV04}, graph colorings (see e.g. 
\cite{Jerrum95,Vigoda99,HayesVigoda03}),
satisfying assignments of DNF formulas \cite{KarpLuby83,JVV86,KLM89}, of 
linear threshold functions (i.e., knapsack
instances) \cite{MorrisSinclair:04,Dyer03-short}
and more.

\ignore{The focus of this paper is on \emph{inverse} problems in approximate
uniform generation.}  Before describing the inverse problems that we
consider, let us briefly recall the usual framework of approximate uniform
generation.
An approximate uniform generation problem is defined by a class $\C$ of
combinatorial objects
and a polynomial-time relation $R(x,y)$ over
$\C \times \{0,1\}^*$. An input instance of the
problem is an object $x \in \C,$ and the problem, roughly speaking, is to output
an approximately uniformly random element $y$ from the set $R_x := \{y:
R(x,y)$ holds$\}.$ Thus an algorithm
$A$ (which must be randomized)
for the problem must have the property that for all $x \in \C$,
the output distribution of $A(x)$ 
puts approximately equal weight on every element of $R_x$.
For example, taking the class of combinatorial
objects to be $\{$all $n \times n$ bipartite graphs$\}$
and the polynomial-time relation $R$
over $(G,M)$ pairs to be ``$M$ is a perfect matching in $G$,'' the resulting
approximate uniform generation problem is to generate an (approximately)
uniform perfect matching in a given bipartite graph; a
$\poly(n,{\log(1/\eps)})$-time algorithm was given in \cite{JSV04}.
As another example, taking the combinatorial
object to be a linear threshold function (LTF) 
$f(x)=\sign(w \cdot x - \theta)$
mapping $ \{-1,1\}^n \to \{-1,1\}$ (represented as a vector
$(w_1,\dots,w_n,\theta)$)
and the polynomial-time
relation $R$ over $(f,x)$ to be ``$x$ is a satisfying assignment for $f$,''
we arrive at the problem of generating approximately uniform
satisfying assignments for an LTF (equivalently,
feasible solutions to zero-one knapsack).
A polynomial-time algorithm was given by \cite{MorrisSinclair:04}
and a faster algorithm was subsequently proposed by \cite{Dyer03-short}.

The focus of this paper is on \emph{inverse} problems in approximate uniform
generation.  In such problems,
instead of having to \emph{output} (near-)uniform elements of $R_x$,
the \emph{input} is a sample of elements drawn uniformly
from $R_x$, and the problem (roughly speaking) is
to ``reverse engineer'' the sample and output a distribution which
is close to the uniform distribution over $R_x$.
More precisely, following the above framework,
a problem of this sort is again defined by a class $\C$ of combinatorial
objects and a polynomial-time relation $R$.
However, now an input instance of the problem is a sample $\{y_1,\dots,y_m\}$
of strings drawn uniformly at random from the set $R_x := \{y: R(x,y)$
holds$\}$,  where now $x \in \C$ is \emph{unknown}.
The goal is to output an \emph{$\eps$-sampler} for $R_x$, i.e.,
a randomized algorithm (which takes no input) whose output distribution
is $\eps$-close in total variation distance to the uniform
distribution over $R_x$.  
\ifnum\confversion=0
Revisiting the first example from the previous
paragraph, for the inverse problem the input would be
a sample of uniformly random perfect
matchings of an unknown bipartite graph $G$, and the problem
is to output a sampler for the uniform distribution over
all perfect matchings of $G$.
For the inverse problem corresponding to the second example, the input is
a sample of uniform random satisfying assignments of
an unknown LTF over the Boolean hypercube,
and the desired output is a sampler that
generates approximately uniform random satisfying assignments of the LTF.
\else
For the inverse problem corresponding to the second example
(LTF satisfying assignments) given above, the input is
a sample of uniform random satisfying assignments of
an unknown LTF over the Boolean hypercube,
and the desired output is a sampler that
generates approximately uniform random satisfying assignments of the LTF.
\fi

\ignore{(Later in the paper we will give two efficient algorithms that solve
these two inverse approximate uniform generation problems.)}

{

\medskip
\noindent {\bf Discussion.}  Before proceeding we briefly consider some possible alternate
definitions of inverse approximate uniform generation, and argue that our definition is the
``right'' one
(we give a precise statement of our definition in Section~\ref{sec:prelims}, see Definition~\ref{def:inv-approx-unif-gen}).

One stronger possible notion of inverse approximate uniform generation would be that the output distribution
should be supported on $R_x$ and put nearly the same weight on every element of $R_x$, instead of just being $\eps$-close to uniform over $R_x$.  However a moment's thought suggests that this notion is too
strong, since it is impossible to efficiently achieve this
strong guarantee even in simple settings.  (Consider, for example, the problem of inverse approximate
uniform generation of satisfying assignments for an unknown LTF.  
Given access to uniform satisfying assignments of an LTF $f$, 
it is impossible to efficiently determine whether $f$ is (say) the majority
function or an LTF that differs from majority on precisely one point in $\bits^n$, and thus it
is impossible to meet this strong guarantee.)

Another possible definition of inverse approximate uniform generation would be to require
that the algorithm output
an $\eps$-approximation of
the unknown object $x$ instead of an $\eps$-sampler for $R_x$.
Such a proposed definition, though, leads immediately
to the  question of how one should measure
the distance between a candidate object $x'$ and the true ``target''
object $x$.  The most obvious choice would seem to be the total
variation distance between $\calU_{R_x}$ (the uniform distribution over
$R_x$) and $\calU_{R_{x'}}$; but given this distance measure, it seems
most natural to require that the algorithm actually output an $\eps$-approximate sampler
for $R_{x}$.

}

\medskip

\noindent {\bf Inverse approximate uniform generation via reconstruction
and sampling.}
While our ultimate goal, as described above, is to obtain
algorithms that output a sampler,
algorithms that attempt to reconstruct the unknown object $x$ will
also play an important role for us.  Given $\C,R$ as above, we say that
an \emph{$(\eps,\delta)$-reconstruction algorithm} is an
algorithm $A_\reconstruct$ that works as follows:  for any
$x \in \C$, if $A_\reconstruct$ is given as input
a sample of $m=m(\eps,\delta)$
i.i.d. draws from the uniform distribution over
$R_x$, then with probability $1-\delta$ the output of
$A_\reconstruct$ is an object $\tilde{x} \in \tilde{C}$ such that
the variation distance $\dtv(\calU_{R_x},\calU_{R_{\tilde{x}}})$
is at most $\eps$.
(Note that the class $\tilde{\C}$ need not coincide with the
original class $\tilde{\C}$, so $\tilde{x}$ need not necessarily belong to $\C.$)
With this notion in hand, an intuitively appealing schema for algorithms
that solve inverse approximate uniform generation problems is to proceed in the following two stages:

\begin{enumerate}

\item {\bf (Reconstruct the unknown object):}
Run a reconstruction algorithm $A_\reconstruct$ with
accuracy and confidence parameters $\eps/2,\delta/2$
to obtain $\tilde{x} \in \tilde{\C}$;

\item {\bf (Sample from the reconstructed object):}
Let $A_\sample$ be an algorithm which solves the approximate uniform
generation problem $(\tilde{\C},R)$
to accuracy $\eps/2$ with confidence $1-\delta/2.$  The
desired sampler is the algorithm $A_\sample$ with its input set to $\tilde{x}$.

\end{enumerate}

We refer to this as the \emph{standard approach} for
solving inverse approximate uniform generation problems.
Most of our positive results for inverse approximate
uniform generation {can be viewed as} following this approach,
\ifnum\confversion=0
but we will see an interesting exception in
Section~\ref{sec:graph-auto}, where we give an efficient algorithm for
an inverse approximate uniform generation problem which does
not follow the standard approach.\ignore{ -- and we in fact prove that
\emph{no} algorithm following the standard approach can efficiently
solve the problem (under a plausible hardness assumption).}
\else
but there are exceptions; in the full version we
give an efficient algorithm for
an inverse approximate uniform generation problem which does
not follow the standard approach.
\fi

\subsection{Relation between inverse approximate uniform
generation and other problems.}

Most of our results will deal with uniform generation problems in which the
class $\C$ 
\ifnum\confversion=0
of combinatorial objects 
\fi
is a class of syntactically
defined Boolean functions
over $\{-1,1\}^n$
(such as the class of all LTFs, all $\poly(n)$-term DNF
formulas, all 3-CNFs, etc.) and the polynomial-time relation $R(f,y)$
for $f \in \C$ is ``$y$ is a satisfying assignment for $f$.''
In such cases our inverse approximate uniform generation
problem can be naturally recast in the language of learning theory as an
unsupervised learning problem (learning a probability distribution
from a known class of possible target distributions):  we are
given access to samples from $\calU_{f^{-1}(1)}$, the uniform distribution
over satisfying assignments of $f \in \C$, and
the task of the learner is to construct a hypothesis
distribution $D$ such that $\dtv(\calU_{f^{-1}(1)},D) \leq \eps$
with high probability.  We are not aware of prior work in unsupervised
learning that focuses specifically on distribution learning problems 
\ifnum\confversion=0
of this sort (where
the target distribution is uniform over the set of satisfying assignments
of an unknown member of a known class of Boolean functions).
\else
of this sort.
\fi

Our framework also has some similarities to ``uniform-distribution
learning from positive examples only,'' since in both settings the input
to the algorithm is a sample of points drawn uniformly at random
from $f^{-1}(1)$, but there are several
differences as well. One difference is that in uniform-distribution
learning from positive examples the goal is to output a hypothesis
function $h$, whereas here our goal is to output a hypothesis
\emph{distribution} (note that outputting a function $h$ essentially
corresponds to the reconstruction problem described above).
A more significant difference is that the success
criterion for our framework is significantly more demanding than
for uniform-distribution learning.
In uniform-distribution learning of a Boolean function
$f$ over the hypercube $\{-1,1\}^n$, the hypothesis $h$ must satisfy
$\Pr[h(x) \neq f(x)] \leq \eps$, where the probability is uniform
over all $2^n$ points in $\{-1,1\}^n.$  Thus, for a given setting of the
error parameter $\eps$, in uniform-distribution learning the constant $-1$
function is an acceptable hypothesis for any function $f$ that has
$|f^{-1}(1)| \leq \eps 2^n.$  In contrast, in our inverse approximate
uniform generation framework we measure error by the
total variation distance between $\calU_{f^{-1}(1)}$ and the
hypothesis distribution $D$, so no such ``easy way out''
is possible when $|f^{-1}(1)|$ is small;
indeed the hardest instances of inverse approximate uniform
generation problems are often those for which $f^{-1}(1)$ is a very
small fraction of $\{-1,1\}^n.$
Essentially we require a hypothesis with small
\emph{multiplicative} error relative to $|f^{-1}(1)|/2^n$ rather than
the additive-error criterion that is standard in uniform-distribution
learning.
We are not aware of prior work on learning Boolean functions in which such a
``multiplicative-error'' criterion has been employed.


\ifnum\confversion=0
We summarize the above discussion with the following
observation, which essentially says that reconstruction algorithms
directly yield uniform-distribution learning algorithms:

\begin{obs} \label{obs:reconstruct-and-unif-dist-learn}
Let $\C$ be a class of Boolean functions $\{-1,1\}^n \to \{-1,1\}$
and let $R(f,y)$ be the relation ``$y$ is a satisfying assignment
for $f$.''  Suppose there exists a $t(n,\eps,\delta)$-time
$(\eps,\delta)$-reconstruction algorithm for $\C$ that outputs
elements of $\tilde{\C}$. Then there is an $(O(\log(1/\delta)/\eps^2) + O(t(n,\eps,\delta/3)\cdot\log(1/\delta)/\eps))$-time
uniform-distribution learning algorithm that outputs
hypotheses in $\tilde{\C}$ (i.e., given access to uniform
random labeled examples $(x,f(x))$ for any $f \in \C$,
the algorithm with probability $1-\delta$ outputs a hypothesis $h \in \tilde{C}$
such that $\Pr[h(x) \neq f(x)] \leq \eps$).
\end{obs}
\begin{proof}
The learning algorithm draws an initial set of $O(\log(1/\delta)/\eps^2)$
uniform labeled examples to estimate $|f^{-1}(1)|/2^n$ to within an
additive $\pm (\eps/4)$ with confidence $1-\delta/3$.  If the estimate is
less than $3\eps/4$ the algorithm outputs the constant $-1$ hypothesis.
Otherwise, by drawing $O(t(n,\eps,\delta/3)\cdot\log(1/\delta)/\eps))$
uniform labeled examples, with failure probability at most $\delta/3$
it can obtain  $t(n,\eps,\delta/3)$ positive examples (i.e.,
points that are uniformly distributed over $f^{-1}(1)$).
Finally the learning algorithm can use these points to run the reconstruction
algorithm with parameters $\eps,\delta/3$ to obtain a
hypothesis $h \in \tilde{\C}$ that has $\dtv(\calU_{f^{-1}(1)},
\calU_{h^{-1}(1)}) \leq \eps$ with failure probability at most $\delta/3$.
Such a hypothesis $h$ is easily seen to satisfy
$\Pr[h(x) \neq f(x)] \leq \eps.$
\end{proof}
\fi

As described in the following subsection, in this paper we prove
negative results for the inverse approximate uniform generation problem
for classes such as 3CNF-formulas, monotone 2-CNF formulas, and degree-2 polynomial threshold
functions.  Since efficient uniform-distribution learning algorithms
are known for these classes, these results show that
the inverse approximate uniform generation problem is indeed harder than
standard uniform-distribution learning for some natural and interesting
classes of functions.

{ The problem of inverse approximate uniform generation is also 
somewhat reminiscent of the problem of reconstructing Markov Random Fields 
(MRFs) 
from random samples~\cite{BMS08, DMR06-short, Mos07}.   Much progress has been made on
this problem over the past decade, especially when the hidden graph is a tree. 
However, there does not seem to be a concrete connection between this
problem and the problems we study.  
One reason for this seems to be that in MRF 
reconstruction, the task is to reconstruct the \emph{model} and not just the 
distribution; because of this, various conditions need to be imposed 
in order to guarantee the uniqueness of the 
underlying model given random samples from the distribution.  
{In contrast, in our setting the explicit goal is to construct
a high-accuracy distribution, and it may indeed be the case that there
is no unique underlying model (i.e., Boolean function $f$) given
the samples received from the distribution.}
}

\subsection{Our results.} \label{ssec:results}

We give a wide range of both positive and negative results for
inverse approximate uniform generation problems.  As noted
above, most of our results deal with uniform
generation of satisfying assignments, i.e., $\C$ is a class of Boolean functions
over $\{-1,1\}^n$ and for $f \in \C$ the relation $R(f,y)$ is ``$y$
is a satisfying assignment for $f$.''
All the results, both positive and negative, that we present below are
for problems of this sort unless indicated otherwise.

\medskip

\noindent
{\bf Positive results:  A general approach and its applications.}
We begin by presenting a general approach for
obtaining inverse approximate uniform generation algorithms.  
{This technique
combines approximate uniform generation and counting algorithms
and Statistical Query ($\sq$) learning algorithms with a new type of algorithm
called a ``densifier,'' which we introduce and define in Section~\ref{sec:generaltechnique}.
\ifnum\confversion=0
Very roughly speaking, the densifier lets us prune the entire space $\{-1,1\}^n$ to
\else
The densifier lets us prune the entire space $\{-1,1\}^n$ to
\fi
a set $S$ which (essentially) contains all of $f^{-1}(1)$ and is not 
too much larger than $f^{-1}(1)$ (so $f^{-1}(1)$ is ``dense'' in $S$).  
By generating approximately uniform elements of $S$ it is possible to run
an $\sq$ learning algorithm and obtain a high-accuracy hypothesis which can be used, in conjunction with
the approximate uniform generator, to obtain a sampler for a distribution which is close to the uniform
distribution over $f^{-1}(1).$  (The approximate counting algorithm is needed 
for technical reasons
which we explain in Section~\ref{ssec:intuition-motivation-discussion}.)}
In Section~\ref{sec:generaltechnique} we describe this technique in
detail and prove a general result establishing its effectiveness.

\ifnum\confversion=0
In Sections~\ref{sec:LTF} and~\ref{sec:DNF} we give {two main} applications of this
general technique to specific classes of functions.
The first of these is the class $\ltf$ of all LTFs over $\{-1,1\}^n.$
\else
We give two main applications of this
general technique to specific classes of functions.
The first of these (see Section~\ref{sec:LTF}) 
is the class $\ltf$ of all LTFs over $\{-1,1\}^n.$
\fi
Our main technical contribution here is to construct a densifier for 
LTFs;
we do this by carefully combining known efficient online learning algorithms
for LTFs (based on interior-point methods for linear programming) 
\cite{MT:94} with
known algorithms for approximate uniform generation 
{and counting} of satisfying
assignments of LTFs \cite{MorrisSinclair:04,Dyer03-short}.
Given this densifier, our general approach yields the desired inverse 
approximate uniform generator for LTFs:

\ignore{
The hypotheses $h$ that our densifier outputs are not themselves halfspaces;
rather, they are intersections of two halfspaces.  However we show that it
is nonetheless possible to use known algorithms
\cite{MorrisSinclair:04,Dyer03-short} for sampling uniform random satisfying
assignments of a single halfspace to sample near-uniform random
satisfying assignments of $h$, and thu we can carry out Stage~2 of
our ``standard approach'' as well.  We thus obtain our first main positive
result for a specific problem: \mnote{\new{Check, improve exposition}}
}

\begin{theorem} {\bf (Informal statement)} \label{thm:ltf-informal}
There is a $\poly(n,1/\eps)$-time algorithm for the inverse problem
of approximately uniformly generating satisfying assignments for LTFs.
\end{theorem}

\ifnum\confversion=0
Our second main positive result for a specific class, in Section~\ref{sec:DNF},
is for the well-studied class $\dnfns$ of all size-$s$ DNF formulas over $n$
Boolean variables.  
\else
Our second main positive result for a specific class (see the full version)
is for the well-studied class $\dnfns$ of all size-$s$ DNF formulas over $n$
Boolean variables.  
\fi
Here our main technical contribution 
is to give a densifier which runs in time $n^{O(\log(s/\eps))}$ and outputs
a DNF formula.  {A challenge here is that known $\sq$ algorithms for learning
DNF formulas require time exponential in $n^{1/3}$.  To get around this, we view the densifier's output DNF as an
OR over $n^{O(\log(s/\eps))}$ ``metavariables'' (corresponding to all possible conjunctions that could
be present in the DNF output by the densifier), and we show that it is possible to apply known \emph{malicious 
noise tolerant} $\sq$ algorithms for learning
\emph{sparse disjunctions} as the $\sq$-learning component of our general approach.}
\ignore{and we thus obtain an $n^{O(\log(s/\eps))}$-time
reconstruction algorithm (using DNF hypotheses) for the class of $s$-term DNF.
We note that by Observation~\ref{obs:reconstruct-and-unif-dist-learn},
since $n^{O(\log(s/\eps))}$ is the fastest known runtime for learning
$s$-term DNF formulas under the uniform distribution\cite{ver90}, improving the
running time of our reconstruction algorithm would require improving
the state of the art for uniform-distribution DNF learning,
which is a major open question in learning theory.\footnote{Recall that the
$\poly(n,s,1/\eps)$-time algorithm of Jackson \cite{Jackson:97} uses
membership queries, not just uniform random examples.} }
Since efficient {approximate uniform generation and approximate counting} 
algorithms are known \cite{JVV86,KarpLuby83} for
DNF formulas, with the above densifier and $\sq$ learner we can carry out
our general technique, and we thereby obtain
our second main positive result for a specific function class:

\begin{theorem} {\bf (Informal statement)}
There is a $n^{O(\log(s/\eps))}$-time algorithm for the inverse problem
of approximately uniformly generating satisfying assignments for
$s$-term DNF formulas.
\end{theorem}

\noindent
{\bf Negative results based on cryptography.}
In light of the ``standard approach,''
it is clear that in order for an inverse
approximate uniform generation problem $(\C,R)$ to be computationally
hard, it must be the case that either stage (1) (reconstructing
the unknown object) or stage (2) (sampling from the reconstructed object)
is hard.  
\ifnum\confversion=0
(If both stages have efficient algorithms $A_\reconstruct$
and $A_\sample$ respectively, then there is an efficient algorithm for the
whole inverse approximate uniform generation problem that combines
these algorithms according to the standard approach.)
\fi
Our first approach to obtaining negative results can be used to obtain
hardness results for problems for which stage (2), near-uniform sampling,
is computationally hard.  The approach is based on signature
schemes from public-key cryptography; roughly speaking, the general result
which we prove is the following (we note that the statement given below is a simplification of
our actual result which omits several technical conditions; 
\ifnum\confversion=0
see Theorem~\ref{thm:signatures} of Section~\ref{sec:signature}
for a precise statement):
\else
see the full version for a precise statement):
\fi

\begin{theorem} \label{thm:hard-sig} {\bf (Informal statement)}
\ignore{Fix $\eps = \Theta(1).$}  Let $\C$ be a class of functions
such that there is a parsimonious reduction from \textsf{CIRCUIT-SAT} 
to $\C$\textsf{-SAT}.
Then known constructions of secure signature schemes imply
that there is no {subexponential}-time algorithm
for the inverse problem of approximately uniformly generating
satisfying assignments to functions in $\C.$
\end{theorem}

This theorem yields a wide range of hardness results for specific classes 
that show that our positive results (for LTFs and DNF) lie quite close to
the boundary of what classes have efficient
inverse approximate uniform generation algorithms.  We prove:

\begin{corollary} {\bf (Informal statement)} \label{cor:first}
Under known constructions of secure signature schemes,
there is no {subexponential}-time algorithm for the
inverse approximate uniform generation problem for 
either of the following classes of functions:
(i) 3-CNF formulas; (ii) intersections of two halfspaces.

\end{corollary}

We show that our signature-scheme-based hardness approach can be
extended to settings where there is no parsimonious reduction
as described above.  Using ``blow-up''-type constructions of the sort
used to prove hardness of approximate counting, we prove the following:

\begin{theorem} {\bf (Informal statement)}
Under the same assumptions as Corollary~\ref{cor:first},
there is no {subexponential}-time algorithm for the
inverse approximate uniform generation problem for either of the following classes:  
(i) monotone 2-CNF; (ii) degree-2 polynomial threshold functions.
\end{theorem}

{It is instructive to compare the above hardness results with the 
problem of uniform generation of NP-witnesses. In particular, while it 
is obvious  that no
efficient randomized algorithm can produce even a single satisfying 
assignment of a given \textsf{3-SAT} instance 
(assuming $\mathbf{NP} \not \subseteq \mathbf{BPP}$),
the seminal results of 
Jerrum \etal ~\cite{JVV86} showed that given access to an NP-oracle, it 
is possible to generate approximately uniform satisfying assignments for 
a given \textsf{3-SAT} instance.  It is interesting to ask whether
one requires the full power of adaptive access to NP-oracles for this 
task, or whether a weaker form of ``advice" suffices. 
Our hardness results can be understood in this context as 
giving evidence that receiving polynomially many random 
satisfying assignments of a \textsf{3-SAT} instance does not help in 
further uniform generation of satisfying assignments.\footnote{There is a 
small caveat here in that we are not given the \textsf{3-SAT} formula 
\emph{per se} but rather access to random satisfying assignments of the 
formula. However, there is a simple elimination based algorithm to 
reconstruct {a high-accuracy approximation for} a 
\textsf{3-SAT} formula if we have access to random 
satisfying assignments for the formula.}
}

Our signature-scheme based approach cannot give hardness results for 
problems that have polynomial-time algorithms for the ``forward'' 
problem of sampling approximately uniform satisfying assignments. Our 
second approach to proving computational hardness can (at least 
sometimes) surmount this barrier.  The approach is based on Message 
Authentication Codes in cryptography; the following is an informal 
statement of our general result along these lines (as before the 
following statement ignores some technical conditions; see 
\ifnum\confversion=0
Theorem~\ref{thm:hardness-MAC} for a precise statement):
\else
the full version for a precise statement):
\fi

\begin{theorem} \label{thm:hard-MAC-informal} {\bf (Informal statement)} 
There are known constructions of MACs with the following property: Let 
$\C$ be a class of circuits such that the verification algorithm of the 
MAC can be implemented in $\C$. Then there is no {subexponential}-time 
inverse approximate uniform generation algorithm for $\C$. \end{theorem}

We instantiate this general result with a specific construction of a MAC 
that is a slight variant of a construction due to Pietrzak \cite{Pie:12}.
This specific construction yields a class $\C$ for which the ``forward'' 
approximate uniform generation problem is computationally easy, but 
(under a plausible computational hardness assumption) the inverse 
approximate uniform generation problem is computationally hard.

The above construction based on MACs shows that there are problems 
$(\C,R)$ for which the inverse approximate uniform generation problem is 
computationally hard although the ``forward'' approximate uniform 
generation problem is easy.  
\ifnum\confversion=1
As our last result, in the full version we
\else
As our last result, we
\fi
exhibit a group-theoretic problem (based on graph automorphisms) for 
which the reverse situation holds: under a plausible hardness assumption 
the \emph{forward} approximate uniform generation problem is 
computationally hard, but we give an efficient algorithm for the 
\emph{inverse} approximate uniform generation problem (which does not 
follow our general technique or the ``standard approach'').

\medskip

\ifnum\confversion=1
\noindent {\bf Organization.}  Because of space constraints here we
only present our general approach and its application to LTFs;
all other results are in the full version.
\else
\bigskip
\noindent {\bf Structure of this paper.} After the preliminaries in Section~\ref{sec:prelims}, we present in Section~\ref{sec:generaltechnique}
our general upper  bound technique. In Sections~\ref{sec:LTF} and~\ref{sec:DNF} we apply this technique to obtain efficient inverse approximate 
uniform generation algorithms for LTFs and DNFs respectively. Section~\ref{sec:hardness} contains our hardness results. In Section~\ref{sec:graph-auto} we give an example
of a problem for which approximate uniform generation is hard, while the inverse problem is easy. Finally, in Section~\ref{sec:conc} we conclude
the paper suggesting further directions for future work.
\fi

\section{Preliminaries and Useful Tools}
\label{sec:prelims}

\ifnum\confversion=0
\subsection{Notation and Definitions.} \label{ssec:defs}
For $n \in \Z_{+}$, we will denote by $[n]$ the set $\{1, \ldots, n\}$. 
For a distribution $D$ over a finite set $\W$ we denote by
$D(x)$, $x \in \W$, the probability mass that $D$ assigns to point $x$,
so $D(x) \geq 0$ and $\littlesum_{x \in \W} D(x) = 1$. For $S \subseteq \W$,
we write $D(S)$ to denote $\littlesum_{x \in S} D(x)$.
For a finite set $X$ we write $x \in_U X$ to indicate that $x$ is
chosen uniformly at random from $X.$
\fi
\ifnum\confversion=1
\noindent {\bf Notation and definitions.}  
\fi 
For a random variable $x$,
we will write $x \sim D$ to denote that $x$ follows distribution $D$.
Let $D, D'$ be distributions over $\W$.
\ifnum\confversion=1
The {\em total variation distance} between $D$ and $D'$ is
$\dtv (D,D') \eqdef \max_{S \subseteq \W} \left| D(S) - D'(S) \right|$.
\fi
\ifnum\confversion=0
The {\em total variation distance} between $D$ and $D'$ is
$\dtv (D,D') \eqdef \max_{S \subseteq \W} \left| D(S) - D'(S) \right|  = (1/2) \cdot \|D-D'\|_1$,
where $\| D - D' \|_1 = \littlesum_{x \in \W} \left| D(x) - D'(x) \right|$ is the $L_1$--distance between $D$ and $D'$.
\fi

 We will denote by $\C_n$, or simply $\C$, a Boolean concept class,
i.e., a class of functions mapping $\bn$ to $\bits$. {We usually consider
syntactically defined classes of functions such as the class of all 
$n$-variable linear threshold functions or the class of all $n$-variable 
$s$-term DNF formulas. We stress that throughout this paper a class $\C$ is 
viewed as a {\em representation class}.  Thus we will say that an 
algorithm ``takes as input a function $f \in \C$'' to mean that the input 
of the algorithm is a {\em representation} of $f \in \C$.}

We will use the notation $\U_n$ (or simply $\U$, when the dimension $n$ is clear from the context)
for the uniform distribution over $\bn$. Let $f: \bn \to \bits$. We will denote by
$\U_{f^{-1}(1)}$ the uniform distribution over satisfying assignments of $f$.
Let $D$ be  a distribution over $\bn$ with $0 < D(f^{-1}(1)) < 1.$
We write $D_{f,+}$ to denote the conditional distribution $D$
restricted to $f^{-1}(1)$; so for $x \in f^{-1}(1)$ we have
$D_{f,+}(x) = D(x)/D(f^{-1}(1)).$  
\ignore{We similarly write $D_{f,-}$ to denote the conditional
distribution $D$ restricted to $f^{-1}(-1).$} Observe that, with this notation, we have that
$\U_{f^{-1}(1)} \equiv \U_{f, +}$.

\ifnum\confversion=1
We use familiar notions of approximate counting and approximate
uniform generation (see the full version for precise definitions).
\else
We proceed to define the notions of approximate counting and
approximate uniform generation for a class
of Boolean functions:

\begin{definition}[approximate counting]
\label{def:approx-count}
Let $\C$ be a class of $n$-variable Boolean functions. A randomized algorithm $\A^\C_{\co}$ is an {\em efficient approximate
counting algorithm for class $\C$}, if for
any $\eps, \delta>0$ and any $f \in \C$, on input $\eps, \delta$ and {$f \in \C$},
it runs in time $\poly(n, 1/\eps,{\log(1/\delta)})$ and with probability $1-\delta$ outputs a value $\widehat{p}$ such that
\[  {1 \over (1+\eps)} \cdot \Pr_{x \sim \U} [f(x) = 1] \le   \widehat{p} \le (1+\eps) \cdot \Pr_{x \sim \U} [f(x) = 1].\]
\end{definition}

\begin{definition}[approximate uniform generation]
\label{def:approx-unif-gen}
Let $\C$ be a class of $n$-variable Boolean functions.
A randomized algorithm $\A^\C_{\gen}$ is an {\em efficient approximate
uniform generation algorithm for class $\C$}, if for any $\eps>0$
and any $f \in \C$, there is a distribution {
$D=D_{f,\eps}$ supported on $f^{-1}(1)$
with
\[{\frac 1 {1+\eps}} \cdot {\frac 1 {|f^{-1}(1)|}} \leq
D(x) \leq (1+\eps) \cdot
{\frac 1 {|f^{-1}(1)|}}
\]
for each $x \in f^{-1}(1)$,
such that for any $\delta>0$,
on input $\eps,\delta$ and {$f \in \C$}, algorithm $A^\C_{\gen}(\eps,\delta,f)$
runs in time $\poly(n, 1/\eps,\log(1/\delta))$ and
either outputs a point $x \in f^{-1}(1)$ that is distributed precisely
according to $D=D_{f,\eps}$, or outputs $\bot$.
Moreover the probability that it outputs $\bot$ is at most $\delta.$
}
\end{definition}

\ignore{

OLD DEFINITION THAT JUST GUARANTEES d_{TV}:

\begin{definition}[approximate uniform generation]
Let $\C$ be a class of $n$-variable Boolean functions. A randomized algorithm $\A^\C_{\gen}$ is an {\em efficient approximate
uniform generation algorithm for class $\C$}, if for any $\eps>0$
and any $f \in \C$, there is a distribution
$D=D_{f,\eps}$ supported on $f^{-1}(1)$
such that $\dtv(D,\U_{f^{-1}(1)}) \leq \eps$ and
on input $\eps$ and {a circuit computing} $f$, algorithm $A^\C_{\gen}$ runs in time
$\poly(n, 1/\eps)$ and outputs $x \sim D.$
\end{definition}

}

{
An approximate uniform generation algorithm is said to be
\emph{fully polynomial} if its running time dependence on $\eps$
is $\poly(\log(1/\eps)).$}

\ignore{
on input $\eps, \delta$ and
$f$, it runs in time $\poly(n, 1/\eps, 1/\delta)$ and outputs a satisfying assignment $x^{\ast}$ of $f$, i.e., $x^{\star}$ is supported  in $f^{-1}(1)$,
that with probability $1-\delta$ satisfies the following: $x^{\ast} \sim D$, where $\dtv (D, \U_{f^{-1}(1)}) \le \eps.$
}

\fi
Before we define our inverse 
approximate uniform generation problem, 
we need the notion of a {\em sampler} for a distribution:

\begin{definition} \label{def:sampler}
Let $D$ be a distribution over $\bn$.
A \emph{sampler for $D$} is a circuit $C$
with $m=\poly(n)$ input bits $z \in \bits^m$ and
$n$ output bits $x \in \bn$
which is such that when $z \sim \U_m$ 
then  $x \sim D$. 
For $\eps > 0$, an \emph{$\eps$-sampler for $D$} is a sampler
for some distribution $D'$ which has $\dtv(D',D) \leq \eps.$
\end{definition}

For clarity we sometimes write ``$C$ is a $0$-sampler for $D$'' to emphasize the
fact that the outputs of $C(z)$ are distributed {\em exactly} according to
distribution $D.$
We are now ready to formally define the notion of an inverse approximate
uniform generation algorithm:

\begin{definition}[inverse approximate uniform generation] \label{def:inv-approx-unif-gen} 
Let $\C$ be a class of $n$-variable Boolean functions. A randomized algorithm $\A^\C_{\inv}$ is an {\em inverse approximate
uniform generation algorithm for class $\C$}, if for any $\eps, \delta>0$ and any $f \in \C$, on input $\eps, \delta$ and sample access to
$\U_{f^{-1}(1)}$,   with probability $1-\delta$ algorithm $A^\C_{\inv}$ outputs an $\eps$-sampler $C_f$ for $\U_{f^{-1}(1)}$. 
\ignore{
\inote{Do we want to keep this last sentence?}We say that such an algorithm is {\em efficient} if it runs in time $\poly(n, 1/\eps, 1/\delta)$.
}
\end{definition}

\ignore{
{We note that there does not seem to be a useful notion of fully
polynomial inverse approximate uniform generation, since (as is easily
seen) even in very simple cases at least $\Omega(1/\eps)$ samples
are required.}
}

\medskip

\ifnum\confversion=1
\noindent {\bf Hypothesis testing.}
\fi
\ifnum\confversion=0
\subsection{Hypothesis Testing.} \label{ssec:ht}
\fi
Our general approach works 
by generating a collection of hypothesis distributions, one of
which is close to the target distribution $\U_{f^{-1}(1)}$.  Thus, we need a way to select a high-accuracy
hypothesis distribution from a pool of candidate distributions which contains at least one
high-accuracy hypothesis.  This problem has been well studied, see e.g.
Chapter~7 of \cite{DL:01}.  We use the following result which is an extension of Lemma~C.1 of~\cite{DDS12soda}.

\begin{proposition} \label{prop:log-cover-size}
Let $D$ be a distribution over a finite set $\W$
and $\D_\eps = \{ D_j\}_{j=1}^N$ be a collection of $N$ distributions over $\W$
with the property that there exists $i \in [N]$ such that $\dtv(D,D_i) \leq \eps$.
There is an algorithm~${\cal T}^D$, which is given access to:
\begin{enumerate}
\item[(i)] {samplers for} $D$ and $D_k$, for all $k \in [N]$,
\item[(ii)] a {$(1+\beta)$--approximate} evaluation oracle $\eval_{D_k}{(\beta)}$, for all $k \in [N]$,
which, on input $w \in \W$, {deterministically} outputs
{a value $\widetilde{D}^{\beta}_k(w)$, such that $D_k(w) / (1+\beta) \le \widetilde{D}^{\beta}_k(w) \le (1+\beta) D_k (w)$, where
$\beta>0$ is any parameter satisfying $(1+\beta)^2 \le 1+\eps/8$,}
\end{enumerate}
an accuracy parameter $\eps$ and a confidence parameter $\delta$,
and has the following behavior:
It makes 
\[m = O\left( (1/ \eps^{2}) \cdot (\log N + \log(1/\delta)) \right)
\] draws from $D$ and {from} each $D_k$, $k \in [N]$,
and $O(m)$ calls to each oracle $\eval_{D_k}{(\beta)}$, $k \in [N]$, performs $O(m N^2)$ arithmetic operations,
and with probability $1-\delta$ outputs an index $i^{\star} \in [N]$ that satisfies $\dtv(D,D_{i^{\star}}) \leq 6\eps.$
\end{proposition}

\ifnum\confversion=0
\noindent Before we proceed with the proof, we 
\else
We
\fi
note that there are certain crucial differences between the current setting and the setting of~\cite{DDS12soda, DDS12stoc} (as well as other related works that use versions of Proposition~\ref{prop:log-cover-size}). In particular, in our setting, {the set
${\cal W}$ is of size $2^n$, which was not the case in \cite{DDS12soda,DDS12stoc}.}\ignore{the support of the distributions involved is exponential in the input size -- as opposed to polynomial, as was the case
in~\cite{DDS12soda, DDS12stoc}.} Hence, we cannot assume the distributions $D_i$ are given explicitly in the input.
Thus Proposition~\ref{prop:log-cover-size} carefully specifies what kind of access
to these distributions {is required.}\ignore{we have and what kind of operations we need to perform.}
Proposition~\ref{prop:log-cover-size} is an extension of similar results in
the previous works; while the idea of the proof
is essentially the same, the details 
\ifnum\confversion=1
(given in the full version)
\fi
are more involved.

\ifnum\confversion=0

\begin{proof}[Proof of Proposition~\ref{prop:log-cover-size}]
At a high level, the algorithm $\mathcal{T}^D$ performs a tournament by running a ``competition''
{\tt Choose-Hypothesis}$^D$
for every pair of distinct distributions in the collection ${\cal D}_{\eps}$.
It outputs a distribution $D^{\star} \in {\cal D}_\eps$
that was never a loser (i.e., won or achieved a draw in all its competitions).
If no such distribution exists in ${\cal D}_{\eps}$ then the algorithm outputs ``failure.''
We start by describing and analyzing the competition subroutine between a pair of distributions in the collection.

\begin{lemma} \label{lem:choosehypothesis}
In the context of Proposition~\ref{prop:log-cover-size},
there is an algorithm {\tt Choose-Hypothesis}$^D({D}_i,{D}_j,\eps',\delta')$ which is given access to
\begin{enumerate}
\item[(i)] independent samples from $D$ and $D_k$, for $k \in \{ i, j \}$,
\item[(ii)] an evaluation oracle $\eval_{D_k}{(\beta)}$, for $k \in \{ i, j \}$,
\end{enumerate}
an accuracy parameter $\eps'$ and a confidence parameter $\delta'$,
and has the following behavior:
It uses $m'=O\left( (1/{\eps'}^2) \log(1/\delta') \right)$ samples from each of $D$, $D_i$ and $D_j$, it makes
$O(m')$ calls to the oracles $\eval_{D_k}{(\beta)}$, $k \in \{ i, j \}$, performs $O(m')$ arithmetic operations,
and if some $D_k$, $k \in \{ i, j\}$, has $\dtv(D_k, D) \leq \eps'$ then
with probability $1-\delta'$ it
outputs an index $k^{\star} \in \{i, j\} $ that
satisfies $\dtv(D,D_{k^{\star}}) \leq 6\eps'.$
\end{lemma}

\ignore {
\inote{Regarding the algorithm's running time: If we want to be all honest here, the running time is certainly not what I wrote above for the following reasons:
First, in step 1 we need to subtract these two numbers that are explicitly given to us in the input; so, the bit operations are linear in their bit complexities.
Second, to decide whether a sample belongs to the set $H_{ij}$ we similarly need to compare two probability values that are given to us by the evaluation oracles,
however we do not know a priori how many bits they have..they could even be irrational. Now, of course I could include all these things to the bit complexity of the algorithm.
However, it seems that in the approximate version of the theorem, things will be a bit simpler..indeed, in this case we will only have approximations for the $p_{i,j}, q_{i,j}$
up t an additive $\eps'/100$ or sth and these approximations should suffice. A similar thing should hold for the computation of $\tau$, since the evaluation oracles will be approximate, but at this point I'm not totally clear on that.}
}

\begin{proof}
To set up the competition between $D_i$ and $D_j$, we consider the following subset of ${\cal W}$:
\begin{eqnarray*}
H_{ij} = H_{ij}(D_i, D_j) \eqdef \{ w \in \W \mid D_i(w) \ge D_j(w) \}
\end{eqnarray*}
and the corresponding probabilities $p_{i,j} \eqdef D_i (H_{ij})$ and $q_{i,j} \eqdef D_j(H_{ij})$. Clearly, it holds $p_{i, j} \ge q_{i, j}$ and
by definition of the total variation distance we can write $$\dtv(D_i, D_j) = p_{i,j}-q_{i,j}.$$

\nnew{For the purposes of our algorithm, we would ideally want oracle access to the set $H_{ij}$.
Unfortunately though, this is not possible since the evaluation oracles are only approximate.  Hence, we will need to define a more robust version
of the set $H_{ij}$ which will turn out to have similar properties. In particular, we consider the set
\begin{eqnarray*}
H^{\beta}_{ij}  \eqdef \{ w \in \W \mid \widetilde{D}^{\beta}_i(w) \ge \widetilde{D}^{\beta}_j(w) \}
\end{eqnarray*}
and the corresponding probabilities $p^{\beta}_{i,j} \eqdef D_i (H^{\beta}_{ij})$ and $q^{\beta}_{i,j} \eqdef D_j(H^{\beta}_{ij})$.
We claim that the difference $\Delta \eqdef p^{\beta}_{i,j} - q^{\beta}_{i,j}$ is an accurate approximation to $\dtv(D_i, D_j)$. In particular, we show:
\begin{claim} \label{claim:dtv-approx}
We have
\begin{equation} \label{eqn:dtv-approx}
\Delta    \le \dtv(D_i, D_j) \le \Delta +\eps/4.
\end{equation}
\end{claim}
Before we proceed with the proof, we stress that (\ref{eqn:dtv-approx}) crucially
uses our assumption that the evaluation oracles provide a {\em multiplicative} approximation
to the exact probabilities.
\begin{proof}
To show (\ref{eqn:dtv-approx}) we proceed as follows: Let $A = H_{ij} \cap H^{\beta}_{ij}$,
$B = H_{ij} \cap \overline{H^{\beta}_{ij}}$ and $C = \overline{H_{ij}} \cap H^{\beta}_{ij}$. Then we can write
\[  \dtv(D_i, D_j) = (D_i-D_j)(A) + (D_i-D_j)(B)  \] and
\[  \Delta = (D_i-D_j)(A) + (D_i-D_j)(C).\]
We will show that
\begin{equation} \label{eqn:one}
0 \le (D_i-D_j)(B) \le \eps/8
\end{equation}
and  similarly
\begin{equation} \label{eqn:two}
-\eps/8 \le (D_i-D_j)(C) \le 0
\end{equation}
from which the claim follows.
We proceed to prove (\ref{eqn:one}), the proof of (\ref{eqn:two}) being very similar.
Let $w \in B$. Then $D_i(w) \ge D_j(w)$ (since $w \in H_{ij}$)
which gives $(D_i-D_j)(B) \ge 0$, establishing the LHS of  (\ref{eqn:one}).
We now establish the RHS. For $w \in B$ we also have that $ \widetilde{D}^{\beta}_i(w) < \widetilde{D}^{\beta}_j(w)$
(since $w \in  \overline{H^{\beta}_{ij}}$). Now by the definition of the evaluation oracles, it follows that
$\widetilde{D}^{\beta}_i(w) \ge \frac{D_i(w)}{(1+\beta)}$ and $\widetilde{D}^{\beta}_j(w) \le (1+\beta) D_j(w)$.
Combining these inequalities yields $$D_i(w) \le (1+\beta)^2 D_j(w) \le (1+\eps/8)D_j(w)$$
where the second inequality follows by our choice of $\beta$. Therefore,
$$(D_i - D_j)(B) = \littlesum_{w \in B} \left( D_i(w) - D_j(w) \right) \le (\eps/8) \cdot D_j(B) \le \eps/8$$
as desired.
\end{proof}
}

\noindent \nnew{Note that the probabilities $p^{\beta}_{i,j}$ and $q^{\beta}_{i,j}$ are not available to us explicitly.} Hence,
{\tt Choose-Hypothesis} requires a way to empirically estimate each of these \new{probability} values
(up to a small additive accuracy). This task can be done efficiently because we have sample access to the distributions $D_i, D_j$ and oracle access to
the set \nnew{$H^{\beta}_{ij}$}
\new{thanks to the $\eval_{D_k}\nnew{(\beta)}$ oracles.} The following claim provides the details:
\begin{claim} \label{claim:sample-pij}
There exists a subroutine  {\tt Estimate}$({D}_i, \nnew{H^{\beta}_{ij}},  \gamma, \delta)$ which is given access to
\begin{enumerate}
\item[(i)] independent samples from $D_i$,
\item[(ii)] an evaluation oracle $\eval_{D_k}\nnew{(\beta)}$, for $k \in \{ i, j \}$,
\end{enumerate}
an accuracy parameter $\gamma$ and a confidence parameter $\delta$,
and has the following behavior:
It makes $m=O\left( (1/{\gamma}^2) \log(1/\delta) \right)$ draws from $D_i$
and $O(m)$ calls to the oracles $\eval_{D_k}\nnew{(\beta)}$, $k=i, j$, performs $O(m)$ arithmetic operations,
and with probability  $1-\delta$ outputs a number $\widetilde{p}^{\beta}_{i,j}$ such that $|\widetilde{p}^{\beta}_{i,j} - p^{\beta}_{i,j}| \leq \gamma.$
\end{claim}
\begin{proof}
The desired subroutine amounts to a straightforward random sampling procedure, which we include here for the sake of completeness.
We will use the following elementary fact, a simple consequence of the additive Chernoff bound.
\begin{fact} \label{fact:rs}
Let $X$ be a random variable taking values in the range $[-1, 1]$.
Then $\E[X]$ can be estimated to within an additive $\pm \tau$, with confidence probability
$1-\delta$, using $m = \Omega((1/\tau^2)\log(1/\delta))$ independent samples from $X$. In particular,
the empirical average $\widehat{X}_m = (1/m) \littlesum_{i=1}^m X_i$, where the $X_i$'s are independent samples of $X$,
satisfies $\Pr \left[ | \widehat{X}_m - \E[X] | \leq \tau \right] \geq 1-\delta.$
\end{fact}
\noindent We shall refer to this as ``empirically estimating'' the value of $\E[X]$.

Consider the indicator function $I_{H^{\beta}_{ij}}$ of the set $H^{\beta}_{ij}$, i.e., $I_{H^{\beta}_{ij}}: \W \to \{0, 1\}$
with $I_{H^{\beta}_{ij}}(x) = 1$ if and only if
$x \in H^{\beta}_{ij}$. It is clear that $\E_{x \sim D_i} \left[ I_{H^{\beta}_{ij}}(x) \right] = D_i (H^{\beta}_{ij}) = p^{\beta}_{i,j}$.
The subroutine is described in the following pseudocode:
\begin{framed}
\noindent Subroutine  {\tt Estimate}$({D}_i, H^{\beta}_{ij},  \gamma, \delta)$:

\smallskip

\noindent {\bf Input:} Sample access to $D_i$ and oracle access to $\eval_{D_k}\nnew{(\beta)}$, $k=i, j$.

\noindent {\bf Output:} A number $\widetilde{p}^{\beta}_{ij}$ such that with probability $1-\delta$
it holds $|\widetilde{p}^{\beta}_{ij} - D_i (H^{\beta}_{ij}) | \leq \gamma.$

\begin{enumerate}

\item  Draw $m=\Theta \left( (1/{\gamma}^2)  \log(1/\delta)\right)$ samples
$\mathbf{s} = \{ s_{\ell} \}_{\ell=1}^{m}$ from $D_i$.

\item For each sample $s_{\ell}$, $\ell \in [m]$:
\begin{enumerate}
\item[(a)] Use the oracles $\eval_{D_i}\nnew{(\beta)}$, $\eval_{D_j}\nnew{(\beta)}$, to approximately evaluate $D_i(s_{\ell})$,  $D_j(s_{\ell})$.

\item[(b)] If \nnew{$\widetilde{D}^{\beta}_i(s_{\ell}) \ge \widetilde{D}^{\beta}_j(s_{\ell})$}
set $I_{H^{\beta}_{ij}}(s_{\ell}) =1$, otherwise $I_{H^{\beta}_{ij}}(s_{\ell}) =0$.
\end{enumerate}

\item Set $\widetilde{p}^{\beta}_{ij} = {1 \over m} \littlesum_{\ell=1}^m
I_{H^{\beta}_{ij}}(s_{\ell}).$

\item  Output $\widetilde{p}^\nnew{{\beta}}_{ij}.$
\end{enumerate}
\end{framed}
The computational efficiency of this simple random sampling procedure
follows from the fact that we can efficiently decide membership in $H^{\beta}_{ij}$. To do this, for a given $x \in \W$, we make a query to
each of the oracles $\eval_{D_i}\nnew{(\beta)}$, $\eval_{D_j}\nnew{(\beta)}$ to
obtain the probabilities
\nnew{$\widetilde{D}^{\beta}_i(x)$, $\widetilde{D}^{\beta}_j(x)$}. We
have that $x \in H^{\beta}_{ij}$ (or equivalently $I_{H^{\beta}_{ij}}(x) =1$) if and only if
\nnew{$\widetilde{D}^{\beta}_i(x) \ge \widetilde{D}^{\beta}_j(x)$}.
By Fact~\ref{fact:rs}, applied for the random variable $I_{H^{\beta}_{ij}}(x)$, where $x \sim D_i$, after
$m = \Omega((1/\gamma^2)\log(1/\delta))$ samples from $D_i$ we obtain
a $\pm \gamma$-additive estimate to $p^{\beta}_{i,j}$ with
probability $1-\delta$. For each sample, we make one query to each of the oracles, hence the total number of oracle queries is
$O(m)$ as desired. The only non-trivial arithmetic operations
are the $O(m)$ comparisons done in Step~2(b), and Claim~\ref{claim:sample-pij}
is proved.
\end{proof}

Now we are ready to prove Lemma~\ref{lem:choosehypothesis}.
The algorithm {\tt Choose-Hypothesis}$^D({D}_i,{D}_j,\eps',\delta')$ performing the competition between $D_i$ and $D_j$
is the following:
\begin{framed}
\noindent Algorithm  {\tt Choose-Hypothesis}$^D({D}_i,{D}_j,\eps',\delta')$:

\smallskip

\noindent {\bf Input:} Sample access to $D$ and $D_k$, $k=i, j$, oracle access to $\eval_{D_k}\nnew{(\beta)}$, $k=i, j$.

\begin{enumerate}

\item Set $\widetilde{p}^{\beta}_{i,j} =${\tt Estimate}$({D}_i, H^{\beta}_{ij},  \eps'/8, \delta'/4)$ and
                 $\widetilde{q}^{\beta}_{i,j} =${\tt Estimate}$({D}_j,
H^{\beta}_{ij},  \eps'/8, \delta'/4)$.

\item  If $\widetilde{p}^{\beta}_{i,j}-\widetilde{q}^{\beta}_{i,j} \leq \nnew{9 \eps' / 2}$, declare a draw and return either $i$ or $j$. Otherwise:

\item  Draw $m'=\Theta \left( (1/{\eps'}^2)  \log(1/\delta')\right)$ samples
$\mathbf{s'} = \{ s_{\ell} \}_{\ell=1}^{m'}$ from $D$.

\item For each sample $s_{\ell}$, $\ell \in [m']$:
\begin{enumerate}
\item[(a)] Use the oracles $\eval_{D_i}\nnew{(\beta)}$, $\eval_{D_j}\nnew{(\beta)}$ to evaluate \nnew{$\widetilde{D}^{\beta}_i(s_{\ell})$,
$\widetilde{D}^{\beta}_j(s_{\ell})$.}

\item[(b)] If \nnew{$\widetilde{D}^{\beta}_i(s_{\ell}) \ge \widetilde{D}^{\beta}_j(s_{\ell})$} set $I_{H^{\beta}_{ij}}(s_{\ell}) =1$,
otherwise $I_{H^{\beta}_{ij}}(s_{\ell}) =0$.
\end{enumerate}

\item Set $\tau = {1 \over m'} \littlesum_{\ell=1}^{m'}  I_{H^{\beta}_{ij}}(s_{\ell})$,
          i.e., $\tau$ is the fraction of  samples that fall inside $H^{\beta}_{ij}.$

\item  If $\tau > \widetilde{p}^{\beta}_{i, j}- \new{{13 \over 8}} \eps'$, declare $D_i$ as winner and return $i$; otherwise,

\item  if $\tau < \widetilde{q}^{\beta}_{i,j}+ \new{{13 \over 8}} \eps'$, declare $D_j$ as winner and return $j$; otherwise,

\item  declare a draw and return either $i$ or $j$.
\end{enumerate}
\end{framed}

It is not hard to check that the outcome of the competition does not depend on the ordering of the pair of distributions provided in the input; that is, on inputs $(D_i,D_j)$ and
$(D_j, D_i)$ the competition outputs the same result for a fixed set of samples $\{ s_1,\ldots,s_{m'}\}$ drawn from $D$.

The upper bounds on sample complexity, query complexity and number of arithmetic operations can be straightforwardly verified.
Hence, it remains to show correctness.
\new{By Claim~\ref{claim:sample-pij} and a union bound, with probability at least $1-\delta'/2$, we will have that
$|\widetilde{p}^{\beta}_{i,j} - p^{\beta}_{i,j}| \le \eps'/8$ and
$|\widetilde{q}^{\beta}_{i,j} - q^{\beta}_{i,j}| \le \eps'/8$. In the following, we condition on this good event.}
The correctness of {\tt Choose-Hypothesis} is then an immediate consequence of the following claim.

\begin{claim} \label{claim:correct-ht}
Suppose that $\dtv(D, D_i) \leq \eps'$. Then:
\begin{enumerate}
\item[(i)] If $\dtv(D, D_j)>6 \eps'$, then the probability that the competition between
$D_i$ and $D_j$ does not declare $D_i$ as the winner is at most $e^{- {m' \eps'^2/8 } }$.  (Intuitively,
if $D_j$ is very far from $D$ then it is very likely that $D_i$ will be declared winner.)

\item[(ii)]  \ignore{If $\dtv(D, D_j)>\nnew{13 \eps'/4}$, then} The probability that the competition between
$D_i$ and $D_j$ declares $D_j$ as the winner is at most $e^{- {m' \eps'^2/ 8 } }$.  (Intuitively, \new{since $D_i$ is close to $D$,
a draw with some other $D_j$ is possible}, but it is very unlikely that $D_j$ will be declared winner.)
\end{enumerate}
\end{claim}
\begin{proof}
Let $r^\nnew{{\beta}}=D(H^{\beta}_{ij})$. The definition of the variation distance implies that $ |r^\nnew{{\beta}}-p^{\beta}_{i,j}| \le  \dtv(D, D_i) \le  \eps'.$
Therefore, we have that
$|r^\nnew{{\beta}} - \widetilde{p}^{\beta}_{i,j}| \le  |r^\nnew{{\beta}}-
p^{\beta}_{i,j}| + |\widetilde{p}^{\beta}_{i,j} - p^{\beta}_{i,j}| \le 9\eps'/8.$
Consider the indicator ($0/1$) random variables
$\{Z_{\ell}\}_{\ell=1}^{m'}$ defined as $Z_{\ell}=1$ if and only if $s_{\ell} \in H^{\beta}_{ij}$. Clearly, $\tau={1 \over m'} \littlesum_{\ell=1}^{m'} Z_{\ell}$ and
$\E_{\mathbf{s'}}[\tau]=\E_{s_{\ell} \sim D} [Z_\ell]=r^\nnew{{\beta}}$. Since the $Z_\ell$'s are mutually independent, it follows from the Chernoff bound that
$\Pr[\tau \le r^\nnew{{\beta}}-{\eps'/2}] \le e^{- {m' {\eps'}^2/8 } }$. Using $|r^\nnew{{\beta}} - \widetilde{p}^{\beta}_{i,j}| \le 9\eps'/8.$
we get that $\Pr[\tau\le \widetilde{p}^{\beta}_{i,j}-\new{13\eps'/8}] \le e^{- {m' {\eps'}^2/8 } }$.
\begin{itemize}
\item For part (i): If $\dtv(D, D_j) > 6 \eps'$, from the triangle inequality we get that $p_{i,j}-q_{i,j}=\dtv(D_i, D_j) > 5 \eps' $
\nnew{Claim~\ref{claim:dtv-approx} implies that $p^{\beta}_{i,j} - q^{\beta}_{i,j}  > 19 \eps' /4$ and our conditioning finally gives
$\widetilde{p}^{\beta}_{i,j} - \widetilde{q}^{\beta}_{i,j}  > 9 \eps' /2$}.  Hence, the algorithm will go beyond Step~2, and with probability at least
$1-e^{- {m' \eps'^2/8 } }$, it will stop at Step~6, declaring $D_i$ as the winner of the competition between $D_i$ and $D_j$.

\item For part (ii):  If  \nnew{$\widetilde{p}^{\beta}_{i,j} -
\widetilde{q}^{\beta}_{i,j}  \le 9 \eps' /2$} then the competition
declares a draw,
hence $D_j$ is not the winner. Otherwise we have \nnew{$\widetilde{p}^{\beta}_{i,j} - \widetilde{q}^{\beta}_{i,j}  > 9 \eps' /2$}
and the  argument of the previous paragraph implies that the competition between $D_i$ and $D_j$ will declare $D_j$
as the winner with probability at most $e^{- {m' \eps'^2/8 } }$.
\end{itemize}
This concludes the proof of Claim~\ref{claim:correct-ht}.
\end{proof}
\ignore { The correctness of the lemma follows directly from the above claim. The bound on the sample complexity is also clear.
To analyze the running time, we note that the only operation we need is to compute the value of $\tau$.
To do this, we first need to read the samples $s_{\ell}$. Reading a single sample requires time $\Theta (\log |\mathcal{W}|)$ bit operations,
(assuming we have an efficient encoding of the finite set $\mathcal{W}$). Hence, to read all the samples we need $O( m' \cdot \log |\mathcal{W}|)$ bit operations.
Now, for each sample $s_{\ell}$, $\ell \in [m']$, we need to decide whether it belongs to the set $H_{ij}$, i.e., if $D_i(s_{\ell}) > D_{j}(s_{\ell})$.
To do this, we need to call each evaluation oracle once. Hence, the number of oracle calls is at most $m'$ (for each oracle).
The running time is dominated by $O(m')$ oracle calls.}  
This completes the proof of Lemma~\ref{lem:choosehypothesis}.
\end{proof}

We now proceed to describe the algorithm ${\cal T}^D$ and establish Proposition~\ref{prop:log-cover-size}.
The algorithm performs a tournament by running the competition
{\tt Choose-Hypothesis}$^D(D_i,D_j,\eps, \delta/(2N))$ for every pair of distinct distributions
$D_i,D_j$ in the collection ${\cal D}_{\eps}$.  It outputs a distribution $D^{\star} \in {\cal D}_{\eps}$
that was never a loser (i.e., won or achieved a draw in all its competitions).
If no such distribution exists in ${\cal D}_{\eps}$ then the algorithm outputs ``failure.''
A detailed pseudocode follows:
\begin{framed}
\noindent Algorithm  $\mathcal{T}^D( \{D_j\}_{j=1}^N, \eps, \delta)$:

\smallskip

\noindent {\bf Input:} Sample access to $D$ and $D_k$, $k \in [N]$, and oracle access to $\eval_{D_k}$, $k \in [N]$.

\begin{enumerate}

\item Draw $m = \Theta \left( (1/\eps^2) (\log N+ \log(1/\delta)) \right)$ samples from $D$ and each $D_k$, $k \in [N]$.

\item For all $i, j \in [N]$, $i \neq j$, run {\tt Choose-Hypothesis}$^D({D}_i,{D}_j,\eps,\delta/(2N))$ using this sample.

\item Output an index $i^{\star}$ such that $D_{i^{\star}}$ was never  declared a loser, if one exists.

\item Otherwise, output ``failure''.
\end{enumerate}
\end{framed}
We now proceed to analyze the algorithm.
The bounds on the sample complexity, running time and query complexity to the evaluation oracles follow from
the corresponding bounds for {\tt Choose-Hypothesis}.
Hence, it suffices to show correctness. We do this below.

By definition, there exists some $D_i \in {\cal D}_\eps$ such that $\dtv(D,D_i) \leq \eps.$
By Claim~\ref{claim:correct-ht}, the distribution
$D_i$ never loses a competition against any other
$D_j \in {\cal D}_\eps$ (so the algorithm does not output
``failure'').\ignore{
Consider any $D_j \in {\cal D}_\eps$.
If $\dtv(D, D_j) > \nnew{13 \eps/4}$,
by Claim~\ref{claim:correct-ht}(ii) the probability that $D_i$ loses to $D_j$ is at most $2e^{-m \eps^2 /8} \leq \delta/(2N).$
On the other hand, if $\dtv(D, D_j) \leq \nnew{13\eps/4}$, the triangle
inequality gives that $\dtv(D_i, D_j) \leq \nnew{17\eps/4}$ and \nnew{by Claim~\ref{claim:dtv-approx} and our conditioning
$|\widetilde{p^{\beta}_{i,j}} - \widetilde{q^{\beta}_{i,j}}| \leq 9\eps/2$},
thus $D_i$ draws against $D_j.$}  A union
bound over all $N$ distributions in ${\cal D}_{\eps}$ shows that with probability $1-\delta/2$, the distribution $D'$ never loses a competition.

We next argue that with probability at least $1-\delta/2$, every distribution $D_j \in {\cal D}_\eps$ that never loses has small variation distance from $D.$
Fix a distribution $D_j$ such that $\dtv(D_j,D) > 6 \eps$; Claim~\ref{claim:correct-ht}(i) implies that $D_j$ loses to $D_i$ with probability
$1 - 2e^{-m \eps^2 /8} \geq 1 - \delta/(2N)$.  A union bound yields that with probability $1-\delta/2$, every distribution $D_j$ that has
$\dtv(D_j, D) > 6 \eps$ loses some competition.

Thus, with overall probability at least $1-\delta$, the tournament does not output ``failure'' and outputs some distribution $D^{\star}$ such that $\dtv(D, D^{\star})$ is at most $6 \eps.$
The proof of Proposition~\ref{prop:log-cover-size} is now complete.
\end{proof}

\begin{remark} \label{rem:sampler}
{\em As stated Proposition~\ref{prop:log-cover-size} assumes that
algorithm ${\cal T}^D$ has access to samplers for all the
distributions $D_k$, so each call to such a sampler is guaranteed
to output an element distributed according to $D_k$.
Let $D^\bot_k$ be a distribution over ${\cal W} \cup \{\bot\}$ which
is such that (i) $D^\bot_k(\bot) \leq 1/2$, and (ii) the conditional 
distribution $(D^\bot_k)_{\cal W}$ of $D^\bot_k$ conditioned on not
outputting $\bot$ is precisely $D_k.$
It is easy to see that the proof of
Proposition~\ref{prop:log-cover-size} extends to a setting in which
${\cal T}^D$ has access to samplers for $D^\bot_k$ rather than
samplers for $D_k$; each time a sample from $D_k$ is required
the algorithm can simply invoke the sampler for $D^\bot_k$ repeatedly
until an element other than $\bot$ is obtained. (The low-probability
event that many repetitions are ever needed can be ``folded into''
the failure probability $\delta$.)}
\end{remark}

\fi

\newcommand{\default}{\mathrm{default}}




\section{A general technique for inverse approximate uniform generation}
\label{sec:generaltechnique}

In this section we present {a general technique}
for solving  inverse approximate uniform generation problems.
{Our main} positive results follow this conceptual framework.
At the heart of our approach is a new type of algorithm which we call a
\emph{densifier} for a concept class $\C$.
Roughly speaking, this is an algorithm
which, given uniform random positive examples of an unknown $f \in \C$,
constructs a set $S$ which (essentially) contains
all of $f^{-1}(1)$ and which is such that $f^{-1}(1)$ is
``dense'' in $S$.
Our main result in this section, Theorem~\ref{thm:learn-by-dense},
states (roughly speaking) that the existence
of (i) a computationally efficient densifier, (ii) an efficient
approximate uniform generation algorithm, (iii) an efficient
approximate counting algorithm, and (iv) an
efficient {\em statistical query (SQ)} learning algorithm, together
suffice to yield an efficient algorithm for our inverse approximate uniform generation problem.

\ifnum\confversion=0
{We have already defined approximate uniform generation and
approximate counting algorithms, so we need to define $\sq$ learning
algorithms and densifiers.}
The {\em statistical query} (SQ) learning model is
\else
Recall that the {\em statistical query} (SQ) learning model is
\fi
a natural restriction of the PAC learning model in which a learning algorithm
is allowed to obtain estimates of statistical properties of the examples but cannot directly access the examples themselves.
Let $D$ be a distribution over $\bn$. In the SQ model~\cite{Kearns:98}, the learning algorithm has access to a {\em statistical query oracle}, $\stat(f,D)$,
to which it can make a query of the form $(\chi, \tau)$, where $\chi : \bn \times \bits  \to  [-1,1]$ is the \emph{query function} and $\tau>0$ is the \emph{tolerance}.
The oracle responds with a value $v$ such that $ \left| \E_{x \sim D} \left[ \chi \left( x, f(x) \right) \right]  -  v \right| \le \tau$, where $f \in \C$ is the target concept. The goal of the algorithm is to output a hypothesis $h : \bn \to \bits$ such that $\Pr_{x \sim D} [h(x) \ne  f(x)] \le \eps$.
\ifnum\confversion=1
(See the full version for a precise definition.)
We sometimes write an ``$(\eps, \delta)$--$\sq$ learning algorithm''
to explicitly state the accuracy parameter $\eps$ and confidence parameter
$\delta$.

\else
The following is a precise definition:

\begin{definition} \label{def:sq-algo}
Let $\C$ be a class {of $n$-variable boolean functions} and $D$ be a distribution over $\bn$.  An
{\em $\sq$ learning algorithm for $\C$ under $D$} is
a randomized algorithm $\A^\C_{\sq}$ that for every $\eps, \delta>0$, every
target concept $f \in \C$,
on input $\eps$, $\delta$ and with access to oracle $\stat(f,D)$
{and to independent samples drawn from $D$},
outputs with probability $1-\delta$
a hypothesis $h:\bn \to \bits$ such that $\Pr_{x \sim D} [h(x) \ne  f(x)] \le \eps$.
Let $t_1(n, 1/\eps, 1/\delta)$ be the running time of $\A^\C_{\sq}$ (assuming each oracle query is answered in unit time),
$t_2 (n)$ be the maximum running time to evaluate any query provided to $\stat(f,D)$
and $\tau(n, 1/\eps)$ be the minimum value of the tolerance parameter
ever provided to $\stat(f,D)$ in the course of $\A^\C_{\sq}$'s execution.
We say that $\A^\C_{\sq}$ is {\em efficient} (and that $\C$ is {\em efficiently} $\sq$ learnable with respect to distribution $D$),
if $t_1(n, 1/\eps, 1/\delta)$ is polynomial in $n$, $1/\eps$ and $1/\delta$, $t_2(n)$ is polynomial in $n$
and $\tau(n, 1/\eps)$ is lower bounded by an inverse polynomial in $n$ and $1/\eps$. We call
an $\sq$ learning algorithm $\A^\C_{\sq}$ for $\C$ {\em distribution independent} if $\A^\C_{\sq}$ succeeds for any distribution $D$.
If $\C$ has an efficient distribution independent  $\sq$ learning algorithm we say that $\C$ is {\em efficiently $\sq$ learnable (distribution independently)}.
\end{definition}
\fi

\ifnum\confversion=0
We sometimes write an ``$(\eps, \delta)$--$\sq$ learning algorithm''
to explicitly state the accuracy parameter $\eps$ and confidence parameter
Throughout this paper, we will only deal with distribution independent
$\sq$ learning algorithms.
\fi

\ignore{$\sq$ learning algorithms are well known \cite{Kearns:98} to
``automatically'' yield PAC learning algorithms that are robust in the
presence of random misclassification noise affecting the labels
of examples; we will see how this robustness is useful for us later.
\mnote{tweak prose}
}

\ignore{
We measure the time complexity of an SQ--learning algorithm in the following way (which is somewhat
crude but adequate for our purposes).
Let $\A_\sq$ be an SQ--learning algorithm that runs for $t_1(\eps,n)$ time steps, where each
query to the $\mathrm{STAT}(f,D)$ oracle is considered a unit time operation. Suppose that each
query function $\chi$ that $\A_\sq$ gives as input to $\mathrm{STAT}(f,D)$ can be evaluated
in time $t_2(n)$, and that the minimum value $\tau$ of the tolerance parameter that is
ever provided to $\mathrm{STAT}(f,D)$ in the course of $\A_\sq$' s execution is $\tau_3.$  Then, the time complexity
of $\A_\sq$ is said to be $T_{\A_\sq}(\eps,n) \eqdef t_1(\eps, n) \cdot \poly \left( t_2(n), 1/\tau_3 \right).$  (The motivation for
this definition is that an algorithm with access to labeled examples $(x,f(x))$, $x \sim D$,
could simulate the execution of $\A_\sq$ and output an $\eps$-accurate hypothesis $h$ with confidence $1-\delta$
in time $T_{\A_\sq}(\eps,n) \cdot \log(1/\delta)$.)
}

To state our main result, we introduce the notion of a
{\em densifier} for a class $\C$ of Boolean functions.
Intuitively, a densifier is an algorithm which is given access to samples from
$\U_{f^{-1}(1)}$ (where $f$ is an unknown element of $\C$)
and outputs a subset $S \subseteq \bn$ which is such that
(i) $S$ contains ``almost all'' of $f^{-1}(1)$, but (ii) $S$ is ``much smaller'' than
$\bn$ -- in particular it is small enough that $f^{-1}(1) \cap S$ is
(at least moderately) ``dense'' in $S$.

\begin{definition}\label{def:dense}
{Fix a function $\gamma(n,1/\eps,1/\delta)$
taking values in $(0,1]$ and a class $\C$ of $n$-variable Boolean functions.}
An algorithm $\A^{(\C, \C')}_{\den}$
is said to be {\em a $\gamma$-densifier for function class $\C$ using class $\C'$}
if it has the following {behavior}:
\new{For every $ \eps,\delta > 0$, every $1/2^n \leq \widehat{p} \leq 1$,}
and every $f \in \C$,
given as input $\eps,\delta, \widehat{p}$ and a set of independent
samples from $\U_{f^{-1}(1)}$,
the following holds:
Let $p \eqdef \Pr_{x \sim \U_n}[f(x)=1].$
If $p \leq \widehat{p} < (1+\eps)p$, then
with probability at least $1 - \delta$, algorithm
$\A^{(\C, \C')}_{\den}$ outputs a function
$g \in \C'$ such that:
\begin{enumerate}
\item[(a)] $\Pr_{x \sim \U_{f^{-1}(1)}}\left[ g(x)=1 \right] \ge 1-\eps.$
\item[(b)] $\Pr_{x \sim \U_{g^{-1}(1)}} \left[ f(x)=1 \right] \ge \gamma
(n,1/\eps,1/\delta).$
\end{enumerate}
 \end{definition}

We will sometimes write an ``$(\epsilon,{\gamma},\delta)$--densifier'' to explicitly state the parameters in the definition.

Our main conceptual approach is summarized in the following theorem:

 \begin{theorem}[General Upper Bound]\label{thm:learn-by-dense}
Let $\C, \C'$ be classes of $n$-variable boolean functions. Suppose that
\begin{itemize}
\item $\A^{(\C, \C')}_{\den}$ is an $(\eps, {\gamma}, \delta)$-densifier for $\C$ using $\C'$
running in time $T_{\den}(n, 1/\eps, 1/\delta)$.

\item $\A^{\C'}_{\gen}$ is an $(\eps,\delta)$-approximate uniform generation algorithm for $\C'$ running in time $T_{\gen}(n, 1/\eps,{1/\delta})$.

\item $\A^{\C'}_{\co}$ is an $(\eps, \delta)$-approximate counting algorithm for $\C'$ running in time $T_{\co}(n, 1/\eps, 1/\delta)$.

\item $\A^\C_{\sq}$ is an $(\eps, \delta)$-$\sq$ learning algorithm for $\C$ such that: $\A^\C_{\sq}$ runs in time $t_1(n, 1/\eps, 1/\delta)$ ,
$t_2 (n)$ is the maximum time needed to evaluate any query provided
to $\stat(f,D)$, and $\tau(n, 1/\eps)$ is the minimum value
of the tolerance parameter ever provided to $\stat(f,D)$ in the course of $\A^\C_{\sq}$'s execution.
\end{itemize}
Then there exists an inverse approximate uniform generation
algorithm $\A^\C_{\inv}$ for $\C$.
\new{
The running time of $\A^\C_\inv$ 
is polynomial in 
$T_{\den}(n, 1/\eps, 1/\delta)$, $1/\gamma$, $T_{\gen}(n, 1/\eps,1/\delta)$, $T_{\co}$ $(n,1/\eps,1/\delta)$, $t_1(n,1/\eps,1/\delta)$,
$t_2(n)$ and $1/\tau(n,1/\eps)$.}
\ifnum\confversion=0
\footnote{It is straightforward to derive an explicit running time bound
for $\A^\C_\inv$ in terms of the above functions from our analysis, but
the resulting expression is extremely long and rather
uninformative so we do not provide it.}
\fi
\end{theorem}

\noindent {\bf Sketch of the algorithm.}
The inverse approximate uniform generation algorithm $\A^\C_{\inv}$ for $\C$ works in three main conceptual steps.
 Let $f \in \C$ be the unknown target function and recall that our algorithm $\A^\C_{\inv}$ is given access to samples from $\U_{f^{-1}(1)}$.

\begin{enumerate}

\item[(1)] In the first step, $\A^\C_{\inv}$ runs the densifier $\A^{(\C, \C')}_\den$ on a set of samples from $\U_{f^{-1}(1)}$.
Let  $g \in \C'$ be the output function of $\A^{(\C, \C')}_\den$.

\end{enumerate}
Note that
by setting the input to the
approximate uniform generation algorithm $\A^{\C'}_{\gen}$ to $g$, we obtain
an approximate sampler $C_g$ for $\U_{g^{-1}(1)}$.
The output distribution $D'$ of this sampler, is by definition
supported on $g^{-1}(1)$ and is close to $D = \U_{g^{-1}(1)}$ in total variation distance.
\begin{enumerate}
\item[(2)] The second step is to run the $\sq$-algorithm $\A^\C_{\sq}$ to learn the function $f \in C$
under the distribution $D$.
Let $h$ be the hypothesis constructed by $\A^\C_{\sq}$.

\item[(3)] In the third and final step, the algorithm simply samples from $C_g$
until it obtains an example $x$ that has $h(x)=1$, and outputs this $x$.

\end{enumerate}

{
\begin{remark} \label{rem:counting}
{\em The reader may have noticed that the above sketch does not
seem to use the approximate counting algorithm $\A^{\C'}_{\co}$; we will
revisit this point below.}
\end{remark}

\begin{remark} \label{rem:standard}
{\em The connection between the above algorithm sketch and the ``standard approach'' discussed in the Introduction is as follows:
The function $g \wedge h$ essentially corresponds to the reconstructed object $\tilde{x}$ of the ``standard approach.''
The process of sampling from $C_g$ and doing rejection sampling until an input that satisfies $h$ is obtained,
essentially corresponds to the $A_\sample$ procedure of the ``standard approach.'' }
\end{remark}
}

{
\subsection{Intuition, motivation and discussion.}
\label{ssec:intuition-motivation-discussion}
To motivate the high-level idea behind our algorithm,
consider a setting in which $f^{-1}(1)$ is only a tiny fraction (say
$1/2^{\Theta(n)}$) of $\{-1,1\}^n.$  It is intuitively clear
that we would like to use some kind of a learning algorithm
in order to come up with a good approximation of $f^{-1}(1)$, but
we need this approximation to be accurate at the ``scale'' of
$f^{-1}(1)$ itself rather than at the scale of all of $\{-1,1\}^n$,
so we need some way to ensure that the learning algorithm's hypothesis
is accurate at this small scale.
By using a densifier to construct $g$ such that $g^{-1}(1)$ is not
too much larger than $f^{-1}(1)$, we can use the distribution
$D = \U_{g^{-1}(1)}$ to run a learning algorithm and obtain a good approximation of
$f^{-1}(1)$ at the desired scale. {(Since $D$ and $D'$ are close in variation distance, this implies
we also learn $f$ with respect to $D'$.)}

To motivate our use of an $\sq$ learning algorithm rather than a
standard PAC learning algorithm,
observe that there seems to be no way to obtain correctly labeled examples
distributed according to $D$.
However, we show that it is possible to accurately simulate statistical
queries under $D$ having access only to random positive
examples from $f^{-1}(1)$ and to unlabeled examples drawn from $D$
(subject to additional technical caveats discussed
\ifnum\confversion=1
in the full version).
\else
below).
\fi
We discuss the issue of how it is possible to successfully use an $\sq$
learner in our setting in more detail below.
}

{
\medskip
\noindent {\bf Discussion and implementation issues.}
While the three main conceptual steps (1)-(3) of our algorithm may (hopefully) seem quite intuitive
in light of the preceding motivation,
a few issues immediately arise in thinking about how to implement
these steps.}
The first one concerns running the $\sq$-algorithm $\A^\C_{\sq}$ in Step~2
to learn $f$ under distribution {$D$ (recall that $D = \U_{g^{-1}(1)}$ and is close to
$D'$).}
Our algorithm $\A^\C_{\inv}$ needs to be able to efficiently
simulate $\A^\C_{\sq}$ given its available information.
While it would be easy to do so given access to random labeled
examples $(x, f(x))$, where $x \sim {D}$, such information is not
available in our setting. To overcome this obstacle,
\ifnum\confversion=0
we show (see Proposition~\ref{prop:sq-sim})
\else
in the full version we show
\fi
that for {\em any} {samplable} distribution {$D$},
we can efficiently simulate a statistical query algorithm under
{$D$} using samples from {$D_{f, +}$}.
This does not quite solve the
problem, since we only have samples from $\U_{f^{-1}(1)}$. However,
\ifnum\confversion=0
we show  (see Claim~\ref{claim:sim-sample})
\else
it can be shown
\fi
that for our setting, i.e., for
{$D = \U_{g^{-1}(1)}$},
we can simulate a sample from {$D_{f, +}$}
by a simple rejection sampling procedure using samples from $\U_{f^{-1}(1)}$
and query access to $g$.

{Some more issues remain to be handled. First, the simulation of the statistical query algorithm sketched in
the previous paragraph only works under the assumption that we are given
a sufficiently accurate approximation $\widetilde{b_f}$ of the probability $\Pr_{x \sim D}[f(x) = 1]$.
(Intuitively, our approximation should be smaller than the smallest tolerance $\tau$ provided to the statistical query oracle by the algorithm $\A^\C_{\sq}$.)
Second, by Definition~\ref{def:dense}, the densifier only succeeds under the assumption
that it is given in its input an $(1+\eps)$-multiplicative approximation $\widehat{p}$ to $p = \Pr_{x\in \U_n}[f(x)=1]$.

We handle these issues as follows:
First, we show
\ifnum\confversion=0
(see Claim~\ref{claim:estimate-bias})
\else
in the full version
\fi
that, given an accurate estimate $\widehat{p}$ and a ``dense'' function $g \in \C'$,
we can use the approximate counting algorithm $\A^{\C'}_{\co}$ to efficiently compute an accurate estimate $\widetilde{b_f}$.
(This is one reason why Theorem~\ref{thm:learn-by-dense} requires an approximate counting algorithm for $\C'.$)
To deal with the fact that we do not a priori have an accurate estimate $\widehat{p}$, we
run our sketched algorithm for all possible values of $\Pr_{x \sim \U_n}[f(x) = 1]$ in an appropriate multiplicative ``grid''
of size $N = O({n}/\eps)$, {covering all possible values from $1/2^n$ to $1$.}
We thus obtain a set $\D$ of $N$ candidate distributions one of which is guaranteed to be close to
the true distribution $\U_{f^{-1}(1)}$ in variation distance.
At this point, we would like to apply our hypothesis testing machinery (Proposition~\ref{prop:log-cover-size}) to find such a distribution.
However, in order to use Proposition~\ref{prop:log-cover-size},
in addition to sample access to the candidate distributions (and the
distribution being learned), we also require a {\em multiplicatively accurate}
approximate evaluation oracle to evaluate the probability mass of any point
under the candidate distributions.
\ifnum\confversion=0
We show (see Lemma~\ref{lem:approx-eval-for-hatDi})
\else
In the full version we show
\fi
that this is possible in our generic setting, using
properties of the densifier and  the approximate {counting algorithm}
$\A^{\C'}_\co$ for $\C'$.}

\ifnum\confversion=1
See the full version for a proof of Theorem~\ref{thm:learn-by-dense}.
\else
{
Now we are ready to begin the detailed proof of
Theorem~\ref{thm:learn-by-dense}.
}

\subsection{Simulating statistical query algorithms.}

Our algorithm $\A^\C_{\inv}$ will need to simulate a statistical query algorithm for $\C$,
 with respect to a specific distribution $D$.
Note, however that $\A_\inv$ only has access to \new{uniform positive
examples of $f \in \C$, i.e., samples from
$\U_{f^{-1}(1)}$.}
Hence we need to show that a statistical query algorithm can be efficiently
simulated in \new{such a} setting.
To do this it suffices to show that one can efficiently
\new{provide valid responses to}
queries to the statistical query oracle $\stat(f, D)$, i.e.,
\new{that one can simulate the oracle.}
Assuming this can be done, the simulation algorithm $\A_{\sqs}$ is very simple:
Run the statistical query algorithm $\A_{\sq}$, and whenever it makes a query
to $\stat(f, D)$, simulate it.
To this end, in the following lemma we describe a procedure that simulates
\new{an $\sq$ oracle.
}
(Our approach here is similar to that of earlier 
simulation procedures that have been given in the literature,
see e.g. Denis \etal \cite{DGL05}.)

\begin{lemma} \label{lem:sq-sim}
Let $\C$ be a concept class \new{over $\{-1,1\}^n$}, $f \in \C$, and $D$
be a samplable distribution over $\bn$.
There exists an algorithm {\tt Simulate-STAT}$^D_f$ with the following
properties:  It is given access to
independent samples from $D_{f,+}$, and takes as input
a number $\widetilde{b_f} \in [0,1]$,
a $t(n)$-time computable query function $\chi:\bn \times \bits \to [-1,1]$,
a tolerance $\tau$ and a confidence $\delta$.  It has the following behavior:
it uses $m = O\left( (1/\tau^2) \log(1/\delta) \right)$ samples from \nnew{$D$ and} $D_{f,+}$,
runs in time $O\left(m \cdot t(n) \right),$ and \nnew{if  $| \widetilde{b_f} - \Pr_{x \sim D} [f(x)=1] | \leq  \tau'$}, then with probability $1-\delta$
it outputs a number $v$ such that
\begin{equation} \label{eqn:req}
\left| \E_{x \sim D} \left[ \chi \left( x, f(x) \right) \right]  -  v \right| \le \tau+\tau'.
\end{equation}
\end{lemma}

\begin{proof}
\ignore {
We will use the following elementary fact, a simple consequence of the additive Chernoff bound.
\begin{fact} \label{fact:rs}
Let $X$ be a random variable taking values in the range $[-1, 1]$.
Then, its expectation can be estimated to within an additive $\pm \tau$, with confidence probability
$1-\delta$, using $m = \Omega((1/\tau^2)\log(1/\delta))$ independent samples from $X$. In particular,
the empirical average $\widehat{X}_m = (1/m) \littlesum_{i=1}^m X_i$, where the $X_i$'s are independent samples of $X$,
satisfies $\Pr \left[ | \widehat{X}_m - \E[X] | \leq \tau \right] \geq 1-\delta.$
\end{fact}
\noindent We shall refer to this as ``empirically estimating'' the value of $\E[X]$. }
To prove the lemma, we start by rewriting the expectation in (\ref{eqn:req}) as follows:
$$
\E_{x \sim D} \left[ \chi (x, f(x)) \right]  = \E_{x \sim D_{f, +}} \left[ \chi (x, 1) \right] \cdot \Pr_{x \sim D}[f(x)=1] +  \E_{x \sim D_{f, -}} \left[ \chi (x, -1) \right] \cdot \Pr_{x \sim D}[f(x)=-1].
$$
We also observe that
$$ \E_{x \sim D} \left[ \chi (x, -1) \right] =\E_{x \sim D_{f, +}} \left[ \chi (x, -1) \right] \cdot  \Pr_{x \sim D}[f(x)=1] +  \E_{x \sim D_{f, -}} \left[ \chi (x, -1) \right] \cdot \Pr_{x \sim D}[f(x)=-1].$$
Combining the above equalities we get
\begin{eqnarray} \label{eqn:identity}
\E_{x \sim D} \left[ \chi (x, f(x)) \right]  = \E_{x \sim D} \left[ \chi (x, -1) \right]  + \E_{x \sim D_{f, +}} \left[ \chi (x, 1) - \chi (x, -1) \right]  \cdot \Pr_{x \sim D}[f(x)=1].
\end{eqnarray}
Given the above identity, the algorithm {\tt Simulate-STAT}$^D_f$ is very simple:
We use random sampling from $D$
to empirically estimate the expectations $\E_{x \sim D} \left[ \chi (x, -1)
\right]$ \new{(recall that $D$ is assumed to be a samplable distribution),}
\new{and we use the independent samples from $D_{f,+}$ to empirically estimate}
$ \E_{x \sim D_{f, +}} \left[ \chi (x, 1) - \chi (x, -1) \right] $.
Both estimates are obtained
to within an additive accuracy of
$\pm \tau/2$ (with confidence probability $1-\delta/2$ each).
We combine these estimates with our estimate $\widetilde{b_f}$
for $\Pr_{x \sim D} [f(x)=1]$ in the obvious way (see Step~2 of
pseudocode below).

\begin{framed}
\noindent Subroutine   {\tt Simulate-STAT}$^D_f ( \new{D, D_{f,+},} \chi, \tau, \widetilde{b_f}, \delta)$:

\smallskip

\noindent {\bf Input:} Independent samples from \new{$D$ and} $D_{f,+}$, query access to $\chi: \bn \to \bits$,
                                    accuracy $\tau$, $\widetilde{b_f} \in [0,1]$ and confidence $\delta$.

\noindent {\bf Output:} If $|\widetilde{b_f} - \Pr_{x \sim D} [f(x)=1] | \leq  \tau'$,
a number $v$ that with probability $1-\delta$  satisfies $ | \E_{x \sim D} [ \chi ( x, f(x)) ]  -  v | \le \tau+\tau'$.

\begin{enumerate}

\item  Empirically estimate the values
           $\E_{x \sim D} [ \chi (x, -1)]$ and $\E_{x \sim D_{f, +}} [  \chi (x, 1) - \chi (x, -1)]$ to within an additive $\pm \tau/2$ with confidence
           probability $1-\delta/2$. Let $\widetilde{E}_1,\widetilde{E}_2 $ be the corresponding estimates.

\item Output $v = \widetilde{E}_1 + \widetilde{E}_2 \cdot \widetilde{b_f}. $
\end{enumerate}
\end{framed}

\ignore{
To see that the above expectations can indeed by empirically estimated, recall that by assumption both $D$ and $D_{f,+}$ are samplable. }
By Fact~\ref{fact:rs}, we can estimate each expectation using $m = \Theta \left( (1/\tau^2) \log(1/\delta) \right)$ samples (from $D$, $D_{f,+}$ respectively).
For each such sample the estimation algorithm needs to evaluate the function $\chi$ (once for the first expectation and twice for the second).
Hence, the total number of queries to $\chi$ is $O(m)$, i.e., the subroutine  {\tt Simulate-STAT}$^D_f$ runs in time $O(m \cdot t(n))$ as desired.

By a union bound, with probability $1-\delta$ both estimates will be $\pm \tau/2$ accurate. The bound (\ref{eqn:req}) follows from this latter fact
and (\ref{eqn:identity}) by a straightforward application of the triangle inequality. This completes the proof of Lemma~\ref{lem:sq-sim}.
\end{proof}

\noindent Given the above lemma, we can state and prove our general
result \new{for simulating $\sq$ algorithms:}
\begin{proposition} \label{prop:sq-sim}
Let $\C$ be a concept class and $D$ be a samplable distribution over $\bn$. Suppose there exists an $\sq$-learning algorithm $\A_{\sq}$ for $\C$
under $D$ with the following performance:
$\A_\sq$ runs in time $T_1 = t_1(n, 1/\eps, 1/\delta)$, each query provided to $\stat(f,D)$ can be evaluated in time $T_2 = t_2(n)$,
and the minimum value of the tolerance provided to $\stat(f, D)$
in the course of its execution is $\tau = \tau(n, 1/\eps)$.
Then, there exists an algorithm $\A_{\sqs}$ that is given access to
\begin{enumerate}
\item[(i)] independent samples from $D_{f,+}$;  and
\item[(ii)] a number $\widetilde{b_f} \in [0,1]$,
\end{enumerate}
and efficiently simulates the behavior of $\A_{\sq}$. In particular, $\A_\sqs$
has the following performance guarantee:
on input an accuracy $\eps$ and a confidence $\delta$,
it uses $m = O\left( (1/\tau^2)  \cdot \log(T_1/\delta) \cdot T_1  \right)$ samples from \nnew{$D$ and} $D_{f,+}$,
runs in time $T_{\sqs} = O\left(m T_2 \right)$,
and \nnew{if $| \widetilde{b_f} - \Pr_{x \sim D} [f(x)=1]| \leq  \tau/2$ }
then with probability $1-\delta$ it
outputs a hypothesis $h:\bn \to \bits$ such that
$\Pr_{x \sim D} [h(x) \ne f(x)] \leq \eps.$
\end{proposition}

\begin{proof}
The simulation procedure is very simple. We run the algorithm $\A_{\sq}$ by simulating
its queries using algorithm  {\tt Simulate-STAT}$^D_f$.
The algorithm is described in the following pseudocode:

\begin{framed}
\noindent Algorithm   $\A_{\sqs} (\new{D, D_{f, +}}, \eps, \widetilde{b_f}, \delta)$:

\smallskip

\noindent {\bf Input:} Independent samples from $D$ and $D_{f,+}$, $\widetilde{b_f} \in [0,1]$, $\eps$, $\delta >0$.

\noindent {\bf Output:}
\nnew{If $| \widetilde{b_f} - \Pr_{x \sim D} [f(x)=1] | \leq  \tau/2$,}
a hypothesis $h$ that
with probability $1-\delta$  satisfies $\Pr_{x\sim D}[h(x) \neq f(x)] \leq \eps.$

\begin{enumerate}

\new{\item Let $\tau = \tau(n, 1/\eps)$ be the minimum accuracy ever used in a query to $\stat(f, D)$ during the execution of $\A_\sq (\eps, \delta/2)$.}

\item  Run the algorithm $\A_\sq (\eps, \delta/2)$, by simulating each query to $\stat(f, D)$ \new{as follows:}
whenever $\A_\sq$ makes a query $(\chi, \tau)$ to $\stat(f, D)$, the simulation algorithm
runs  {\tt Simulate-STAT}$^D_f (\new{D, D_{f,+},} \chi, \tau/2, \tau/2,  \delta/(2T_1))$.

\item Output the hypothesis $h$ obtained by the simulation.
\end{enumerate}
\end{framed}

Note that we run the algorithm $\A_\sq$ with confidence probability $1-\delta/2$.
Moreover, each query to the $\stat(f,D)$ oracle is simulated with confidence $1-\delta/(2T_1)$.
Since $\A_{\sq}$ runs for at most $T_1$ time steps, it certainly performs at most $T_1$ queries in total.
Hence, by a union bound over these events, with probability $1-\delta/2$ all answers to its queries will be accurate to within
an additive $\pm \tau/2$. By the guarantee of algorithm $\A_\sq$ and a union bound, with probability
$1-\delta$, the algorithm $\A_{\sqs}$ will output a hypothesis $h:\bn \to \bits$
such that $\Pr_{x \sim D} [h(x) \ne f(x)] \le \eps$. The sample complexity and running time follow from the bounds for {\tt Simulate-STAT}$^D_f$.
This completes the proof of Proposition~\ref{prop:sq-sim}.
\end{proof}

Proposition~\ref{prop:sq-sim} tells us we can efficiently simulate a
statistical query algorithm for a concept class $\C$ under
\new{a samplable} distribution
$D$ if we have access to samples drawn from $D_{f, +}$
 (and a very accurate estimate of $\Pr_{x\sim D}[f(x)=1]$).
In our setting, we have that $D = \U_{g^{-1}(1)}$
\new{ where $g \in \C'$ is the function that is output by
$\A^{(\C,\C')}_\den$.}
\new{
So, the two issues we must handle are (i) obtaining samples from
$D$, and (ii) obtaining samples from $D_{f,+}.$
}

\nnew{

For (i), we note that, even though we do not have access to samples drawn {\em exactly} from $D$, it suffices
for our purposes to use a $\tau'$-sampler for $D$ for a sufficiently small $\tau'$.
To see this we use the following fact:

\begin{fact} \label{fact:approx-sample-D}
Let $D, D'$ be distributions over $\bn$ with $\dtv(D, D') \le \tau'$. Then for any bounded
function $\phi: \bn \to [-1,1]$ we have that $|\E_{x \sim D}[\phi(x)] - \E_{x \sim D'} [\phi(x)] | \le 2\tau'$.
\end{fact}

\begin{proof}
By definition we have that
\begin{eqnarray*}
\left| \E_{x \sim D}[\phi(x)] - \E_{x \sim D'} [\phi(x)] \right|
&=& \left| \littlesum_{x \in \bn} \left( D(x) - D'(x) \right) \phi(x) \right| \\
&\le&  \littlesum_{x \in \bn} {\left| \left( D(x) - D'(x) \right) \right| \left| \phi(x) \right| }\\
&\le& \mathrm{max}_{x \in \bn} \left| \phi(x) \right| \cdot \littlesum_{x \in \bn} {\left| D(x) - D'(x) \right|}\\
&\le& 1 \cdot \|D - D' \|_1 \\
&=& 2 \dtv(D, D') \\
&\le& 2\tau'
\end{eqnarray*}
as desired.
\end{proof}

The above fact implies that the statement of Proposition~\ref{prop:sq-sim} continuous to hold with the same parameters
if instead of a $0$-sampler for $D$ we have access to a $\tau'$-sampler for $D$, for $\tau' = \tau/8$.
The only difference is that  in Step~1 of the subroutine {\tt Simulate-STAT}$^D_f$ we empirically estimate the expectation
$\E_{x \sim D'}[\chi(x, -1)]$ up to an additive $\pm \tau/4$. By Fact~\ref{fact:approx-sample-D}, this will be a
$\pm (\tau/4+2\tau') = \pm \tau/2$ accurate estimate for the $\E_{x \sim D}[\chi(x, -1)]$. That is, we have:

\begin{corollary} \label{cor:sq-sim}
The statement of Proposition~\ref{prop:sq-sim} continues to hold with the same parameters if instead of a $0$-sampler
for $D$ we have access to a $\tau' = \tau/8$-sampler for $D$.
\end{corollary}

}

\new{
For (ii), even though we do not have access to the
distribution $D=\U_{g^{-1}(1)}$ directly,}
we note below that we can efficiently \new{sample from $D_{f,+}$}
using samples from $\U_{f^{-1}(1)}$ together with
\new{evaluations of $g$ (recall again that $g$ is provided as the
output of the densifier).}

\begin{claim}\label{claim:sim-sample}
Let $g:\bn \to \bits$ be a \new{$t_g(n)$ time computable function} such that
$\Pr_{x \sim \U_{f^{-1}(1)}} \left[ g(x)= 1 \right] \new{\geq \eps'}.$
There is an efficient subroutine that \new{is given $\eps'$ \newer{and a circuit to compute $g$} as input,}
uses $m = O((1/\new{\eps'}) \log(1/\delta))$
samples from $\U_{f^{-1}(1)}$, runs in time $O(m \cdot t_g(n))$,
and with probability $1-\delta$ outputs a sample $x$ such that $x \sim D_{f, +}$, where $D = \U_{g^{-1}(1)}$.
\end{claim}
\begin{proof}
To simulate a sample from $D_{f, +}$ we simply
draw samples from $\U_{f^{-1}(1)}$ until we obtain a sample $x$ with  $g(x)=1$.
The following pseudocode makes this precise:
\begin{framed}
\noindent Subroutine   {\tt Simulate-sample}$^{D_{f, +}} (\new{\U_{f^{-1}(1)}, g, \eps', \delta})$:

\smallskip

\noindent {\bf Input:} Independent samples from $\U_{f^{-1}(1)}$,
a circuit computing $g$, \new{a value $\eps' > 0$ such that $\eps' \le \Pr_{x \sim \U_{f^{-1}(1)}} \left[ g(x) = 1 \right]$ and confidence parameter
$\delta$.}

\noindent {\bf Output:} A point $x \in \bn $ that with probability $1-\delta$  satisfies $x \sim D_{f, +}$.

\begin{enumerate}

\item Repeat the following at most $m = \Theta \left( (1/\new{\eps'})
\log(1/\delta) \right)$ times:
\begin{enumerate}
\item Draw a sample $x \sim \U_{f^{-1}(1)}$.

\item If the circuit for $g$ evaluates to 1 on input $x$ then output $x$.
\end{enumerate}
\item If no point $x$ with $g(x)=1$ has been obtained,
halt and output ``failure.''
\end{enumerate}
\end{framed}

Since $\Pr_{x \sim \U_{f^{-1}(1)}} [g(x) = 1] \new{\geq \eps'}$,
after repeating this process $m = \Omega \left( (1/\new{\eps'}) \log(1/\delta)\right)$ times, we will obtain a satisfying assignment to $g$
with probability at least $1-\delta$.
It is clear that such a sample $x$ is distributed according to $D_{f, +}$.
For each sample we need to evaluate $g$ once, hence the running time follows.
\end{proof}

\medskip
\noindent {\bf Getting a good estimate $\widetilde{b_f}$ of
$\Pr_{x \sim D}[f(x)=1]$.}
The simulations presented above require an additively accurate
estimate $\widetilde{b_f}$ of $\Pr_{x \sim D}[f(x)=1]$.  We now show
that in our context, such an estimate can be easily obtained if we have
access to a good estimate $\widehat{p}$ of $p=\Pr_{x \in \U_n}[f(x)=1]$, using
 the fact that we have an efficient approximate counting
 algorithm for $\C'$ and that $D \equiv \U_{g^{-1}(1)}$ where $g \in \C'$.

\begin{claim} \label{claim:estimate-bias}
Let $g:\bn \to \bits$, $g \in \C'$ be a \new{$t_g(n)$ time computable function}, satisfying
$\Pr_{x \sim \U_{g^{-1}(1)}}[f(x)=1] \geq \gamma'$ and
$\Pr_{x \sim \U_{f^{-1}(1)}}[g(x)=1] \geq 1 - \eps'$.
Let ${\mathcal{A}}^{\C'}_{\co}$ be an $(\eps,\delta)$-approximate counting algorithm
for $\C'$ running in time $T_{\co}(n,1/\eps,1/\delta)$.
There is a procedure {\tt Estimate-Bias} with the following behavior:
{\tt Estimate-Bias} takes as input a value $0 < \widehat{p} \leq 1$, a parameter $\tau'>0$, a confidence
parameter $\delta'$, and \new{a representation of $g \in \C'.$}
{\tt Estimate-Bias} runs in time $O(t_g \cdot T_\co(n,2/\tau',1/\delta'))$
and satisfies the following:  if $p \eqdef \Pr_{x \sim \U_n}[f(x)=1] <
\widehat{p} \leq (1+\eps')p$, then with probability
$1-\delta'$ {\tt Estimate-Bias} outputs a value $\widetilde{b_f}$ such that
$| \widetilde{b_f} -\Pr_{x \sim D}[f(x)=1]| \leq \tau'$.
\end{claim}

\begin{proof}
The procedure {\tt Estimate-Bias} is very simple.
It runs ${\mathcal{A}}^{\C'}_{\co}$ on inputs $\eps^\star = \tau'/2,\delta'$, using the \new{representation
for $g \in \C'$}.  Let $p_g$ be the value returned by the
approximate counter; {\tt Estimate-Bias} returns $\widehat{p}/p_g.$

The claimed running time bound is obvious.
To see that the procedure is correct, first observe that by Definition~\ref{def:approx-count},
with probability $1-\delta'$ we have that
\[
{\frac {|g^{-1}(1)|}{2^n}} \cdot
{\frac 1 {1+\eps^\star}} \leq p_g \leq {\frac {|g^{-1}(1)|} {2^n}} \cdot (1+\eps^\star).
\]
For the rest of the argument we assume that the above inequality indeed holds.
Let $A$ denote $|g^{-1}(1)|$, let $B$ denote $|f^{-1}(1) \cap g^{-1}(1)|$, and
let $C$ denote $|f^{-1}(1) \setminus g^{-1}(1)|$, so the true value
$\Pr_{x \sim D}[f(x)=1]$ equals ${\frac B A}$ and the above inequality can be rephrased as
\[
{\frac A {1+\eps^\star}} \leq p_g \cdot 2^n \leq A \cdot (1+\eps^\star).
\]
By our assumption on $\widehat{p}$ we have that
\[
B+C \leq \widehat{p} \cdot 2^n \leq (1+\eps')(B+C);
\]
since $\Pr_{x \sim \U_{f^{-1}(1)}}[g(x)=1] \geq 1-\eps'$ we have
\[
{\frac C {B+C}} \leq \eps' \quad \quad \text{(i.e.,~}
C \leq {\frac {\eps'}{1-\eps'}} \cdot B \text{ )};
\]
and since $\Pr_{x \sim \U_{g^{-1}(1)}}[f(x)=1] \geq \gamma'$ we have
\[
{\frac B A} \geq \gamma'.
\]
Combining these inequalities we get
\[
{\frac 1 {1+\eps^\star}} \cdot {\frac B A}
\leq
{\frac 1 {1+\eps^\star}} \cdot {\frac {B+C} A}
\leq {\frac {\widehat{p}}{p_g}} \leq {\frac B A} \cdot (1+\eps')(1+\eps^\star) \left( 1 + {\frac {\eps'}
{1-\eps'}}\right) = {\frac B A} \cdot (1 + \eps^\star)
\]
Hence \[
\left|{\frac B A} - {\frac {\widehat{p}}{p_g}}\right| \leq {\frac B A} \left(1 + \eps^\star - {\frac 1 {1+\eps^\star}}\right) \leq {\frac {2 \eps^\star}{1+\eps^\star}} \leq 2 \eps^\star,
\]
where we have used $B \leq A$.  Recalling that $\eps^\star = \tau'/2$, the lemma is proved.
\end{proof}


\subsection{An algorithm that succeeds given the (approximate)
bias of $f.$}
\label{ssec:algo-known-bias}

In this section, we present an algorithm $\A'^\C_\inv (\eps,\delta,
\widehat{p} )$
which, in addition to samples from $\U_{f^{-1}(1)}$, takes as input
parameters $\eps,\delta,\widehat{p}$.
The algorithm succeeds in outputting a hypothesis distribution $D_f$
satisfying $\dtv(D_f, \U_{f^{-1}(1)})
\le \eps$ if the input parameter $\widehat{p}$ is a
multiplicatively accurate approximation to $\Pr_{x \sim \U_n} [f(x)=1]$.
\ignore{\new{ Note that in the pseudocode for algorithm $\A'^\C_\inv
(\eps,\delta, \widehat{p})$ (given below) ``$D$'' denotes the distribution
$D = \U_{g^{-1}(1)}$,} where $g$ is the output of the densifier called in
Step~1 of the algorithm.}
The algorithm follows the three high-level steps previously outlined and uses the subroutines
of the previous subsection to simulate the statistical query algorithm.
Detailed pseudocode follows:

\begin{framed}
\noindent Algorithm  $\A'^\C_\inv ( \U_{f^{-1}(1)}, \eps, \delta, \widehat{p})$:

\smallskip

\noindent {\bf Input:} Independent samples from $\U_{f^{-1}(1)}$,
accuracy and confidence parameters $\eps,\delta$,
and a value $1/2^n < \widehat{p} \leq 1.$

\noindent {\bf Output:}
If \new{$\Pr_{x \sim \U_n}[f(x)=1] \leq \widehat{p}
< (1+\eps)\Pr_{x \sim \U_n}[f(x)=1]$,}
\ignore{OLD:   $|\widetilde{b_f} - \Pr_{x \sim D}[f(x)=1]| \le \tau_2/2$}
 with probability $1-\delta$ outputs an $\eps$-sampler
\new{$C_f$} for $\U_{f^{-1}(1)}$ .\\

\begin{enumerate}

\item
{\bf [Run the densifier to obtain $g$]}

Fix $\eps_1 \eqdef \eps/\new{6}$ and
\new{$\gamma \eqdef \gamma (n,1/\eps_1,3/\delta)$}.
Run the $\gamma$-densifier $\A^{(\C, \C')}_{\den}(\eps_1, \delta/3,\new{\widehat{p}})$
using random samples from $\U_{f^{-1}(1)}.$
Let $g \in \C'$ be its output.

\item {\bf [Run the $\sq$-learner, using the approximate uniform generator
for $g$, to obtain hypothesis $h$]}

\begin{enumerate}

\item
Fix $\eps_2 \eqdef \eps \gamma/\new{7}$,
$\tau_2 \eqdef \tau(n, 1/\eps_2)$ and
$m \eqdef \Theta \left( (1/\tau_2^2)
\cdot \log(T_1/\delta) \cdot T_1  \right)$,
where $T_1 = t_1(n, 1/\eps_2, 12/\delta)$.

\item Run the generator $\A^{\C'}_{\gen}(g, \tau_2/8,
\new{\delta/(12m)} )$ $m$ times and let $S_D \subseteq \bn$ be the multiset of samples obtained.

\item Run   {\tt Simulate-sample}$^{D_{f, +}} (\U_{f^{-1}(1)}, g, \gamma,  \delta/(12m))$ $m$ times and let $S_{D_{f, +}} \subseteq \bn$
         be the multiset of samples obtained.

\item \newer{Run {\tt Estimate-Bias} with parameters $\widehat{p}$, $\tau' = \tau_2/2$,
$\delta' = \delta/12$ , using \new{the representation} for $g \in \C'$, and let $\widetilde{b_f}$ be the value it
returns.}

\item Run $\A_{\sqs} (S_D, S_{D_{f, +}}, \eps_2, \widetilde{b_f}, \delta/12)$.
          Let $h: \bn \to \bits$ be the output hypothesis.
\end{enumerate}

\item
{\bf [Output the sampler which does rejection sampling according to $h$ on draws from
the approximate uniform generator for $g$]}

Output the sampler $C_f$ which works as follows:

\framebox{
\medskip \noindent \begin{minipage}{14cm}

For $i=1$ to $t = \Theta\left( (1/\gamma) \log(1/(\delta\eps) \right)$ do:
\begin{enumerate}

\item Set $\eps_3 \eqdef \new{\eps \gamma/48000}.$

\item Run the generator $\A^{\C'}_{\gen}(g, \eps_3, \newer{\delta \eps /(12t)}
)$ and let $x^{(i)}$
\ignore{$\sim D'$} be its output.

\item If $h(x^{(i)})  =1$, output $x^{(i)}$.
\end{enumerate}
If no $x^{(i)}$ with $h(x^{(i)})  =1$ has been obtained, output
\newer{the default element $\bot$.}

\end{minipage}
}

\newer{
Let $\hat{D}$ denote the distribution over $\{-1,1\}^n \cup \{\bot\}$
for which $C_f$ is a 0-sampler, and let $\hat{D}'$ denote
the conditional distribution of $\hat{D}$
restricted to $\bits^n$ (i.e., excluding $\bot$).
}

\end{enumerate}
\end{framed}

We note that by inspection of the code for $C_f$, we have
that the distribution $\hat{D}'$ is identical to
$(D_{g,\eps_3})_{h^{-1}(1)}$, where $D_{g,\eps_3}$ is the
distribution corresponding to the output
of the approximate uniform generator when called on function $g$
and error parameter $\eps_3$ (see Definition~\ref{def:approx-unif-gen})
and $(D_{g,\eps_3})_{h^{-1}(1)}$ is $D_{g,\eps_3}$ conditioned on
$h^{-1}(1)$.

We have the following:

\begin{theorem} \label{thm:known-bias}
Let $p \eqdef \Pr_{x \in \U_n}[f(x)=1]$.
Algorithm  $\A'^\C_\inv (\eps,\delta,\widehat{p})$
has the following behavior:
If $p \leq \widehat{p} < (1+\eps)p$,
then with probability $1-\delta$
the following both hold:

\begin{itemize}

\item [(i)] the output $C_f$ is a sampler for a distribution $\hat{D}$
such that $\dtv(\hat{D}, \U_{f^{-1}(1)}) \le \eps$; and

\item [(ii)] \new{the functions $h,g$ satisfy
$|h^{-1}(1) \cap g^{-1}(1)|/|g^{-1}(1)| \geq \gamma/2.$
}

\end{itemize}
\noindent 
\new{
The running time of$\A'^\C_\inv$ is 
polynomial in
$T_{\den}(n, 1/\eps, 1/\delta)$,
$T_{\gen}(n, 1/\eps,\newer{1/\delta})$,
$T_{\co}(n,1/\eps,1/\delta)$,
$t_1(n,1/\eps,1/\delta)$,
$t_2(n)$,
$1/\tau(n,1/\eps)$,
and
$1/\gamma(n,1/\eps,1/\delta)$.
}
\end{theorem}
\begin{proof}
We give an intuitive explanation of the pseudocode in tandem with a
proof of correctness.
\new{We argue} that Steps~1-3 of the algorithm implement the corresponding
steps of our
high-level description and that the algorithm succeeds with confidence probability $1-\delta$.

We assume throughout the argument that indeed $\widehat{p}$
lies in $[p,(1+\eps)p).$  Given this, by Definition~\ref{def:dense}
with probability $1-\delta/3$ the function $g$
satisfies properties (a) and (b)
of Definition~\ref{def:dense}, i.e.,
$\Pr_{x \sim \U_{f^{-1}(1)}}[g(x)=1] \geq 1 - \eps_1$
and
$\Pr_{x \sim \U_{g^{-1}(1)}}[f(x)=1] \geq \gamma.$
We condition on this event
\new{(which we denote $E_1$)} going forth.

We now argue that Step~2 simulates the $\sq$ learning algorithm $\A^\C_{\sq}$
to learn the function $f \in \C$ under distribution $D \equiv \U_{g^{-1}(1)}$
to accuracy $\eps_2$ with confidence $1-\delta/3$.
Note that the goal of Step (b) is to obtain $m$ samples from a
distribution $D''$
(the distribution ``$D_{g,\tau_2/8}$'' of
Definition~\ref{def:approx-unif-gen})
such that $\dtv(D'', D) \le \tau_2/8$.
To achieve this, we call the approximate uniform generator for $g$ a total
of $m$ times
\newer{with failure probability $\delta/(12m)$ for each call
(i.e., each call returns $\bot$ with probability at most $\delta/(12m)$).
By a union bound, with failure probability
at most $\delta/12$, all calls to the generator
are successful and we obtain a set $S_D$ of $m$
independent samples from $D''$.} Similarly, the goal
of Step (c) is to obtain $m$ samples from $D_{f, +}$ and to achieve it we call
the subroutine {\tt Simulate-sample}$^{D_{f, +}}$
a total of $m$ times with failure probability $\delta/(12m)$ each.
By Claim~\ref{claim:sim-sample} and a union bound, with
failure probability at most $\delta/12$, this step is successful, i.e.,
it gives a set $S_{D_{f, +}}$ of $m$ independent samples from $D_{f,+}$.
\new{The goal of Step (d) is to obtain a value $\widetilde{b_f}$ satisfying
$|\widetilde{b_f} - \Pr_{x \sim D}[f(x)=1]| \leq \tau_2/2$; by Claim~\ref{claim:estimate-bias},
with failure probability at most $\delta/12$ the value $\widetilde{b_f}$ obtained in this
step is as desired.}  Finally, Step (e) applies the simulation algorithm  $\A_{\sqs}$ using the samples $S_D$ and $S_{D_{f, +}}$ and the estimate $\widetilde{b_f}$ of $\Pr_{x \sim D}[f(x)=1]$
obtained in the previous
steps.
Conditioning on Steps (b), (c) and (d) being successful
Corollary~\ref{cor:sq-sim} implies that
\ignore{if $|\widetilde{b_f} - \Pr_{x \sim D}[f(x)=1]| \le \tau_2/2$
then }Step (e) is successful with probability $1-\delta/12$, i.e., it
outputs a hypothesis $h$ that satisfies $\Pr_{x \sim D} [f(x) \ne h(x)]
\le \eps_2$.
A union bound over Steps (c), (d) and (e) completes the analysis of Step~2.
\new{For future reference, we let $E_2$ denote the event that the
hypothesis $h$ constructed in Step~2(e) has
$\Pr_{x \sim D}[f(x) \neq h(x)] \leq \eps_2$ (so we have
that $E_2$ holds with probability at least $1-\delta/3$; we additionally
condition on this event going forth)}.
We observe that since (as we have just shown)
$\Pr_{x \sim \U_{g^{-1}(1)}}[f(x) \neq h(x)] \leq \eps_2$ and
$\Pr_{x \sim \U_{g^{-1}(1)}}[f(x) =1 ] \geq \gamma$,
we have $\Pr_{x \sim \U_{g^{-1}(1)}}[h(x)=1] \geq \gamma - \eps_2
\geq \gamma/2$, which gives item (ii) of the theorem; so it remains to
establish item (i) and the claimed running time bound.

To establish (i), we need to prove that the output distribution
$\hat{D}$ of the sampler $C_f$
is $\eps$-close in total variation distance to $\U_{f^{-1}(1)}$.
This sampler attempts to draws $t$ samples from a distribution $D'$
such that $\dtv(D', D) \le \eps_3$ (this is the distribution
``$D_{g,\eps_3}$'' in the notation of Definition~\ref{def:approx-unif-gen})
and it outputs
one of these samples that satisfies $h$ (unless none of these samples
satisfies $h$, in which case it outputs
\newer{a default element $\bot$}).
The desired variation distance bound follows from the next
lemma for our choice of parameters:

\begin{lemma} \label{lem:sixsums}
Let \new{$\hat{D}$} be the output distribution of $\A'^\C_\inv(
\U_{f^{-1}(1)}, \eps, \delta, \widehat{p})$.
If \new{$\Pr_{x \sim \U_n}[f(x)=1] \leq \widehat{p}
\leq (1+\eps)\Pr_{x \sim \U_n}[f(x)=1]$}, then
\new{conditioned on Events $E_1$ and $E_2$, we have}
\begin{eqnarray*} \dtv(\hat{D}, \U_{f^{-1}(1)})
&\le&
{\frac {\eps} 6} +
{\frac {\eps} 6} +
{\frac {4 \eps_3} \gamma} +
\eps_1 + \frac{\eps_2}{2\gamma} + \frac {\eps_2}{\gamma-\eps_2}\\
&\leq&
{\frac {\eps} 6} +
{\frac \eps 6} + {\frac \eps {12000}} + {\frac \eps 6} + {\frac \eps {14}}
+ {\frac \eps 6} < \eps.
\end{eqnarray*}
\end{lemma}

\begin{proof}
Consider the distribution $D' = D_{g,\eps_3}$
\nnew{(see Definition~\ref{def:approx-unif-gen})} produced by the approximate uniform generator in Step~3 of the algorithm.
Let $D'|_{h^{-1}(1)}$ denote distribution $D'$ restricted to $h^{-1}(1).$
Let $S$ denote the set $g^{-1}(1) \cap h^{-1}(1)$.
The lemma is an immediate consequence of Claims~\ref{claim:c0},
\ref{claim:c1},
\ref{claim:c2} and~\ref{claim:c3} below
using the triangle inequality (everything below is conditioned
on $E_1$ and $E_2$).
\end{proof}

\begin{claim} \label{claim:c0}
$\dtv(\hat{D},\hat{D}') \leq \eps/6.$
\end{claim}
\begin{proof}
Recall that $\hat{D}'$ is simply $\hat{D}$ conditioned on
not outputting $\bot$.

We first claim
that with probability at
least $1 - \delta \eps/12$ all
$t$ points drawn in Step~3 of the code for $C_f$ are distributed according to
the distribution $D' = D_{g,\eps_3}$ over $g^{-1}(1)$.
Each of the $t$ calls to the approximate
uniform generator has failure probability
$\delta\eps/(12t)$ (of outputting $\bot$ rather than a point
distributed according to $D'$)
so by a union bound no calls fail with probability \new{at least $1-\delta
\eps/12$,}
and thus with probability at least $1-\delta\eps/12$ indeed all $t$ samples
are independently drawn from such a distribution $D'$.

Conditioned on this,
we claim that a satisfying assignment for $h$
is obtained within the $t$ samples with probability at least $1-\delta\eps/12$.
This can be shown as follows:

\begin{claim} \label{claim:goodxi}
Let $h:\bn \to \bits$ be the hypothesis output by $\A^{\C}_{\sqs}$. We have
$$\Pr_{x \sim D'}[h(x)=1] \geq \gamma/4.$$
\end{claim}
\begin{proof}
First recall that, by property (b) in the definition of the densifier
(Definition~\ref{def:dense}), we have
$\Pr_{x \sim D} [f(x)=1] \ge \gamma$. Since $\dtv(D', D) \leq \eps_3$, by definition we get
$$\Pr_{x \sim D'} [f(x)=1] \ge \Pr_{x \sim D} [f(x)=1]  -\eps_3 \ge \gamma - \eps_3  \geq  3\gamma/4.$$
Now by the guarantee of Step~2 we have that $\Pr_{x \sim D} [ f(x) \ne h(x)] \le \eps_2$.
Combined with the fact that $\dtv(D', D) \leq \eps_3$, this implies that $$\Pr_{x \sim D'}  [ f(x) \ne h(x)] \le \eps_2 +\eps_3 \le \gamma/2.$$
Therefore, we conclude that
$$\Pr_{x \sim D'} [ h(x) = 1] \geq \Pr_{x \sim D'} [ f(x) = 1]  - \Pr_{x \sim D'} [ f(x) \ne h(x)]  \ge 3\gamma/4 - \gamma/2 \ge \gamma/4$$
as desired.
\end{proof}

Hence, for an appropriate constant in the big-Theta
specifying $t$, with probability at least $1-\delta \eps/12
> 1 - \delta/12$ some $x^{(i)}$ is a satisfying
assignment of $h$.
that with probability at least $1 - \delta \eps/12$ some $x^{(i)}$,
$i \in [t]$, has $h(x)=1.$  Thus with overall failure
probability at most $\delta \eps / 6$ a draw from $\hat{D}'$ is not
$\bot$, and consequently we have $\dtv(\hat{D},\hat{D}') \leq
\delta \eps/6 \leq \eps/6.$
\end{proof}

\begin{claim} \label{claim:c1}
$\dtv(\hat{D}',D'|_{h^{-1}(1)}) \leq \eps/6.$
\end{claim}
\begin{proof}
\new{The probability that any of the $t$ points $x^{(1)},\dots,x^{(t)}$
is not drawn from $D'$ is at most $t \cdot \delta \eps/(12t) < \eps/12.$
Assuming that this does not happen,}
the probability that no
$x^{(i)}$ lies in $h^{-1}(1)$ is at most $(1-\gamma/4)^t < \delta \eps/12
< \eps/12$
by Claim~\ref{claim:goodxi}.  Assuming this does not happen, the output
of a draw from $\hat{D}$ is distributed identically according to
$D'|_{h^{-1}(1)}$.  Consequently
we have that
$\dtv(\hat{D},D'|_{h^{-1}(1)}) \leq \eps/6 $ as claimed.
\end{proof}

\begin{claim} \label{claim:c2}
$\dtv(D'|_{h^{-1}(1)}, \calU_{S}) \leq 4\eps_3/\gamma.$
\end{claim}
\begin{proof}
The definition of an approximate uniform generator gives us that
$\dtv(D',\calU_{g^{-1}(1)} ) \leq \eps_3$, and Claim~\ref{claim:goodxi} gives
that $\Pr_{x \sim D'}[h(x)=1] \geq \gamma/4.$  We now recall the
fact that for any two distributions $D_1,D_2$ and any event $E$, writing
$D_i|_E$ to denote distribution $D_i$ conditioned on event $E$, we have
\[
\dtv(D_1|_E,D_2|_E) \leq {\frac {\dtv(D_1,D_2)}{D_1(E)}}.
\]
The claim follows since $\calU_{g^{-1}(1)}|_{h^{-1}(1)}$ is
equivalent to $\calU_S.$
\end{proof}

\begin{claim} \label{claim:c3}
$\dtv(\U_{S},\U_{f^{-1}(1)}) \leq
\eps_1 + \frac{\eps_2}{2\gamma} + \frac {\eps_2}{\gamma-\eps_2}
$.
\end{claim}
\begin{proof}
The proof requires a careful combination of the properties of the function $g$ constructed by the densifier
and the guarantee of the $\sq$ algorithm.
Recall that $S = g^{-1}(1) \cap h^{-1}(1)$. We consider the set $S' = g^{-1}(1) \cap f^{-1}(1)$. By the triangle inequality,
we can bound the desired variation distance as follows:
\begin{equation} \label{eqn:triangle}
\dtv( \U_{S},\U_{f^{-1}(1)} ) \le \dtv( \U_{f^{-1}(1)}, \U_{S'}) + \dtv (\U_{S'}, \U_S).
\end{equation}
We will bound from above each term of the RHS in turn.
To proceed we need an expression for the total variation distance between the uniform distribution on two finite sets.
The following fact is obtained by straightforward calculation:

\begin{fact} \label{fact:dtv-uniform}
Let $A, B$ be subsets of a finite set $\W$ and $\U_A$, $\U_B$ be the uniform distributions on $A$, $B$ respectively.
Then,
\begin{equation} \label{eqn:dtv-uniform}
\dtv(\U_A, \U_B ) = (1/2) \cdot \frac{|A\cap \overline{B}|}{|A|} + (1/2) \cdot \frac{|B\cap \overline{A}|}{|B|} + (1/2) \cdot |A\cap B| \cdot \left| \frac{1}{|A|}  -  \frac{1}{|B|} \right|.
\end{equation}
\end{fact}

\smallskip

\noindent To bound the first term of the RHS of (\ref{eqn:triangle}) we apply the above fact for $A = f^{-1}(1)$ and $B = S'$.
Note that in this case $B \subseteq A$, hence the second term of (\ref{eqn:dtv-uniform}) is zero. Regarding the first term, note that
$$\frac{|A\cap \overline{B}|}{|A|} = \frac{|f^{-1}(1) \cap \overline{g^{-1}(1)}|}{|f^{-1}(1)|} \leq \eps_1,$$
where the inequality follows from Property (a) of the densifier definition. Similarly, for the third term we can write
$$|A\cap B| \cdot \left| \frac{1}{|A|}  -  \frac{1}{|B|} \right|  = |B| \cdot \left| \frac{1}{|A|}  -  \frac{1}{|B|} \right| = 1-\frac{|B|}{|A|} = 1- \frac{|f^{-1}(1) \cap g^{-1}(1)|}{|f^{-1}(1)|}  \le \eps_1,$$
where the inequality also follows from Property (a) of the densifier definition. We therefore conclude that $\dtv( \U_{f^{-1}(1)}, \U_{S'}) \le \eps_1.$

We now proceed to bound the second term of the RHS of (\ref{eqn:triangle}) by applying Fact~\ref{fact:dtv-uniform}
for $A = S'$ and $B = S$. It turns out that bounding the individual terms of (\ref{eqn:dtv-uniform}) is trickier in this case.
For the first term we have:
$$  \frac{|A\cap \overline{B}|}{|A|}  =  \frac{|f^{-1}(1) \cap g^{-1}(1) \cap \overline{h^{-1}(1)} |}{|f^{-1}(1) \cap g^{-1}(1)|} =
 \frac{|f^{-1}(1) \cap g^{-1}(1) \cap \overline{h^{-1}(1)} |}{|g^{-1}(1)|} \cdot  \frac{|g^{-1}(1)|}{|f^{-1}(1) \cap g^{-1}(1)|}  \le \frac{\eps_2}{\gamma},$$
where the last inequality follows from the guarantee of the $\sq$ learning algorithm and Property (b) of the densifier definition.
For the second term we have
\[  \frac{|B\cap \overline{A}|}{|B|}  = \frac{|\overline{f^{-1}(1)} \cap g^{-1}(1) \cap h^{-1}(1)|}{|g^{-1}(1) \cap h^{-1}(1)|}. \]
To analyze this term we recall that by the guarantee of the $\sq$ algorithm it follows that the numerator satisfies
$$|\overline{f^{-1}(1)} \cap g^{-1}(1) \cap h^{-1}(1)| \le \eps_2 \cdot |g^{-1}(1)|.$$
From the same guarantee we also get
$$|f^{-1}(1) \cap g^{-1}(1) \cap \overline{h^{-1}(1)}| \le \eps_2 \cdot |g^{-1}(1)|.$$
Now, Property (b) of the densifier definition gives
$|f^{-1}(1) \cap g^{-1}(1)| \ge \gamma \cdot |g^{-1}(1)|$.
Combing these two inequalities implies that
\[ | g^{-1}(1) \cap h^{-1}(1) |  \ge |f^{-1}(1) \cap g^{-1}(1) \cap h^{-1}(1)| \ge (\gamma-\eps_2) \cdot  |g^{-1}(1)|.\]
In conclusion, the second term is upper bounded by $(1/2) \cdot \frac{\eps_2}{\gamma-\eps_2}.$

For the third term, we can write
\[ |A\cap B| \cdot \left| \frac{1}{|A|}  -  \frac{1}{|B|} \right| = |f^{-1}(1) \cap g^{-1}(1) \cap h^{-1}(1)| \cdot
\left| \frac{1}{ |f^{-1}(1) \cap g^{-1}(1)|}  -  \frac{1}{ |g^{-1}(1) \cap h^{-1}(1)|} \right|.\]
To analyze these term we relate the cardinalities of these sets. In particular, we can write
\begin{eqnarray*}
|f^{-1}(1) \cap g^{-1}(1) | &=& |f^{-1}(1) \cap g^{-1}(1) \cap h^{-1}(1)| + |f^{-1}(1) \cap g^{-1}(1) \cap \overline{h^{-1}(1)|} \\
                                            &\le& |f^{-1}(1) \cap g^{-1}(1) \cap h^{-1}(1)| +  \eps_2 \cdot |g^{-1}(1)| \\
                                           &\le& |f^{-1}(1) \cap g^{-1}(1) \cap h^{-1}(1)| +  \frac{\eps_2}{\gamma} \cdot |f^{-1}(1) \cap g^{-1}(1)|
\end{eqnarray*}
where the last inequlity is Property (b) of the densifier defintion.
Therefore, we obtain
\[  (1-\frac{\eps_2}{\gamma}) \cdot   |f^{-1}(1) \cap g^{-1}(1) | \le  |f^{-1}(1) \cap g^{-1}(1) \cap h^{-1}(1)| \le |f^{-1}(1) \cap g^{-1}(1) | .   \]
Similarly, we have
\begin{eqnarray*}
|g^{-1}(1) \cap h^{-1}(1) | &=& |f^{-1}(1) \cap g^{-1}(1) \cap h^{-1}(1)| + |\overline{f^{-1}(1)} \cap g^{-1}(1) \cap h^{-1}(1)| \\
                                            &\le& |f^{-1}(1) \cap g^{-1}(1) \cap h^{-1}(1)| +  \eps_2 \cdot |g^{-1}(1)| \\
                                           &\le& |f^{-1}(1) \cap g^{-1}(1) \cap h^{-1}(1)| +  \frac{\eps_2}{\gamma-\eps_2} \cdot |g^{-1}(1) \cap h^{-1}(1)|
\end{eqnarray*}
and therefore
\[  (1-\frac{\eps_2}{\gamma-\eps_2}) \cdot   |g^{-1}(1) \cap h^{-1}(1) | \le  |f^{-1}(1) \cap g^{-1}(1) \cap h^{-1}(1)| \le |g^{-1}(1) \cap h^{-1}(1) | .   \]
The above imply that the third term is bounded by $(1/2) \cdot \frac{\eps_2}{\gamma- \eps_2}.$
This completes the proof of the claim.
\end{proof}

With Lemma~\ref{lem:sixsums} established,
to finish the proof of Theorem~\ref{thm:known-bias} it remains only to
establish the claimed running time bound.  This follows from
a straightforward (but somewhat tedious) verification, using the
running time bounds established in Lemma~\ref{lem:sq-sim},
Proposition~\ref{prop:sq-sim}, Corollary~\ref{cor:sq-sim},
Claim~\ref{claim:sim-sample} and
Claim~\ref{claim:estimate-bias}.
\end{proof}

\ignore{
\new{
Need to comment on sample complexity and running time.
The only non-trivial part regarding sample complexity is Step 2, i.e., the implementation of the SQ algorithm.
We claim that...

We apply Claim~\ref{claim:sim-sample} for the function $g$ computed in Step 1. It is clear that $t_g(n)$ is
is upper bounded by the running time of the densifier. Moreover, the parameter $\alpha = 1-\eps$.
Note that we want all the samples produced by this simulation process to be accurate with probability $1-\gamma$. By a union bound, it suffices
to set the confidence parameter in the above lemma to $\gamma/m$, where $m$ is the number of samples in Proposition~\ref{prop:sq-sim}
necessary to simulate $\A_{\sq}$. \new{(All this needs to be incorporated in the pseudocode. In particular, Step~3 of the main algorithm's pseudocode
needs to make a call to $\A_\sqs$, which in turns makes a call to  {\tt Simulate-sample}$^{D_{f, +}}$ for each sample it draws.)}
}
}

\subsection{Getting from $\A'^\C_{\inv}$ to $\A^\C_{\inv}$:
An approximate evaluation oracle.} \label{ssec:approx-eval-oracle}

Recall that the algorithm $\A'^\C_{\inv}$ from the previous subsection
is only guaranteed (with high probability) to output a sampler for a hypothesis
distribution $\hat{D}$ that is statistically close to the
target distribution $\U_{f^{-1}(1)}$ if it is given an input
parameter $\widehat{p}$ satisfying
$p \leq \widehat{p} < (1+\eps)p$, where $p \eqdef \Pr_{x \in \U_n}[f(x)=1].$
Given this, a natural idea is to run $\A'^\C_{\inv}$ a total
of $k=O(n/\eps)$ times, using ``guesses'' for $\widehat{p}$
that increase multiplicatively as powers of $1+\eps$, starting
at $1/2^n$ (the smallest possible value) and going up to 1.
This yields
hypothesis distributions $\hat{D}_1,\dots,\hat{D}_k$
where $\hat{D}_i$ is the distribution obtained by setting
$\widehat{p}$ to $\widehat{p}_i \eqdef (1+\eps)^{i-1}/2^n.$
With such distributions in hand, an obvious approach is to use
the ``hypothesis testing'' machinery of Section~\ref{sec:prelims} to
identify a high-accuracy $\hat{D}_i$ from this collection.

This is indeed the path we follow, but some care is needed to
make the approach go through.  Recall that as described in
Proposition~\ref{prop:log-cover-size}, the hypothesis testing algorithm
requires the following:

\begin{enumerate}

\item independent samples from the target distribution $\U_{f^{-1}(1)}$
(this is not a problem since such samples are available in our framework);

\item independent samples from $\hat{D}_i$ for each $i$ (also not
a problem since the $i$-th run of algorithm $\A'^\C_{\inv}$
outputs a sampler for distribution $\hat{D}_i$; and

\item a $(1+O(\eps))$-approximate evaluation oracle $\eval_{\hat{D}_i}$
for each distribution $\hat{D}_i.$

\end{enumerate}

In this subsection we show how to construct item (3) above, the approximate
evaluation oracle.  In more detail, we first describe a randomized
procedure {\tt Check} which is applied to the output of each execution
of $\A'^\C_{\inv}$ (across all $k$ different settings of the input
parameter $\widehat{p}_i$).  We show that
with high probability the ``right'' value $\widehat{p}_{i^\ast}$
(the one which satisfies $p \leq \widehat{p}_{i^\ast} < (1+\eps)p$)
will pass the procedure {\tt Check}.
Then we show that for each value $\widehat{p}_{i^\ast}$ that
passed the check a simple deterministic algorithm gives the
desired approximate evaluation oracle for $\hat{D}_i$.

We proceed to describe the {\tt Check} procedure
and characterize its performance.

\begin{framed}
\noindent Algorithm  {\tt Check}$(g,h,\delta',\eps):$

\smallskip

\noindent {\bf Input:}  functions $g$ and $h$ as described in
Lemma~\ref{lem:check}, a confidence
parameter $\delta'$, and an accuracy parameter $\eps$

\noindent {\bf Output:}
If $|h^{-1}(1) \cap g^{-1}(1)|/|g^{-1}(1)| \geq \gamma/2$,
with probability $1-\delta'$
outputs a pair $(\alpha,\kappa)$ such that
$|\alpha -
|h^{-1}(1) \cap g^{-1}(1)|/|g^{-1}(1)|
| \leq \mu \cdot
|h^{-1}(1) \cap g^{-1}(1)|/|g^{-1}(1)|
$
and
${\frac {|g^{-1}(1)|}{1+\tau}} \leq \kappa \leq
(1 + \tau)|g^{-1}(1)|,$
where $\mu = \tau = \eps/40000.$

\begin{enumerate}

\item Sample $m = O(\log(2/\delta')/(\gamma \mu^2))$
points $x^1,\dots,x^m$ from $\A^{\C'}_\gen(g,\gamma/4,\delta'/(2m))$.
If any $x^j = \bot$ halt and output ``failure.''

\item Let $\alpha$ be $(1/m)$ times the number of points $x^j$
that have $h(x)=1.$

\item Call $\A^{\C'}_\co(\tau,\delta'/2)$ on $g$
and set $\kappa$ to $2^n$ times the value it returns.

\end{enumerate}
\end{framed}

\begin{lemma}
\label{lem:check}
Fix $i \in [k]$.  Consider a sequence of $k$ runs of
$\A'^\C_{\inv}$ where in the $i$-th run it is
given $\widehat{p}_i \eqdef (1+\eps)^{i-1}/2^n$
as its input parameter.  Let $g_i$ be the function
in $\C'$ constructed by $\A'^\C_{\inv}$ in Step~1 of
its $i$-th run and $h_i$ be the hypothesis function
constructed by $\A'^\C_{\inv}$ in Step~2(e) of its $i$-th run.
Suppose {\tt Check} is given as input $g_i$, $h_i$,  a
confidence parameter $\delta'$, and an accuracy parameter $\eps'$.
Then it either outputs ``no'' or a pair $(\alpha_i,\kappa_i) \in [0,1]
\times [0,2^{n+1}]$, and satisfies the following
performance guarantee:
If $|h_i^{-1}(1) \cap g_i^{-1}(1)|/|g_i^{-1}(1)| \geq \gamma/2$
then with probability at least $1-\delta'$
{\tt Check} outputs a pair $(\alpha_i,\kappa_i)$ such that
\begin{equation} \label{eq:alphagood}
\left|\alpha_i - {\frac {|h_i^{-1}(1) \cap g_i^{-1}(1)|}{|g_i^{-1}(1)|}}
\right| \leq
\mu \cdot {\frac {|h_i^{-1}(1) \cap g_i^{-1}(1)|}{|g_i^{-1}(1)|}}
\end{equation}
and
\begin{equation} \label{eq:kappagood}
{\frac {|g_i^{-1}(1)|}{1+\tau}} \leq \kappa_i \leq
(1 + \tau)|g_i^{-1}(1)|,
\end{equation}
\new{where $\mu = \tau = \eps/40000.$}
\end{lemma}

\ignore{

OLD LEMMA:

\begin{lemma}
Fix $i \in [k]$.  Consider a sequence of $k$ runs of
$\A'^\C_{\inv}$ where in the $i$-th run it is
given $\widehat{p}_i \eqdef (1+\eps)^{i-1}/2^n$
as its input parameter.  Let $g_i$ be the function
in $\C'$ constructed by $\A'^\C_{\inv}$ in Step~1 of
its $i$-th run, $h_i$ be the hypothesis function
constructed by $\A'^\C_{\inv}$ in Step~2(e) of its $i$-th run,
and $(C_f)_i$ be the sampler for the distribution $\hat{D}_i$ constructed by
$\A'^\C_{\inv}$ in Step~3 of its $i$-th run.
Suppose {\tt Check} is given as input $g_i$, $h_i$, $(C_f)_i$, a
confidence parameter $\delta'$, and an accuracy parameter $\eps'$.
Then it either outputs ``no'' or a pair $(\alpha_i,\kappa_i) \in [0,1]
\times [0,2^{n+1}]$, and satisfies the following
performance guarantee:
If $|h_i^{-1}(1) \cap g_i^{-1}(1)|/|g_i^{-1}(1)| \geq \gamma/2$
then with probability at least $1-\delta'$
{\tt Check} outputs a pair $(\alpha_i,\kappa_i)$ such that
\[
\left|\alpha - {\frac {|h^{-1}(1) \cap g^{-1}(1)|}{|g^{-1}(1)|}}
\right| \leq
\mu \cdot {\frac {|h^{-1}(1) \cap g^{-1}(1)|}{|g^{-1}(1)|}}
\]
and
\[
\left|\alpha - {\frac {|h^{-1}(1) \cap g^{-1}(1)|}{|g^{-1}(1)|}}
{\frac {|g_i^{-1}(1)|}{1+\tau}} \leq \kappa_i \leq
(1 + \tau)|g_i^{-1}(1)|,
\]
\new{where $\mu = \tau = \eps/40000.$}
\end{lemma}

}

\begin{proof}
Suppose that $i$ is such that
$|h_i^{-1}(1) \cap g_i^{-1}(1)|/|g_i^{-1}(1)|  \geq \gamma/2$.
Recall from Definition~\ref{def:approx-unif-gen} that each point $x^j$ drawn
from $\A^{\C'}_\gen(g_i,\gamma/4,\delta'/(2m))$ in Step~1
is with probability $1-\delta'/(2m)$ distributed according
to $D_{g_i,\gamma/4}$; by a union bound we have that
with probability at least $1-\delta'/2$ all $m$ points
are distributed this way (and thus none of them are $\bot$).
We condition on this going forward.
Definition~\ref{def:approx-unif-gen} implies that $\dtv(D_{g_i,\gamma/4},
\U_{g_i^{-1}(1)}) \leq \gamma/4$; together with the assumption that
$|h_i^{-1}(1) \cap g_i^{-1}(1)|/|g_i^{-1}(1)|  \geq \gamma/2$,
this implies that each $x^j$ independently has proability
at least $\gamma/4$ of having $h(x)=1.$  Consequently, by the choice
of $m$ in Step~1, a standard multiplicative Chernoff bound
implies that
\[
\left|\alpha_i - {\frac {|h^{-1}(1) \cap g^{-1}(1)|}{|g^{-1}(1)|}}
\right| \leq
\mu \cdot {\frac {|h^{-1}(1) \cap g^{-1}(1)|}{|g^{-1}(1)|}}
\]
with failure probability at most $\delta'/4$, giving
(\ref{eq:alphagood}).

Finally, Definition~\ref{def:approx-count} gives that
(\ref{eq:kappagood}) holds with failure probability at most
$\delta'/4.$  This concludes the proof.
\end{proof}

Next we show how a high-accuracy estimate
$\alpha_i$ of $|h_i^{-1}(1) \cap g_i^{-1}(1)|/|g_i^{-1}(1)|$
yields a deterministic approximate evaluation oracle for
$\hat{D}'_i$.

\begin{lemma} \label{lem:approx-eval-for-hatDi}
Algorithm {\tt Simulate-Approx-Eval} (which is deterministic)
takes as input a value $\alpha \in [0,1]$, a string $x \in \bits^n$,
a parameter $\kappa$,
(a circuit for) $h: \bits^n \to \bits,$
and \new{(a representation for) $g: \bits^n \to \bits$, $g \in \C'$,}
where $h,g$ are obtained from a run of $\A'^\C_{\inv}.$
Suppose that
\[
\left|\alpha - {\frac {|h^{-1}(1) \cap g^{-1}(1)|}{|g^{-1}(1)|}}
\right| \leq
\mu \cdot {\frac {|h^{-1}(1) \cap g^{-1}(1)|}{|g^{-1}(1)|}}
\]
and
\[
{\frac {|g^{-1}(1)|}{1+\tau}} \leq \kappa \leq
(1 + \tau)|g^{-1}(1)|
\]
\new{where $\mu = \tau = \eps/40000.$}
Then {\tt Simulate-Approx-Eval} outputs a value $\rho$
such that
\begin{equation}
{\frac {\hat{D}'(x)}{1 + \beta}} \leq \rho \leq (1+\beta)
\hat{D}'(x),
\label{eq:goal}
\end{equation}
where
$\beta = \eps/192$, $\hat{D}$ is the output distribution constructed in
Step~3 of the run of $\A^\C_{\inv}$ that produced $h,g$,
and $\hat{D}'$ is $\hat{D}$ conditioned on $\{-1,1\}^n$
(excluding $\bot$).
\end{lemma}

\begin{proof}
The {\tt Simulate-Approx-Eval} procedure is very simple.  Given
an input $x \in \bits^n$ it evaluates both $g$ and $h$ on $x$,
and if either evaluates to $-1$ it returns the value 0.
If both evaluate to 1 then it returns the value $1/(\kappa \alpha).$

For the correctness proof,
note first that it is easy to see from the definition of the sampler
$C_f$ (Step~3 of $\A'^\C_{\inv}$) and Definition~\ref{def:approx-unif-gen}
(recall that the approximate uniform generator
$\A^{\C'}_\gen(g)$ only outputs strings that satisfy $g$)
that if $x \in \{-1,1\}^n$, $x \notin h^{-1}(1) \cap g^{-1}(1)$
then $\hat{D}$ has zero probability
of outputting $x$, so {\tt Simulate-Approx-Eval} behaves
appropriately in this case.

Now suppose that $h(x)=g(x)=1$.
We first show that the value $1/(\kappa \alpha)$ is multiplicatively
close to $1/|h^{-1}(1) \cap g^{-1}(1)|.$
Let us write $A$ to denote $|g^{-1}(1)|$ and
$B$ to denote $|h^{-1}(1) \cap g^{-1}(1)|$.  With this notation we have
\[
\left|
\alpha - {\frac B A}
\right|
\leq \mu \cdot {\frac B A}
\quad \quad \text{and} \quad \quad
{\frac {A}{1+\tau}} \leq \kappa \leq
(1 + \tau)A.
\]
Consequently, we have
\begin{eqnarray*}
B(1-\mu-\tau) \leq B \cdot {\frac {1-\mu}{1+\tau}} =
{\frac B A}(1-\mu) \cdot {\frac A {1+\tau}} \leq \kappa \alpha \leq
{\frac B A}(1+\mu) \cdot (1+\tau)A \leq B (1 + 2 \mu + 2 \tau),
\end{eqnarray*}
and hence
\begin{equation}
{\frac 1 B} \cdot {\frac 1 {1 + 2 \mu + 2 \tau}}
\leq {\frac 1 {\kappa \alpha}} \leq {\frac 1 B} \cdot {\frac 1 {1-\mu-\tau}}.
\label{eq:ok}
\end{equation}
Now consider any $x \in h^{-1}(1) \cap g^{-1}(1).$
By Definition~\ref{def:approx-unif-gen} we have that
\[
{\frac 1 {1+\eps_3}} \cdot {\frac 1 {|g^{-1}(1)|}}
\leq D_{g,\eps_3}(x) \leq
(1+\eps_3) \cdot {\frac 1 {|g^{-1}(1)|}}.
\]
Since a draw from $\hat{D}'$ is obtained by taking a draw
from $D_{g,\eps_3}$ and conditioning on it lying in $h^{-1}(1)$,
it follows that we have
\[
{\frac 1 {1+\eps_3}} \cdot {\frac 1 {B}}
\leq \hat{D}'(x) \leq
(1+\eps_3) \cdot {\frac 1 {B}}.
\]
Combining this with (\ref{eq:ok})
and
recalling that $\mu=\tau=\eps/40000$ and $\eps_3=\eps \gamma / 48000$,
we get (\ref{eq:goal}) as desired.
\end{proof}

\ignore{

\begin{framed}
\noindent Algorithm  {\tt Simulate-Approx-Eval}$(blah,blah):$

\smallskip

\noindent {\bf Input:} A string $x \in \bits^n$

\noindent {\bf Output:}
A value $\tilde{D}^\beta_i(x)$ such that
\[
{\frac {\hat{D}_i(x)}{1+\beta}} \leq
\tilde{D}^\beta_{i}(x) \leq
(1+\beta)\hat{D}_i(x).
\]

\begin{enumerate}

\item put

\item the alg

\item here

\end{enumerate}
\end{framed}

}

\subsection{The final algorithm: Proof of Theorem~\ref{thm:learn-by-dense}.}
\label{ssec:algo-gen}

Finally we are ready to give the inverse approximate uniform generation
algorithm $\A^\C_{\inv}$ for $\C$.

\begin{framed}
\noindent Algorithm  $\A^\C_{\inv} (\U_{f^{-1}(1)}, \eps, \delta)$

\smallskip

\noindent {\bf Input:} Independent samples from $\U_{f^{-1}(1)}$,
accuracy and confidence parameters $\eps,\delta$.

\noindent {\bf Output:}
With probability $1-\delta$ outputs an $\eps$-sampler
\new{$C_f$} for $\U_{f^{-1}(1)}$ .\\

\begin{enumerate}

\item For $i=1$ to $k=O(n/\eps)$:

\begin{enumerate}

\item Set $\widehat{p}_i \eqdef (1+\eps)^{i-1}/2^n$.

\item Run $\A'^\C_{\inv}(\U_{f^{-1}(1)},\eps/12,\delta/3,\widehat{p}_i)$.
Let \new{$g_i \in \C'$} be the function constructed in Step~1,
$h_i$ be the hypothesis function constructed in Step~2(e), and
$(C_f)_i$ be the sampler for distribution $\hat{D}_i$ constructed
in Step~3.

\item Run {\tt Check}$(g_i,h_i,\delta/3,\eps).$
If it returns a pair $(\alpha_i,\kappa_i)$ then add $i$ to the
set $S$ (initially empty).

\end{enumerate}

\item Run the hypothesis testing procedure
$\mathcal{T}^{\U_{f^{-1}(1)}}$ over the set $\{\hat{D}'_i\}_{i \in S}$
of hypothesis distributions, using
accuracy parameter $\eps/12$ and confidence parameter $\delta/3$.
Here $\mathcal{T}^{\U_{f^{-1}(1)}}$ is given
access to $\U_{f^{-1}(1)}$, uses the samplers $(C_f)_{i}$
to generate draws from distributions $\hat{D}'_i$
\new{(see Remark~\ref{rem:sampler}),}
and uses the procedure
{\tt Simulate-Approx-Eval}$(\alpha_i,\kappa_i,h_i,g_i)$
for the $(1+\eps/192)$-approximate evaluation oracle
$\eval_{\hat{D}'_i}$ for $\hat{D}'_i.$  Let
$i^\star \in S$ be the index of the distribution that it returns.

\item Output the sampler $(C_f)_{i^\star}$.

\end{enumerate}
\end{framed}

\medskip

\noindent {\bf Proof of Theorem~\ref{thm:learn-by-dense}:}
Let $p \equiv \Pr_{x \in \U_n}[f(x)=1]$ denote the true fraction of
satisfying assignments for $f$ in $\bits^n.$
Let $i^\ast$ be the element of $[k]$ such that
$p \leq \widehat{p}_{i^\ast} < (1+\eps/6)p.$
By Theorem~\ref{thm:known-bias}, with probability
at least $1-\delta/3$ we have that both

\begin{itemize}

\item [(i)]
$(C_f)_{i^\ast}$ is a sampler for a distribution
$\hat{D}_{i^\ast}$ such that
$\dtv(\hat{D}_{i^\ast},\U_{f^{-1}(1)}) \leq \eps/6$; and

\item [(ii)]
$|h^{-1}_{i^\ast}(1) \cap g^{-1}_{i^\ast}(1)|/|g^{-1}_{i^\ast}(1)| \geq
\gamma/2.$

\end{itemize}

We condition on these two events holding.
By Lemma~\ref{lem:check}, with probability
at least $1-\delta/3$ the procedure {\tt Check} outputs a value
$\alpha_{i^\ast}$ such that
\[
\left|
\alpha_{i^\ast} - {\frac {|h_{i^\ast}^{-1}(1) \cap g_{i^\ast}^{-1}(1)|}
{|g_{i^\ast}^{-1}(1)|}}
\right| \leq
\mu \cdot {\frac {|h_{i^\ast}^{-1}(1) \cap g_{i^\ast}^{-1}(1)|}
{|g_{i^\ast}^{-1}(1)|}}
\]
for $\mu= \eps/40000.$
We condition on this event holding.
Now Lemma~\ref{lem:approx-eval-for-hatDi} implies that
{\tt Simulate-Approx-Eval}$((C_f)_{i^\ast})$ meets the
requirements of a $(1+\beta)$-approximate evaluation oracle for
$\eval_{\hat{D}'_{i^\ast}}$ from Proposition~\ref{prop:log-cover-size},
for $\beta = {\frac \eps {192}}.$
Hence by Proposition~\ref{prop:log-cover-size}
\new{(or more precisely by Remark~\ref{rem:sampler})}
with probability at least $1-\delta/3$ the index $i^\star$
that $\mathcal{T}^{\U_{f^{-1}(1)}}$ returns is such that
$\hat{D}'_{i^\star}$ is an $\eps/2$-sampler for $\U_{f^{-1}(1)}$
as desired.
\ignore{
Recalling from Claim~\ref{claim:c0} that $\dtv(\hat{D}_{i^\star},
\hat{D}'_{i^\star}) \leq \eps/6$, we have that
$(C_f)_{i^\star}$ is an $\eps$-sampler for
$\U_{f^{-1}(1)}$ as desired, with overall failure probability at most
$\delta.$
}

As in the proof of Theorem~\ref{thm:known-bias}, the claimed running time
bound is a straightforward consequence of the
various running time bounds established for all the procedures
called by $\A^\C_\inv$.  This concludes the proof of our
general positive result, Theorem~\ref{thm:learn-by-dense}.
\qed

\fi

\section{Linear Threshold Functions} \label{sec:LTF}
In this section we apply our general framework from Section~\ref{sec:generaltechnique} to prove Theorem~\ref{thm:ltf-informal}, i.e.,
obtain a polynomial time algorithm for the problem of inverse approximate uniform generation for the class $\C = \ltf_n$
of $n$-variable linear threshold functions over $\bn$. More formally, we prove:
\begin{theorem} \label{thm:ltf-formal}
There is an algorithm $\A^{\ltf}_{\inv}$ which is a
$\poly \left( n, 1/\eps, \log(1/\delta) \right)$-time
inverse approximate uniform generation algorithm for the class $\ltf_n$.
\end{theorem}
The above theorem will follow as an application of Theorem~\ref{thm:learn-by-dense} for $\C' = \C = \ltf_n$.
The literature provides us with three of the four ingredients that our general approach requires for LTFs -- approximate uniform generation,
approximate counting, and Statistical Query learning -- and our main technical contribution is giving the fourth necessary ingredient,
a densifier. We start by recalling the three known ingredients in the following subsection.

\subsection{Tools from the literature.} \label{ssec:lit-ltfs}
We first record two efficient algorithms for approximate uniform generation and approximate counting for $\ltf_n$, due
to Dyer~\cite{Dyer03-short}:
\begin{theorem}\label{thm:dyer-gen} \ifnum\confversion=0{(approximate uniform generation for $\ltf_n$, \cite{Dyer03-short})} \else \fi
 There is an algorithm $\A^{\ltf}_{\gen}$ that on input (a weights--based representation of) an arbitrary $h \in \ltf_n$ and a confidence parameter
$\delta>0$, runs in time {$\poly(n,\log(1/\delta))$} and with probability $1-\delta$ outputs a point $x$ such that $x \sim \U_{h^{-1}(1)}$.
\end{theorem}
\noindent We note that the above algorithm gives us a somewhat stronger guarantee
\ifnum\confversion=0
than that in Definition~\ref{def:approx-unif-gen}.
\else
than we need.
\fi
Indeed,  the algorithm $\A^{\ltf}_{\gen}$ with high probability outputs a point $x \in \bn$
whose distribution is {\em exactly} $\U_{h^{-1}(1)}$ (as opposed to a point whose distribution is {\em close} to $\U_{h^{-1}(1)}$).
\begin{theorem}\label{thm:dyer-count} \ifnum\confversion=0{(approximate counting for $\ltf_n$, \cite{Dyer03-short})} \else  \fi
 There is an algorithm $\A^{\ltf}_{\co}$ that on input (a weights--based representation of) an arbitrary $h \in \ltf_n$, an accuracy parameter $\eps>0$ and a
confidence parameter $\delta>0$, runs in time {$\poly(n,1/\epsilon, \log(1/\delta))$}
and  outputs $\widehat{p} \in [0,1]$ that with probability $1-\delta$ satisfies $\widehat{p}
\in [1-\epsilon, 1+\epsilon] \cdot
\Pr_{x \sim \U_n}[h(x)=1]$.
\end{theorem}
We also need an efficient $\sq$ learning algorithm for halfpaces. This is provided to us by a result of  Blum et.~al.~\cite{BFK+:97}:
\ifnum\confversion=1 {\vspace{-0.2cm}} \fi
\begin{theorem} \label{thm:sq-ltfs} \ifnum\confversion=0{($\sq$ learning algorithm for $\ltf_n$, \cite{BFK+:97})} \else \fi
  There is a distribution-independent $\sq$ learning
algorithm $\A^\ltf_\sq$ for $\ltf_n$ that has running time
$t_1=\poly(n,1/\eps,\log(1/\delta))$, uses at most $t_2=\poly(n)$ time to
evaluate each query, and requires tolerance of its
queries no smaller than $\tau = 1/\poly(n,1/\eps).$
\end{theorem}

\subsection{A densifier for $\ltf_n$.}
\label{ssec:densifier-for-LTFs}
The last ingredient we need in order to apply our Theorem~\ref{thm:learn-by-dense} is a computationally efficient densifer for $\ltf_n$.
This is the main technical contribution of this section and is summarized in the following theorem:
\begin{theorem} \label{thm:ltf-densifier} \ifnum\confversion=0{(efficient proper densifier for $\ltf_n$) }\else \fi
 Set $\gamma (\eps, \delta, n) \eqdef \Theta \left( \delta / (n^2 \log n) \right)$.
There is  an $(\eps, \gamma, \delta)$--densifier $\A^{\ltf}_{\den}$
for $\ltf_n$ that, for any input parameters
$0 < \eps, \delta$, $1/2^n \leq \widehat{p} \leq 1$,
outputs a function
$g \in \ltf_n$ and runs in time $\poly(n, 1/\eps,$ $\log(1/\delta))$.
\end{theorem}

\noindent {\bf Discussion and intuition.} Before we prove Theorem~\ref{thm:ltf-densifier}, we provide some intuition.
Let $f \in \ltf_n$ be the unknown LTF and suppose that we would like to design an $(\eps, \gamma, \delta)$--densifier $\A^{\ltf}_{\den}$ for $f$.
That is, given sample access to $\U_{f^{-1}(1)}$, {and a number $\widehat{p}$ satisfying $p \le  \widehat{p}
< (1+\eps) p$,
where $p = \Pr_{x \in \U_n} [f(x)=1]$,}
we would like to efficiently compute (a weights--based representation for) an LTF $g:\bn \to \bits$
such that 
\ignore{with probability $1-\delta$}
 the following conditions are satisfied:
\begin{enumerate}
\item[(a)] $\Pr_{x \sim \U_{f^{-1}(1)}} \left[ g(x)=1 \right] \ge 1-\eps$, and
\item[(b)] $\Pr_{x \sim \U_n} \left[ g(x)=1 \right] \le (1/\gamma) \cdot \Pr_{x \sim \U_n}[f=1]$.
\end{enumerate}
(While  condition (b) above appears slightly different than property (b) in our Definition~\ref{def:dense},
because of property (a), the two statements are essentially equivalent up to a factor of $1/(1-\eps)$ in the value of $\gamma$.)

We start by noting that it is easy to handle the case that $\widehat{p}$ is
large.
In particular, observe that if $\widehat{p} \geq 2\gamma$ then $p = \Pr_{x \sim \U_n} [f(x) = 1] \ge \widehat{p} / (1+\eps) \ge \widehat{p} /2 \ge \gamma$, and we can just output $g \equiv 1$ since it clearly satisfies both properties of the definition.
{For the following intuitive discussion} we will henceforth assume that $\widehat{p} \le 2 \gamma.$

Recall that our desired function $g$ is an LTF, i.e., $g(x) = \sign(v \cdot x - t)$, for some $(v, t) \in \R^{n+1}$. {Recall also that our densifier has sample access to $\U_{f^{-1}(1)}$, so
it can obtain random positive examples of $f$, each of which gives
a linear constraint over the $v,t$ variables.} Hence a natural first approach is
to attempt to construct an appropriate linear program over these variables
whose feasible solutions satisfy conditions (a) and (b) above.
{We
begin by analyzing this approach; while it turns out to not
quite work,
it will gives us valuable intuition for our actual algorithm, which is
presented further below.
}

Note that following this approach, condition (a) is relatively easy to satisfy. Indeed, consider any $\eps>0$ and
suppose we want to construct an LTF $g = \sign(v \cdot x - t)$ such that $\Pr_{x \sim \U_{f^{-1}(1)}} \left[ g(x)=1 \right] \ge 1-\eps$.
This can be done as follows: draw a set $S_{+}$ of
$N_{+} = \Theta \left( (1/\eps) \cdot (n^2+\log(1/\delta))\right) $ samples from $\U_{f^{-1}(1)}$ and consider a linear program $\cal{LP}_+$
with variables $(w, \theta) \in \R^{n+1}$ that enforces all these examples to be positive. That is, for each $x \in S_+$, we will have an inequality
$w \cdot x \ge \theta$. 
\ignore{ (We also need a normalizing condition, but this is standard and easy to handle, hence we omit it from this intuitive explanation.)}
It is clear that $\cal{LP}_+$ is feasible (any weights--based representation for $f$ is a feasible solution) and that it can be solved in
$\poly(n, 1/\eps, \log(1/\delta))$ time, since it is defined by $O(N_+)$ many linear constraints and the coefficients of the constraint matrix are in $\{\pm 1\}$.
\ignore{Our algorithm outputs a feasible solution of $\cal{LP}_+$. }
The following simple claim shows that with high
probability any feasible solution of ${\cal LP}_+$
satisfies condition (a):

\begin{claim} \label{claim:a}
With probability at least $1-\delta$ over the sample $S_+$,
{\em any} $g \in \ltf_n$ consistent with $S_{+}$ satisfies condition (a).
\end{claim}
\begin{proof}
Consider an LTF $g$ and suppose that it does not satisfy condition (a), i.e., $\Pr_{x \sim \U_n} [g(x) = -1 | f(x) = 1] > \eps$.
Since each sample $x \in S_+$ is uniformly distributed in $f^{-1}(1)$, the probability
it does not ``hit'' the set $g^{-1}(-1) \cap f^{-1}(1)$ is at most $1-\eps$.
The probability that {\em no} sample in $S_+$ hits  $g^{-1}(-1) \cap f^{-1}(1)$ is thus at most
$(1-\eps)^{N_+} \le  {\delta}/{2^{n^2}}.$
Recalling that there exist at most $2^{n^2}$ distinct LTFs over $\bn$~\cite{Muroga:71},
it follows by a union bound that the probability there exists an LTF that does not satisfy condition (a) is at most $\delta$ as desired.
\end{proof}

The above claim directly implies that with high probability
{\em any} feasible solution $(w^{\ast}, \theta^{\ast})$ to $\cal{LP}_+$ is such that  $g^{\ast}(x) = \sign(w^{\ast} \cdot x - \theta^{\ast})$
satisfies condition (a).
Of course, an arbitrary feasible solution to $\cal{LP}_+$ is by no means guaranteed to satisfy condition (b).
(Note for example that the constant $1$ function
is certainly feasible for  $\cal{LP}_+$.) Hence, a natural idea is
to include additional constraints in our linear program so that condition (b) is also satisfied.

Along these lines, consider the following procedure: Draw a set $S_{-}$of
$N_{-} = \lfloor \delta/\widehat{p} \rfloor$ uniform unlabeled samples
from $\bn$ and label them negative. That is, for each sample $x \in S_{-}$,
we add the constraint $w \cdot x < \theta$ to our linear program.
{Let $\cal{LP}$ be the linear program that contains all the constraints
defined by $S_{+} \cup S_{-}$. It is not hard to prove that with  probability
at least $1-2\delta$ over the sample $S_{-}$, we have that $S_{-}
\subseteq f^{-1}(-1)$ and hence (any weight based representation of) $f$
is a feasible solution to $\cal{LP}$. In fact, it is possible to show that
{\em if $\gamma$ is sufficiently small} ---
roughly, $\gamma \le \delta/\left( 4(n^2+\log(1/\delta)) \right)$ is what
is required --- then with high probability each solution to  $\cal{LP}$ also
satisfies condition (b). The catch, of course, is that the above
procedure is not computationally efficient because $N_{-} $ may be very large
-- if $\widehat{p}$ is very small, then it is infeasible even to
write down the linear program ${\cal{LP}}$.
}

\smallskip

\noindent {\bf Algorithm Description.}
The above discussion motivates our actual densifier algorithm as follows:
{The problem with the above described naive approach is that
it generates {(the potentially very large set)} $S_{-}$
all at once at the beginning of the algorithm.
Note that having a large set $S_{-}$ is not necessarily in and of itself
a problem, since one could potentially use the ellipsoid method
to solve $\cal{LP}$ if one could obtain an efficient separation oracle.
Thus intuitively,  if one had an online algorithm which would generate
$S_{-}$ \emph{on the fly}, then one could potentially get a
feasible solution to $\cal{LP}$ in polynomial time.
This serves as the intuition behind our actual algorithm.}

More concretely, our densifier $\A^{\ltf}_{\den}$ will invoke a computationally efficient {\em online} learning algorithm for LTFs.
In particular, $\A^{\ltf}_{\den}$ will run the online learner $\A^{\ltf}_{\mt}$ for a sequence of stages
and in each stage it will provide as counterexamples to $\A^{\ltf}_{\mt}$ judiciously chosen labeled examples,
which will be positive for the online learner's current hypothesis, but negative for $f$ (with high probability).
Since $\A^{\ltf}_{\mt}$ makes a small number of mistakes in the worst-case, this process is guaranteed to terminate
after a small number of stages (since in each stage we {\em force} the online learner to make a mistake).

We now provide the details.
We start by recalling the notion of {\em online learning} for a class $\C$ of boolean functions.
In the online model, learning proceeds in a sequence of stages. In each stage the learning algorithm is given an
unlabeled example $x \in \bn$ and is asked to predict the value $f(x)$, where $f \in \C$ is the unknown target concept.
After the learning algorithm makes its prediction, it is given the correct value of $f(x)$.
The goal of the learner is to identify $f$ while minimizing the total number of mistakes.
We say that an online algorithm learns class $\C$ with mistake bound $M$
if it makes at most $M$ mistakes on {\em any} sequence of examples consistent with some $f \in \C$.
Our densifier makes essential use of a computationally efficient online learning algorithm for the class
of linear threshold functions by Maass and Turan~\cite{MT:94}:
\ifnum\confversion=1 {\vspace{-0.3cm}} \fi
\begin{theorem}  \ifnum\confversion=0{ (\cite{MT:94}, Theorem 3.3)} \else \fi \label{thm:online-mt}
There exists a $\poly(n)$ time deterministic online learning algorithm $\A^{\ltf}_{\mt}$ for the class $\ltf_n$ with mistake bound $M(n) \eqdef \Theta(n^2 \log n)$.
In particular, at every stage of its execution, the current hypothesis maintained by $\A^{\ltf}_{\mt}$ is a (weights--based representation of an) LTF that is  consistent
with all labeled examples received so far.
\end{theorem}
\ifnum\confversion=0
We note that the above algorithm works by reducing the problem of online learning for LTFs to a convex optimization 
problem.
Hence, one can use any efficient convex optimization algorithm to do online learning for LTFs, e.g. the ellipsoid method~\cite{Kha:80, GLS:88}.
The mistake bound in the above theorem follows by plugging in the algorithm of Vaidya~\cite{Vaidya:89, Vaidya:96}.
\else
We note that $\A^{\ltf}_{\mt}$ works by reducing online learning of 
LTFs to a convex optimization 
problem; however, 
our densifier will use algorithm $\A^\ltf_\mt$ as a black box.
\fi
\ignore{
That is, we essentially provide an efficient reduction of our problem to the problem of online learning for linear threshold functions.
}

We now proceed with a more detailed description of our densifier followed by
pseudocode and a proof of correctness.
As previously mentioned, the basic idea is to execute the online learner to learn $f$ while cleverly providing counterexamples to it in each stage of its execution.
Our algorithm starts by sampling $N_{+}$ samples from $\U_{f^{-1}(1)}$ and making sure that these are classified correctly by the online learner.
This step guarantees that our final solution will satisfy condition (a) of the densifier.
Let $h \in \ltf_n$ be the current hypothesis at the end of this process.
If $h$ satisfies condition (b) (we can efficiently decide this by using our approximate counter for $\ltf_n$), we output $h$ and terminate the algorithm.
Otherwise, we use our approximate uniform generator to construct a uniform satisfying assignment $x \in \U_{h^{-1}(1)}$ and we label it negative, i.e.,
we give the labeled example $(x, -1)$ as a counterexample to the online learner.
Since $h$ does not satisfy condition (b), i.e., it has ``many'' satisfying assignments,
it follows that with high probability
(roughly, at least $1-\gamma$)
over the choice of $x \in \U_{h^{-1}(1)}$,
the point $x$ output by the generator will indeed be
negative for $f$. We continue this process for a number of stages.
If all counterexamples thus generated are indeed consistent with $f$ (this happens with probability roughly $1-\gamma \cdot M$, where $M = M(n) = \Theta (n^2 \log n)$ is an upper bound on the number of stages), after at most $M$ stages we have either found a hypothesis $h$  satisfying condition (b) or the online learner terminates. In the latter case, the current hypothesis of the online learner is identical to $f$, as follows from Theorem~\ref{thm:online-mt}. (Note that the above argument puts an upper bound of $O(\delta/M)$
on the value of $\gamma$.) Detailed pseudocode follows:

\begin{framed}
\noindent Algorithm  $\A^{\ltf}_\den ( \U_{f^{-1}(1)}, \eps, \delta, \widehat{p})$:

\smallskip

\noindent {\bf Input:} Independent samples from $\U_{f^{-1}(1)}$,  parameters $\eps,\delta >0$,
and a value $1/2^n \leq \widehat{p} \leq 1.$

\noindent {\bf Output:}
If  $p \leq \widehat{p} \leq (1+\eps) p$, with probability $1-\delta$ outputs a
{function $g \in \ltf_n$ satisfying conditions (a) and (b).}

\begin{enumerate}
\item Draw a set $S_{+}$ of $N_{+} = \Theta \left( (1/\eps) \cdot (n^2+\log(1/\delta)) \right)$
         examples from $\U_{f^{-1}(1)}$.

\item Initialize $i=0$ and set $M \eqdef \Theta(n^2\log n)$.\\
While $(i \le M)$ do the following:
\begin{enumerate}
\item Execute the $i$-th stage of $A^{\ltf}_{\mt}$ and let
         $h^{(i)} \in \ltf_n$ be its current hypothesis.

\item If there exists $x \in S_{+}$ with $h^{(i)}(x) = -1$ do the following:

\begin{itemize}
\item Give the labeled example $(x, 1)$ as a counterexample to $A^{\ltf}_{\mt}$.
\item Set $i = i+1$ and go to Step 2.
\end{itemize}

\item Run $\A^{\ltf}_{\co} (h^{(i)}, \eps, \delta/(4M))$ and let $\widehat{p_i}$ be its output.
\item Set $\gamma  \eqdef \delta/(16M)$. If $\widehat{p_i} \le \widehat{p}/\big( \gamma \cdot (1+\eps)^2 \big)$ then output $h^{(i)}$;
\item otherwise, do the following:
\begin{itemize}
\item Run $\A^{\ltf}_{\gen} (h^{(i)}, \delta/(4M))$ and let $x^{(i)}$ be its output.

\item Give the point $(x^{(i)}, -1)$ as a counterexample to  $A^{\ltf}_{\mt}$.
\item Set $i = i+1$ and go to Step 2.
\end{itemize}
\end{enumerate}

\item Output the current hypothesis $h^{(i)}$ of $A^{\ltf}_{\mt}$.

\end{enumerate}
\end{framed}
\begin{theorem} \label{thm:ltf-den}
Algorithm $\A^{\ltf}_\den ( \U_{f^{-1}(1)}, \eps, \delta, \widehat{p})$
runs in time $\poly\left(n, 1/\eps, \log(1/\delta)\right)$.  If
\ignore{\log(1/\widehat{p})\right)$. If }$p \leq \widehat{p} < (1+\eps) p$
then with probability $1-\delta$
it outputs a
vector $(w, \theta)$ such that $g(x) = \sign(w \cdot x - \theta)$ satisfies conditions (a) and (b) at the start of Section~\ref{ssec:densifier-for-LTFs}.
\end{theorem}
\begin{proof}
First note that by Claim~\ref{claim:a}, with probability at least $1-\delta/4$ over $S_{+}$ any LTF consistent with $S_{+}$ will satisfy condition (a).
We will condition on this event and also on the event that each call to the approximate {counting algorithm} and to the approximate uniform generator is successful.
Since Step~2 involves at most $M$ iterations, by a union bound, with probability at least $1-\delta/4$ all calls to $\A^\ltf_{\co}$
will be successful, i.e., for all $i$ we will have that $p_i/(1+\eps)\le \widehat{p_i} \le (1+\eps) \cdot p_i$, where $p_i = \Pr_{x \in \U_n}[h^{(i)}(x)=1]$.
Similarly, with failure probability at most $\delta/4$, all points $x^{(i)}$ constructed  by $\A^{\ltf}_{\gen}$ will be uniformly random over
$(h^{(i)})^{-1}(1)$. Hence, with failure probability at most $3\delta/4$ all three conditions will be satisfied.

Conditioning on the above events, if the algorithm outputs a hypothesis $h^{(i)}$ in Step~2(d), this hypothesis will certainly satisfy condition (b), since
$p_i \le (1+\eps) \widehat{p_i} \le \widehat{p}/\big( \gamma \cdot (1+\eps) \big) \le p/\gamma$.
In this case, the algorithm succeeds with probability at least $1-3\delta/4$.
It remains to show that if the algorithm returns a hypothesis in Step~3, it will be successful with probability at least $1-\delta$.
To see this, observe that if no execution of Step~2(e) generates a point $x^{(i)}$ with $f(x^{(i)})=1$,
all the counterexamples given to $A^{\ltf}_{\mt}$ are consistent with $f$. Therefore, by Theorem~\ref{thm:online-mt}, the hypothesis
of Step~3 will be identical to $f$, which trivially satisfies both conditions.

We claim that with overall probability at least $1-\delta/4$ all executions of  Step~2(e) generate points $x^{(i)}$ with $f(x^{(i)})=-1$.
Indeed, fix an execution of Step~2(e). Since $\widehat{p_i} > \widehat{p}/\left( (1+\eps)^2 \cdot \gamma \right) $, it follows that $p \le (4\gamma) p_i.$
Hence, with probability at least $1-4\gamma$ a uniform point $x^{(i)} \sim \U_{(h^{i})^{-1}(1)}$ is a negative example for $f$, i.e., $x^{(i)} \in f^{-1}(-1)$.
By a union bound over all stages, our claim holds except with failure probability $4\gamma \cdot M = \delta/4$, as desired. This completes the proof
of correctness.

It remains to analyze the running time. Note that Step~2 is
repeated at most $M = O(n^2 \log n)$ times.
\ignore{The running time of each iteration is linear in $\log(1/\widehat{p})$
and involves}
Each iteration involves (i) one round of the online learner $\A^{\ltf}_{\mt}$
(this takes $\poly(n)$ time by Theorem~\ref{thm:online-mt}),
(ii) one call of $\A^{\ltf}_{\co}$ (this takes
$\poly(n, 1/\eps, \log(1/\delta))$ time by Theorem~\ref{thm:dyer-count}),
and (iii) one call to $\A^{\ltf}_{\gen}$
(this takes $\poly(n, 1/\eps, \log(1/\delta))$ time by
Theorem~\ref{thm:dyer-gen}). This completes the proof of
Theorem~\ref{thm:ltf-den}.
\end{proof}
\ignore{
\anote{I could not understand a couple of points: First, is why do we have $p \le (4\gamma) p_i.$ and not $p \le \gamma \cdot p_i.$. Second, is why this factor of $\log(1/\widehat{p})$. At worst, its $n$. So, just put it as $\poly(n)$, right?}
}

\ignore{
{\bf OLD STUFF:}
Correctness: Consider the $i$-th iteration: Since $|h_i^{-1}(1)|$ is high, it means that whp a uniform element from this set is negative for $f$.
Hence, this example provides a counterexample for the learning algorithm. This is true in each round independently. If we have $M$ rounds,
and the probability of that event is $\delta/M$ we are ok. So, basically this dictates what exactly ``big'' means.

The algorithm either outputs a correct $h_i$ in some iteration, or the number of iterations reaches $M$. But, with high probability all the examples
given to the online learner as counterexamples are in fact consistent with $f$. Hence, MT says after $M$ mistakes we have learnt $f$ exactly.
The end. }

\section{DNFs} \label{sec:DNF}
In this section we apply our general positive result, Theorem~\ref{thm:learn-by-dense}, to give a
quasipolynomial-time algorithm for the inverse approximate uniform generation problem for $s$-term DNF
formulas.  Let $\dnfns$ denote the class of all $s$-term DNF formulas over $n$ Boolean
variables (which for convenience we think of as $0/1$ variables).  Our main
result of this section is the following:

\begin{theorem} \label{thm:dnf-formal}
There is an algorithm $\A^{\dnfns}_{\inv}$ which is an inverse approximate uniform
generation algorithm for the class $\dnfns$.  Given input parameters $\eps,\delta$
the algorithm runs in time $\poly \left(n^{\log(s/\eps)},\log(1/\delta) \right)$.
\end{theorem}

We note that even in the standard uniform distribution learning model the fastest known running
time for learning $s$-term DNF formulas to accuracy $\eps$ is
$\poly(n^{\log(s/\eps)},\log(1/\delta))$
\cite{ver90,gregvaliantfocs12}.
Thus it seems likely that obtaining a
$\poly(n,s,1/\eps)$-time algorithm would require a significant breakthrough in computational learning
theory.

For our application of Theorem~\ref{thm:learn-by-dense} for DNFs we
shall have $\C=\dnfns$ and $\C' = \dnf_{n,t}$ for some $t$ which
we shall specify later.
As in the case of LTFs, the literature provides us with three of the four ingredients that our general
approach requires for DNF ---  approximate uniform generation, approximate counting, and Statistical Query learning (more on this below) --- and our main technical contribution is giving the fourth necessary
ingredient, a densifier.  Before presenting and analyzing our densifier algorithm we recall
the other three ingredients.

\subsection{Tools from the literature.}
Karp, Luby and Madras~\cite{KLM89} have given approximate
uniform generation and approximate counting algorithms for DNF formulas.
(We note that~\cite{JVV86} give an efficient algorithm that with high probability outputs an {\em exactly} 
uniform satisfying assignment for DNFs.)

\begin{theorem}\label{thm:KLM-generate} (Approximate uniform generation for DNFs, \cite{KLM89})
There is an approximate uniform generation algorithm $\A^{\dnf_{n,t}}_\gen$
for the class $\dnf_{n,t}$ that runs in time $\poly(n,t,1/\eps,\log(1/\delta)).$
\ignore{
}
\end{theorem}

\begin{theorem}\label{thm:KLM-count} (Approximate counting for DNFs, \cite{KLM89})
There is an approximate counting algorithm $\A^{\dnf_{n,t}}_\gen$ for the
class $\dnf_{n,t}$ that runs in time $\poly(n,t,1/\eps,\log(1/\delta)).$
\ignore{
}
\end{theorem}

The fastest known algorithm in the literature for $\sq$ learning $s$-term DNF
formulas under arbitrary distributions
runs in time $n^{O(n^{1/3} \log s)} \cdot \poly(1/\eps)$
\cite{KlivansServedio:04jcss}, which is much more than our desired running
time bound.  However, we will see that we are able to use known
\emph{malicious noise tolerant} $\sq$ learning algorithms
for learning \new{\emph{sparse disjunctions} over $N$ Boolean variables
rather than DNF formulas.  In more detail,
our densifier will provide us with a set of $N=
n^{O(\log (s/\eps))}$ many conjunctions which is such that
the target function $f$ is very close to a disjunction (which we call $f'$) over
an unknown subset of at most $s$ of these $N$ conjunctions.  Thus intuitively any learning
algorithm for disjunctions, run over the ``feature space'' of conjunctions provided by the
densifier, would succeed if the target function were $f'$, but the target function is
actually $f$ (which is not necessarily exactly a disjunction over these $N$ variables).
Fortunately, known results on the malicious noise tolerance of specific
$\sq$ learning algorithms imply that it is in fact possible to use these $\sq$ algorithms
to learn $f$ to high accuracy, as we now explain.

We now state the precise $\sq$ learning result that we will use.
The following theorem is a direct consequence of, e.g., Theorems~5 and~6 of \cite{Decatur:93}
or alteratively of Theorems~5 and~6 of \cite{AslamDecatur:98}:

\begin{theorem} \label{thm:SQ-DNF}
(Malicious noise tolerant $\sq$ algorithm for learning sparse disjunctions)
Let $\C_{\disj,k}$ be the class of all disjunctions of length at most $k$
over $N$ Boolean variables
$x_1,\dots,x_N.$  There is a distribution-independent $\sq$ learning
algorithm $\A^\disj_\sq$ for $\C_{\disj,k}$ that has running time
$t_1=\poly(N,1/\eps,\log(1/\delta))$, uses at most $t_2=\poly(N)$ time to
evaluate each query, and requires tolerance of its
queries no smaller than $\tau = 1/\poly(k,1/\eps).$
The algorithm outputs a hypothesis which is a disjunction over
$x_1,\dots,x_N.$

Moreover, there is a fixed polynomial $\ell(\cdot)$ such that algorithm
$\A^\disj_\sq$ has the following property:
Fix a distribution $D$ over $\{0,1\}^N.$  Let $f$ be an $N$-variable Boolean function
which is such that $\Pr_{x \sim D}[f'(x) \neq f(x)] \leq \kappa$,
where $f' \in \C_{\disj,k}$ is some $k$-variable disjunction and $\kappa \leq \ell(\eps/k)<\eps/2.$
Then if $\A^\disj_\sq$ is run with a $\stat(f,D)$ oracle, with probability
$1-\delta$ it outputs a hypothesis $h$ such that
$\Pr_{x \sim D}[h(x) \neq f'(x)] \leq \eps/2$, and hence
$\Pr_{x \sim D}[h(x \neq f(x)] \leq \eps.$
\end{theorem}

(We note in passing that at the heart of Theorem~\ref{thm:SQ-DNF} is an
\emph{attribute-efficient} $\sq$ algorithm for learning sparse
disjunctions.
Very roughly speaking, an attribute
efficient $\sq$ learning algorithm
is one which can learn a target function over
$N$ variables, which actually depends only on an unknown subset of $k \ll N$
of the variables, using statistical queries for which the minimum value of the tolerance $\tau$
is ``large.''  The intuition behind Theorem~\ref{thm:SQ-DNF}
is that since the distance between $f$ and $f'$ is much less than $\tau$, the effect of using
a $\stat(f,\D)$ oracle rather than a $\stat(f',\D)$ oracle is negligible, and hence the
$\sq$ algorithm will succeed whether it is run with $f$ or $f'$ as the target function.)
}

\ignore{
We will also need the following result of
Aslam and Decatur, which, roughly speaking,
states that $\sq$ algorithms which use a large
tolerance parameter can learn in the presence of noise.  (More precisely,
the result says that if $g$ is close to $f$ and $g$ belongs to a class
$C$ that can be learned by an SQ algorithm $A$ with a large tolerance
parameter, then $A$ can be used to learn $f$ to high accuracy.)

\begin{theorem}
\label{thm:Decatur93}
Let $f$ be an $N$-variable Boolean function belonging to a class
$C$.  Fix a distribution $\D$ over $\{0,1\}^N$ and let
$g$ be a function for which $\Pr_{x \sim \D}[f(x) \neq g(x)]
\leq \kappa$.
\end{theorem}
}

\subsection{A densifier for $\dnfns$ and the proof of Theorem~\ref{thm:dnf-formal}.} \label{subsec:proof-of-dnf-formal}

In this subsection we state our main theorem regarding the existence
of densifiers for DNF formulas, Theorem~\ref{thm:dnf-densifier},
and show how Theorem~\ref{thm:dnf-formal} follows from this theorem.

\begin{theorem} \label{thm:dnf-densifier}
Let $\gamma(n,s,1/\eps,1/\delta) 
= 1/(4n^{2\log (2s/\ell(\eps/s))}\log(s/\delta))$.
Algorithm $\A^{\dnfns}_\den(\U_{f^{-1}(1)},\eps,\delta,\widehat{p})$
outputs a collection ${\cal S}$ of
conjunctions $C_1,\dots,C_{|\cal S|}$ and has
the following performance guarantee:  If
$p \eqdef \Pr_{x \sim \U_n}[f(x)=1] \leq \widehat{p} < (1+\eps)p$, then with probability
at least $1-\delta$, the function
$g(x) \eqdef  \vee_{i \in  [|{\cal S}|]  } C_i$
satisfies the following:
\begin{enumerate}
\item $\Pr_{x \sim \U_{f^{-1}(1)}}[g(x)=1] \geq 1-\eps$;
\item $\Pr_{x \sim \U_{g^{-1}(1)}}[f(x)=1] \geq \gamma(n,s,1/\eps,1/\delta)$.
\item \new{There is a DNF $f' = C_{i_1} \vee \cdots \vee C_{i_{s'}}$, which is a disjunction
of $s' \leq s$ of the conjunctions $C_1,\dots,C_{|{\cal S}|}$, such that
$\Pr_{x \sim \U_{g^{-1}(1)}}[f'(x) \neq f(x)] \leq \ell(\eps/s)$, where $\ell(\cdot)$
is the polynomial from Theorem~\ref{thm:SQ-DNF}.}
\end{enumerate}
The size of ${\cal S}$ and the running time of $\A^\dnfns_\den(\U_{f^{-1}(1)},\eps,\delta,\widehat{p})$
is $\poly(n^{\log(s/\eps)},\log(1/\delta))$.
\end{theorem}

With a slight abuse of terminology
we may rephrase the above theorem as saying that $\A^\dnfns_\den$ is a
$(\eps,\gamma,\delta)$-densifier for function class $\C=\dnfns$ using class
$\C'=\dnf_{n,t}$ where $t=n^{O(\log(s/\eps))}.$
We defer the description of Algorithm~$\A^{\dnfns}_\den$ and the proof
of Theorem~\ref{thm:dnf-densifier} to the next subsection.


\begin{proof}[Proof of Theorem~\ref{thm:dnf-formal}]
\new{
The proof is essentially just an application of
Theorem~\ref{thm:learn-by-dense}.  The only twist is the use of a $\sq$ disjunction
learning algorithm rather than a DNF learning algorithm, but the special
properties of Algorithm~$\A^\disj_\sq$ let this go through without a problem.

In more detail, in Step~2(e) of Algorithm~$\A'^\C_\inv$ (see
Section~\ref{ssec:algo-known-bias}), in the execution of
Algorithm~$\A_\sqs$, the $\sq$ algorithm that is simulated
is the algorithm $\A^\disj_\sq$ run over the feature
space ${\cal S}$ of all conjunctions that are output by
Algorithm~$\A^{\dnfns}_\den$ in Step~1 of Algorithm~$\A'^\C_\inv$
(i.e., these conjunctions play the role of variables $x_1,\dots,x_N$
for the $\sq$ learning algorithm).  Property (3) of Theorem~\ref{thm:dnf-densifier}
and Theorem~\ref{thm:SQ-DNF} together imply that the algorithm $\A^\disj_\sq$, run on
a $\stat(f,\U_{g^{-1}(1)})$ oracle with parameters $\eps,\delta$, would
with probability $1-\delta$ output a hypothesis $h'$ satisfying
$\Pr_{x \sim \U_{g^{-1}(1)}}[h'(x) \neq f(x)] \leq \eps.$  Hence the hypothesis $h$
that is output by $\A_{\sqs}$ in Step~2(e) of Algorithm~$\A'^\C_\inv$
fulfills the necessary accuracy (with respect to $f$ under $D=\U_{g^{-1}(1)}$) and confidence
requirements, and the overall algorithm~$A^\C_\inv$ succeeds as described in Theorem~\ref{thm:learn-by-dense}.

Finally, combining the running time bounds of $\A^\dnfns_\den$
and $\A^\disj_\sq$ with the time bounds of the other procedures described earlier,
one can straightforwardly verify that the running time of the overall
algorithm $\A^\C_\inv$ is $\poly(n^{\log(s/\eps)},\log(1/\delta)).$
}
\end{proof}

\subsection{Construction of a densifier for $\dnfns$ and proof of
Theorem~\ref{thm:dnf-densifier}.} \label{ssec:dnf-densifier-construction}

Let $f=T_1 \vee \cdots \vee T_s$ be the target $s$-term DNF formula, 
where $T_1,\dots,T_s$ are the terms (conjunctions).
The high-level idea of our densifier is quite simple:  If $T_i$ is a term 
which is ``reasonably likely'' to be satisfied by a uniform draw of $x$ 
from $f^{-1}(1)$, then $T_i$ is at least ``mildly likely'' to be satisfied 
by $r=2 \log n$ consecutive independent draws of $x$
from $f^{-1}(1)$.  Such a sequence of draws $x^1,\dots,x^r$ will
with high probability 
\emph{uniquely identify} $T_i$.  By repeating this process sufficiently many 
times, with high probability we will obtain a pool $C_1,\dots,C_{|{\cal S}|}$ 
of conjunctions which contains all of the terms $T_i$ that are reasonably 
likely to be satisfied by a uniform draw of $x$
from $f^{-1}(1).$  Theorem~\ref{thm:dnf-densifier} follows straightforwardly 
from this.

We give detailed pseudocode for our densifier algorithm below:

\begin{framed}
\noindent Algorithm  $\A^{\dnfns}_\den ( \U_{f^{-1}(1)}, \eps, \delta, \widehat{p})$:

\smallskip

\noindent {\bf Input:} Independent samples from $\U_{f^{-1}(1)}$,
parameters $\eps,\delta >0$, and a value $1/2^n < \widehat{p} \leq 1.$

\noindent {\bf Output:}
If  $p \leq \widehat{p} \leq (1+\eps) p$, with probability $1-\delta$
outputs a set ${\cal S}$ of conjunctions $C_1,\dots,C_{|{\cal S}|}$
as described in Theorem~\ref{thm:dnf-densifier}

\begin{enumerate}

\item Initialize set ${\cal S}$ to $\emptyset.$  Let $\ell(\cdot)$ be the polynomial from
Theorem~\ref{thm:SQ-DNF}.

\item For $i=1$ to $M=
2 n^{2\log (2s/\ell(\eps/s))}\log(s/\delta)$,
repeat the following:

\begin{enumerate}

\item Draw $r=2 \log n$ satisfying assignments $x^1,\dots,x^{r}$ from $\U_{f^{-1}(1)}$.

\item Let $C_i$ be the AND of all literals that take the same value in all $r$ strings
$x^1,\dots,x^r$ (note $C_i$ may be the empty conjunction).  We say $C_i$ is a
\emph{candidate term.}

\item If the candidate term $C_i$ satisfies $\Pr_{x \sim \U_n}[C_i(x)=1] \leq 
\widehat{p}$ then add $C_i$ to the set ${\cal S}$.

\end{enumerate}

\item Output ${\cal S}.$
\end{enumerate}

\end{framed}

The following crucial claim makes the intuition presented at the start of this
subsection precise:

\begin{claim} \label{claim:dnf-dens-key}
Suppose $T_j$ is a term in $f$ such that 
$\Pr_{x \sim \U_{f^{-1}(1)}}[T_j(x)=1] \geq \ell(\eps/s)/(2s).$
Then with probability at least $1 - \delta/s$, term $T_j$ 
is a candidate term at some iteration of Step~2 of Algorithm~$\A^\dnfns_\den
(\U_{f^{-1}(1)},\eps,\delta,\widehat{p}).$
\end{claim}

\begin{proof}
Fix a given iteration $i$ of the loop in Step~2.  With probability at least
$$(\ell(\eps/s)/(2s))^{2 \log n} = (1/n)^{2 \log(2s/\ell(\eps/s))},$$
all $2 \log n$ points $x^1,\dots,x^{2 \log n}$ satisfy
$T_j$; let us call this event $E$, and condition on $E$ taking place.  We claim that conditioned
on $E$, the points $x^1,\dots,x^{2 \log n}$ are independent uniform samples drawn from
$T_j^{-1}(1)$.  (To see this, observe that each $x^i$ is an independent sample chosen uniformly at random
from $f^{-1}(1) \cap T_j^{-1}$; 
but $f^{-1}(1) \cap T_j^{-1}(1)$ is identical to 
$T_j^{-1}(1).$)  Given that $x^1,\dots,x^{2 \log n}$ are independent uniform 
samples drawn from $T_j^{-1}(1)$, the probability that any literal 
which is \emph{not} present in $T_j$ is contained in $C_i$ (i.e., is satisfied 
by all $2 \log n$ points) is at most $2n/n^2 \leq 1/2.$
So with overall probability at 
least ${\frac 1 {2 n^{2 \log(2s/\ell(\eps/s))}}}$, the term $T_j$ is a 
candidate term at iteration $i$.  Consequently $T_j$ is a candidate term at 
some iteration with probability at least $1-\delta/s$, by the choice of $M=
2 n^{2\log (2s/\ell(\eps/s))}\log(s/\delta).$
\end{proof}

\noindent Now we are ready to prove Theorem~\ref{thm:dnf-densifier}:

\begin{proof}[Proof of Theorem~\ref{thm:dnf-densifier}]
The claimed running time bound of $\A^\dnfns_\den$ is easily verified, so it remains only to establish 
(1)-(3).  Fix $\widehat{p}$ such that $p \leq \widehat{p} < (1+\eps)p$ where $p=\Pr_{x \sim \U_n}[f(x)=1].$

Consider any fixed term $T_j$ of $f$ such that
$\Pr_{x \sim \U_{f^{-1}(1)}}[T_j(x)=1] \geq \ell(\eps/s)/(2s).$  
By Claim~\ref{claim:dnf-dens-key} we have
that with probability at least $1-\delta/s$, term $T_j$ is a candidate term 
at some iteration of Step~2 of the algorithm.  We claim that
in step (c) of this iteration the term $T_j$ will in fact be added 
to ${\cal S}$.  This is because by assumption we have
\[
\Pr_{x \sim \U_n}[T_j(x)=1] \leq \Pr_{x \sim \U_n}[f(x)=1] = p \leq \widehat{p}.
\]
So by a union bound, with probability at least $1-\delta$ every term $T_j$ 
in $f$ such that $\Pr_{x \sim \U_{f^{-1}(1)}}[T_j(x)=1] \geq \ell(\eps/s)/(2s)$ 
is added to ${\cal S}$.  

Let $L$ be the set of those terms $T_j$ in $f$ that 
have $\Pr_{x \sim \U_{f^{-1}(1)}}[T_j(x)=1] \geq \ell(\eps/s)/(2s)$.  
Let $f'$ be the DNF obtained by taking the OR of all terms in $L$.
By a union bound over the (at most $s$) terms that are in $f$ but 
not in $f'$, we have $\Pr_{x \sim \U_{f^{-1}(1)}}[f'(x)=1] 
\geq 1-\ell(\eps/s)/2.$  
Since $g$ (as defined in Theorem~\ref{thm:dnf-densifier} 
has $g(x)=1$ whenever $f'(x)=1$, it follows that
$\Pr_{x \sim \U_{f^{-1}(1)}}[g(x)=1] \geq 1-\ell(\eps/s)/2
\geq 1 - \eps$, giving item (1) of the theorem.

For item (2), since $f(x)=1$ whenever $f'(x)=1$, we have 
$\Pr_{x \sim \U_{g^{-1}(1)}}[f(x)=1] \geq 
\Pr_{x \sim \U_{g^{-1}(1)}}[f'(x)=1].$  Every $x$ such that $f'(x)=1$ 
also has $g(x)=1$ so to lower bound $\Pr_{x \sim \U_{g^{-1}(1)}}[f'(x)=1]$ 
it is enough to upper bound the number of points in $g^{-1}(1)$ and 
lower bound the number of points in $f'^{-1}(1)$.
Since each $C_i$ that is added to ${\cal S}$ is satisfied by at most 
$\widehat{p}2^n \leq (1+\eps)p2^n$ points, we have that 
$|g^{-1}(1)|\leq (1+\eps)pM2^n$.  
Since at least $1-\eps$ of the points that satisfy $f$ 
also satisfy $f'$, we have that $|f'^{-1}(1)| \geq p(1-\eps)2^n.$  
Thus we have $\Pr_{x \sim \U_{g^{-1}(1)}}[f'(x)=1] 
\geq p(1-\eps)/((1+\eps)pM) = {\frac {1-\eps}{1+\eps}} 
\cdot {\frac 1 M} > {\frac 1 {2M}}$, giving (2).

Finally, for (3) we have that $f(x) \neq f'(x)$ only on those inputs that 
have $f(x)=1$ but $f'(x)=0$ (because some term outside of $L$ is satisfied 
by $x$ and no term in $L$ is satisfied by $x$).  Even if all such inputs $x$ 
lie in $g^{-1}(1)$ (the worst case), there can be at most $(\ell(\eps/s)/2)
p 2^n$ such inputs, and we know that $|g^{-1}(1)| \geq |f^{-1}(1)|
\geq p(1-\eps)2^n$.  So we have
$\Pr_{x \sim \U_{g^{-1}(1)}}[f(x) \neq f'(x)] \leq {\frac {\ell(\eps/s)/2}
{1-\eps}} \leq \ell(\eps/s)$, and we have (3) as desired.
\end{proof}

\subsection{Inverse approximate uniform generation for $k$-DNFs.}
We briefly note that our general approach immediately yields an efficient inverse approximate
uniform generation algorithm for the class of $k$-DNFs for any constant $k$.
Let $k$-$\dnf$ denote the class of all $k$-DNFs over $n$ Boolean variables, i.e.,
DNF formulas in which each term (conjunction) has at most $k$ literals.

\begin{theorem} \label{thm:kdnf-formal}
There is an algorithm $\A^{k\text{-}\dnf}_{\inv}$ which is an inverse approximate uniform
generation algorithm for the class $k$-$\dnf$.  Given input parameters $\eps,\delta$
the algorithm runs in time $\poly \left(n^{k},1/\eps,\log(1/\delta) \right)$.
\end{theorem}

For any $k$-DNF $f$ it is easy to see that
$\Pr_{x \sim \U_n}[f(x)=1] \geq 1/2^k$, and consequently the constant 1 function is a
$\gamma$-densifier for $k$-$\dnf$ with $\gamma = 1/2^k.$  Theorem~\ref{thm:kdnf-formal} then follows
immediately from Theorem~\ref{thm:learn-by-dense}, using the algorithms for approximate uniform generation and counting of DNF formulas mentioned above \cite{KLM89} together with well-known
algorithms for $\sq$ learning $k$-DNF formulas in $\poly(n^k,1/\eps,\log(1/\delta))$ time
\cite{Kearns:98}.

\newpage

\section{Negative results for inverse approximate uniform generation} \label{sec:hardness}

In this section, we will prove hardness results for
inverse approximate uniform generation problems
for specific classes ${\cal C}$ of Boolean functions.
As is standard in
computational learning theory, our hardness results are based on
cryptographic hardness assumptions.
The hardness assumptions we use are well
studied assumptions in cryptography such as the strong RSA assumption,
Decisional Diffie Hellman problem,
and hardness of learning parity with noise.

As was alluded to in the introduction, in light of the standard
approach, there are two potential barriers to obtaining inverse approximate
uniform generation algorithms for a class $\mathcal{C}$ of functions.
The first is that
``reconstructing" the object from class $\mathcal{C}$ may be hard, and the
second is that sampling approximately
uniform random satisfying assignments
from the reconstructed object may be hard.
While any hard inverse approximate uniform generation problem must be hard
because of one of these two potential barriers, we emphasize here
that even if one of the two steps in the standard
approach is shown to be hard, this does not constitute
a proof of hardness of the overall inverse approximate uniform generation
problem, as there is may exist some efficient algorithm for the class
$\mathcal{C}$ which departs from the standard approach. Indeed, we will
give such an example in Section~\ref{sec:graph-auto}, where we give an
efficient algorithm for a specific inverse approximate uniform generation
problem that does not follow the standard approach.
(In fact, for that problem, the second
step of the standard approach is provably no easier than the well-known graph
automorphism problem, which has withstood several decades of effort
towards even getting a sub-exponential time algorithm.)

Our hardness results come in two flavors.  Our first hardness results, 
based on signature schemes, are for problems where it is provably hard 
(of course under a computational hardness assumption) to sample 
approximately uniform satisfying assignments. In contrast, our hardness 
results of the second flavor are based on Message Authentication Codes 
(MACs).  We give such a result for a specific class $\mathcal{C}$ which 
has the property that it is actually easy to sample uniform satisfying 
assignments for functions in $\mathcal{C}$; hence, in an informal sense, 
it is the first step in the standard approach that is algorithmically 
hard for this problem.  The following subsections describe all of our 
hardness results in detail.

\subsection{Hardness results based on signature schemes.} \label{sec:signature}
\new{
In this subsection we prove a general theorem, Theorem~\ref{thm:signatures},
which relates the hardness of inverse approximate uniform generation
to the existence of certain secure signature schemes in cryptography.
Roughly speaking, Theorem~\ref{thm:signatures} says that if secure
signature schemes exist, then the inverse approximate uniform generation
problem is computationally hard for any class $\mathcal{C}$ which is
Levin-reducible from \textsf{CIRCUIT-SAT}.
We will use this general result to establish hardness of inverse approximate
uniform generation for several natural classes of functions, including
3-CNF formulas, intersections of two halfspaces, and degree-2 polynomial
threshold functions (PTFs).}

We begin by recalling the definition of public key signature schemes.
For an extensive treatment of signature schemes, see \cite{gol-fou02}.
For simplicity, and since it suffices for our purposes, we only consider
schemes with deterministic verification algorithms.

\begin{definition}\label{def:sign}

A signature scheme is a triple $(G,S,V)$ of \new{polynomial-time}
algorithms with the following properties :

\begin{itemize}

\item {\bf (Key generation algorithm)}
$G$ is a randomized algorithm which on input $1^n$ produces a
pair $(pk, sk)$ (note that the sizes of both $pk$ and $sk$ are polynomial
in $n$).

\item {\bf (Signing algorithm)}
$S$ is a randomized algorithm which takes as input a message $m$
\new{from the message space ${\cal M}$},
a secret key $sk$ and randomness $r \new{\in \{0,1\}^n}$,
and outputs \new{a signature} $\sigma =
S(m,sk,r)$.

\item {\bf (Verification algorithm)}
$V$ is a deterministic algorithm such that
$V(m,pk,\sigma)=1$ for every $\sigma = S(m,sk,r)$.

\end{itemize}

\end{definition}

We will require signature schemes with some special
properties which we now define, first fixing some notation.
Let $(G,S,V)$ be a signature scheme.
For a message space
$\mathcal{M}$ and pair $(pk,sk)$ of public and secret keys, we define
the set $\mathcal{R}_{1,sk}$ of ``valid" signed messages as \new{the set of all
possible signed messages $(m,\sigma = S(m,sk,r))$ as $m$ ranges over all of
$\mathcal{M}$ and $r$ ranges over all of $\{0,1\}^{n}$.}
Similarly, we define the set $\mathcal{R}_{2,pk}$ of
``potential" signed messages as $\mathcal{R}_{2,pk} = \{(m,\sigma) :
V(m,pk,\sigma)=1\}$.  Likewise, we define the set of valid signatures
for message $m$, denoted $\mathcal{R}_{1,sk}(m)$, as the set of all
possible pairs $(m,\sigma=S(m,sk,r))$ as $r$ ranges over all of
$\{0,1\}^n$, and we define the set of potential signatures for message $m$ as
$\mathcal{R}_{2,pk}(m) = \{(m,\sigma)  : V(m,pk,\sigma)=1\}$.

\begin{definition}\label{def:sign1}

Let $(G,S,V)$ be a signature scheme and $\mathcal{M}$ be a message
space.  A pair $(pk,sk)$ of public and secret keys is
said to be \emph{$(\delta,\eta)$-special} if the following properties hold :

\begin{itemize}

\item Let $\mathcal{R}_{1,sk}$ be the set of valid signed messages
and
$\mathcal{R}_{2,pk}$ be the set of potential signed messages.  Then
$\frac{|\mathcal{R}_{1,sk}|}{|\mathcal{R}_{2,pk}|} \ge 1-\eta$.

\item For any fixed pair $(m,\sigma) \in \R_{1,sk}(m)$,
we have $\Pr_{r \in \{0,1\}^{n}}[\sigma = S(m,sk,r) ] =
\frac{1}{|\mathcal{R}_{1,sk}(m)|}$.

\item Define two distributions $D$ and $D'$ over pairs $(m,\sigma)$
as follows : $D$ is obtained by choosing $m \in_U \mathcal{M}$
and choosing $\sigma \in_U \mathcal{R}_{1,sk}(m)$. $D'$ is the distribution
defined to be uniform over the set $\mathcal{R}_{1,sk}$. Then $d_{TV}(D,D')
\le \delta$.

\end{itemize}

\end{definition}

From now on, in the interest of brevity, $\mathcal{M}$ will denote
the ``obvious" message space $\mathcal{M}$
associated with a signature scheme unless mentioned otherwise.
\new{Similarly, the randomness $r$ for the signing algorithm $S$ will always
assumed to be $r \in_U \{0,1\}^{n}$.}

We next recall the standard
notion of existential \new{unforgeability} under RMA (Random Message Attack):

\begin{definition}\label{def:RMA-sec}

A signature scheme $(G,S,V)$ is said to be \emph{$(t,\epsilon)$-RMA secure}
if the following holds: Let $(pk,sk) \leftarrow G(1^n)$. Let
$(m_1, \ldots, m_t)$ be chosen uniformly at random from $\mathcal{M}$.
Let $\sigma_i \leftarrow S(m_i,sk,r) $. Then, for any probabilistic
algorithm $A$ running in time $t$, $$ \mathop{\Pr}_{(pk, sk) , (m_1,
\ldots, m_t) , (\sigma_1, \ldots, \sigma_t)} [A(pk, m_1, \ldots, m_t,
\sigma_1, \ldots, \sigma_t) = (m', \sigma')] \le \epsilon $$ where
$V(m',pk,\sigma')=1$
and $m' \not = m_i$ for all $i =1,\dots,t.$

\end{definition}

Next we need to formally define the notion of hardness of
inverse approximate uniform generation:

\begin{definition}

Let $\mathcal{C}$ be a class of $n$-variable Boolean functions.
$\mathcal{C}$ is said
to be \emph{$(\new{t(n)},\epsilon,\delta)$-hard for inverse approximate
uniform generation} if there is no algorithm $A$ running in time $t(n)$
\new{which is an $(\eps,\delta)$-inverse approximate uniform generation
algorithm for ${\cal C}.$}
\ignore{
following property : Given $f \in \mathcal{C}$ and $f : \{0,1\}^n
\rightarrow \{0,1\}$, let $(x_1, \ldots, x_t) \sim
\mathcal{U}_{f^{-1}(1)}$.  On input $(x_1, \ldots, x_t)$, $A$ runs in
time $t$ and outputs description of a circuit $A'$.

$$ \Pr_{(x_1, \ldots, x_t)} [A' \textrm{ is a }
(1-\delta)\textrm{-sampler for }\mathcal{U}_{f^{-1}(1)}] \le \epsilon $$

}
\end{definition}

Finally, we will also need the definition of an invertible Levin reduction:

\begin{definition}

A binary relation $R$ is said to reduce to another binary relation $R'$
by a
\emph{time-$t$ invertible Levin reduction} if there are three algorithms
$\alpha$, $\beta$ and $\gamma$, each running in time $t(n)$ on instances of
length $n$, with the following property:

\begin{itemize}

\item For every $(x,y) \in R$, it holds that $(\alpha(x), \beta(x,y)) \in R'$;

\item For every $(\alpha(x), z) \in R'$, it holds that
$(x,\gamma(\alpha(x), z)) \in R$.


\end{itemize}

Furthermore, the functions $\beta$ and $\gamma$ are
injective maps with the property that $\gamma(\alpha(x), \beta(x,y)) = y$.

\end{definition}

Note that for any class of functions $\mathcal{C}$, we can define the
binary relation $R_{\mathcal{C}}$ as follows : $(f,x) \in
R_{\mathcal{C}}$ if and only if $f(x) =1$ and $f \in \mathcal{C}$. In
this section, whenever we say that there is an invertible
Levin reduction from class $\mathcal{C}_1$ to class $\mathcal{C}_2$,
we mean that there is an invertible Levin
reduction between the corresponding binary relations $R_{\mathcal{C}_1}$
and $R_{\mathcal{C}_2}$.

\subsubsection{A general hardness result based on signature schemes.}
We now state and prove our main theorem relating signature schemes to hardness
of inverse approximate uniform generation:

\begin{theorem}\label{thm:signatures} Let $(G,S,V)$ be a
$(t,\epsilon)$-RMA secure signature scheme. Suppose that with
probability at least \new{99/100}
a random pair $(pk,sk) \leftarrow G(1^n)$ is $(\delta, \eta)$-special.
Let $\mathcal{C}$ be a class of $n$-variable Boolean functions
such that there is a Levin
reduction from \textsf{CIRCUIT-SAT} to $\mathcal{C}$ running in time
$t'(n)$. Let $\kappa_1$ and $\kappa_2$ be such that  \nnew{$\kappa_1 \le 1- 
\new{2}\cdot (2\eta + \delta +
t'(n)/|\mathcal{M}|)$, 
 $\kappa_2 \le 1- \new{2}t'(n)\cdot
(\eta + \delta)$  and $\epsilon \le (1-\kappa_1) (1-\kappa_2)/ 4$}. If $t_1(\cdot)$
is a time function such that $2t_1 (t '(n)) \le t(n) $, then $\mathcal{C}$
is $(t_1(n), \kappa_1, \kappa_2)$-hard for inverse approximate uniform
generation. \end{theorem}

{
The high-level idea of the proof is simple:  Suppose there were an
efficient algorithm for
the inverse approximate uniform generation problem for ${\cal C}.$
Because of the invertible Levin reduction from \textsf{CIRCUIT-SAT} to
${\cal C}$, there is a signature scheme for which the verification
algorithm (using any given public key) corresponds to a function in ${\cal C}.$
The signed messages $(m_1,\sigma_1),\dots,(m_t,\sigma_t)$ correspond to
points from ${\cal U}_{f^{-1}(1)}$ where $f \in {\cal C}$.
Now the existence of an efficient algorithm for
the inverse approximate uniform generation problem for ${\cal C}$
(i.e. an algorithm which, given points from ${\cal U}_{f^{-1}(1)}$,
can generate more such points) translates
into an algorithm which, given a sample of signed messages, can generate a
new signed message.  But this violates the existential unforgeability under
RMA of the signature scheme.

We now proceed to the formal proof.
}

\begin{proof} Assume towards a contradiction
that there is an algorithm $A$ for inverse approximate uniform generation
$A_{\inv}$ which runs in time $t_1$ such that with probability
\new{$1-\kappa_2$}, the output distribution is $\kappa_1$-close to the target
distribution. 
If we can show that for any $(\delta,\eta)$-special key pair 
$(pk,sk)$ the resulting signature scheme is not $(t,\epsilon)$ secure,
then this will result in a contradiction. We will now use algorithm $A$ to 
construct an adversary which breaks the signature scheme for 
$(\delta,\eta)$-special key pairs $(pk,sk).$

Towards this, fix a $(\delta,\eta)$-special key pair $(pk,sk)$ and 
consider the function $V_{pk} : \mathcal{M} \times
\{0,1\}^{\ast} \rightarrow \{0,1\}$ defined as $V_{pk}(m,\sigma) =
V(m,pk,\sigma)$. Clearly, $V_{pk}$ is an instance of
\textsf{CIRCUIT-SAT}
\new{(i.e. $V_{pk}$ is computed by a satisfiable polynomial-size
Boolean circuit)}.  Since there is an invertible
Levin reduction from \textsf{CIRCUIT-SAT} to
$\mathcal{C}$, given $pk$, the adversary in time $t'(n)$ can compute
$\Phi_{pk} \in \mathcal{C}$ with the following properties (let $\beta$ and
$\gamma$ be the corresponding algorithms in the definition of the Levin
reduction): 

\begin{itemize} 

\item For every $(m,\sigma)$ such that
$V_{pk}(m,\sigma) =1$, $\Phi_{pk}(\beta(V_{pk},(m,\sigma)))=1$. 

\item For every $x$ such that $\Phi_{pk}(x)=1$, $V_{pk}(\gamma(\Phi_{pk}, 
x))=1$.

\end{itemize} 

Recall that the adversary receives signatures $(m_1, \sigma_1),
\ldots, (m_{t'(n)}, \sigma_{t'(n)})$. Let $x_i = \beta(V_{pk},(m_i,\sigma_i))$.
Let $D_x$ be the distribution of $(x_1, \ldots, x_{t'(n)})$. We next make
the following claim. 

\begin{claim}\label{clm:stat} 
Let $y_1, \ldots, y_{t'}$ be drawn uniformly at random from 
$\Phi_{pk}^{-1}(1)$ and let $D_y$ be the corresponding distribution
of $(y_1,\dots,y_t)$.  Then, $D_y$ and $D_x$ are $t'\new{(n) \cdot}
(2 \eta + \delta)$-close in statistical distance. 
\end{claim} 

\begin{proof} Note
that $D_y$ and $D_x$ are $t'\new{(n)}$-way product distributions.  If
$D_x^{(1)}$ and $D_y^{(1)}$ are the corresponding marginals on the first
coordinate, then $t'(n) \cdot d_{TV}(D_x^{(1)} ,D_y^{(1)} ) \le d_{TV}(D_x,
D_y)$.  Thus, it suffices to upper bound $d_{TV}(D_x^{(1)} ,D_y^{(1)})$,
which we now do.

$$ d_{TV}(D_x^{(1)} ,D_y^{(1)} ) \le \sum_{z \in 
\mathop{supp}(D_y^{(1)}) \setminus \mathop{supp}(D_x^{(1)}) }
\left|D_x^{(1)}(z) -D_y^{(1)}(z)\right|+\sum_{z \in
\mathop{supp}(D_x^{(1)}) } \left|D_x^{(1)}(z) -D_y^{(1)}(z)\right|.
$$
By definition of $(pk,sk)$ being $(\delta, \eta)$-special, we
get that $$ \sum_{z \in \mathop{supp}(D_y^{(1)}) \setminus
\mathop{supp}(D_x^{(1)}) } |D_x^{(1)}(z) -D_y^{(1)}(z)| \le \eta.$$ 

To bound the next sum, let $\tau = \Pr[ D_y^{(1)}\in \mathop{supp}(D_x^{(1)})
]$.
Note that $\tau  \ge 1- \eta$.  We have

\begin{eqnarray*} \sum_{z \in
\mathop{supp}(D_x^{(1)}) } \left|D_x^{(1)}(z) -D_y^{(1)}(z)\right|
&\le& \sum_{z \in \mathop{supp}(D_x^{(1)}) } \left|\tau D_x^{(1)}(z)
-D_y^{(1)}(z)\right| + (1-\tau) \sum_{z \in \mathop{supp}(D_x^{(1)})
}D_x^{(1)}(z) \\ &\le&\eta+ \tau\cdot \sum_{z \in
\mathop{supp}(D_x^{(1)}) } \left| D_x^{(1)}(z)
-\frac{D_y^{(1)}(z)}{\tau}\right|. 
\end{eqnarray*} 
We observe that
$\frac{D_y^{(1)}(z)}{\tau}$ restricted to $ \mathop{supp}(D_x^{(1)})
$ is simply the uniform distribution over the image of the set
$\mathcal{R}_{1,\new{sk}}$ and hence is the same as applying the map $\beta$ 
on the distribution $D'$. Likewise $D_x^{(1)}$ is the same as applying the map
$\beta$ on $D$ (mentioned in Definition~\ref{def:sign1}). Hence, we 
have that
$$ d_{TV}(D_x^{(1)} ,D_y^{(1)} ) \le 2 \eta + d_{TV} (D, D') \le 2
\eta + \delta.
$$

%
\end{proof}

Now, observe that the instances $x_i$ are each of length at most $t'(n)$.
Since the distributions $D_x$ and $D_y$ are $t'\new{(n) \cdot}
(2 \eta + \delta)$ close,
hence our adversary can run $A_{\inv}$ in time $t(n)$ on the examples
$x_1, \ldots , x_{t'(n)}$ and succeed with probability $1-\kappa_2 - t'
\new{(n)\cdot}(2\eta
+ \delta) \ge (1-\kappa_2)/2 $ in producing a sampler whose output
distribution is $\kappa_1$-close to $\mathcal{U}_{\Phi_{pk}^{-1}(1)}$.
Call this output distribution $Z$. Let $\beta(D)$ 
denote the distribution obtained
by applying the map $\beta$ on $D$. The proof of Claim~\ref{clm:stat} shows
that $\beta(D)$ is $(2\eta + \delta)$-close to the distribution
$\mathcal{U}_{\Phi_{pk}^{-1}(1)}$. Thus, with probability $(1-\kappa_2) /2
$, $Z$ is $(\kappa_1 + (2\eta + \delta))$-close to the distribution
$\beta(D)$. By definition of $D$, we have

$$ \Pr_{(m, \sigma) \in D} [\forall i \in [t'], m_i \not = m] \ge 1 -
\frac{t'}{|\mathcal{M}|}. $$

Thus, with probability $\frac{1-\kappa_2}{2}$,

$$ \Pr_{z \in Z} [z = g(m,\sigma) \textrm{ and } \forall i \in [t'], m_i
\not = m] \ge 1-\kappa_1 - (2\eta + \delta) - \frac{t'}{|\mathcal{M}|} \ge
\frac{1-\kappa_1}{2} $$

Thus, with overall probability $(1-\kappa_1)(1- \kappa_2) /4 \ge \epsilon$,
the adversary succeeds in producing $z = g(m,\sigma)$ such that $
\forall i \in [t'], m_i \not = m$.  Applying the map $\gamma$ on $(\Phi_{pk},
z)$, the adversary gets the pair $(m, \sigma)$. Also, note that the
total running time of the adversary is $t_1(t'(n)) + t'(n) \le 2
t_1(t'(n)) \le t(n)$ which contradicts the $(t,\epsilon)$-RMA security
of the signature scheme. 
\end{proof}

\subsubsection{A specific hardness assumption.}
At this point, at the cost of sacrificing some generality, we consider a
particular instantiation of a signature scheme from the literature which meets
our requirements. While similar signature schemes can be constructed
under many different cryptographic assumptions in the literature, we
forsake such generality to keep the discussion from getting too cumbersome.

To state our cryptographic assumption, we need the following notation:

\begin{itemize}

\item $\textrm{PRIMES}_k$ is the set of $k$-bit prime numbers.

\item $\textrm{RSA}_k$ is the set of all products of two primes of length
$\lfloor (k-1)/2 \rfloor$.

\end{itemize}

The following cryptographic
assumption (a slight variant of the standard RSA assumption) appears in
\cite{MRV99}.

\begin{assumption}\label{ass:1}

\textbf{The RSA$'$ $s(k)$ assumption}: Fix any $m \in
\textrm{RSA}_k$ and let $x \in_U \mathbb{Z}_m^{\ast}$ and $p \in_U
\textrm{PRIMES}_{k+1}$. Let $A$ be any probabilistic algorithm running
in time $s(k)$. Then,

$$ \Pr_{(x,p)} [A(m,x,p) = y \textrm{ and } y^p = x \ (mod \ m)]
\le \frac{1}{s(k)}.$$


\end{assumption}

As mentioned in \cite{MRV99}, given the present state of computational
number theory, it is plausible to conjecture the RSA$'$ $s(k)$ assumption
for $s(k) = 2^{k^{\delta}}$ for some absolute constant $\delta>0$.
For the sake of conciseness, for the rest of this section
we write ``Assumption~\ref{ass:1} holds true'' to mean that
Assumption~\ref{ass:1} holds true with $s(k) = 2^{n^\delta}$
for some fixed constant $\delta > 0$.
(We note, though, that all our hardness results go through giving
superpolynomial hardness using only $s(k) =k^{\omega(1)}$.)

Micali \etal~\cite{MRV99} give a construction of a
``unique signature scheme" using Assumption~\ref{ass:1}:

\begin{theorem}\label{thm:unique-sign}
If Assumption~\ref{ass:1} holds true,
then there is a $(t=2^{n^\delta}, \new{\epsilon=1/t})$-RMA
secure signature scheme
$(G,S,V)$ with the following property : For any message $m \in
\mathcal{M}$, there do not exist $\sigma_1 \not = \sigma_2$ such that
$V(m, \sigma_1) = V(m, \sigma_2) =1$.
In this scheme the signing algorithm $S$ is deterministic and the message
space $\mathcal{M}$ is of size $2^{n^{\delta}}$. \end{theorem}

The above theorem says that under the RSA$'$ $s(k)$ assumption,
there is a deterministic signature scheme such that there is
only one signature $\sigma_m$ for every message $m$,
 and for every message $m$ the only
accepting input for $V$ is $(m,\sigma_m)$. As a consequence,
the signature scheme in Theorem~\ref{thm:unique-sign} has
the property that \new{every $(pk,sk)$ pair that can be generated
by $G$ is $(0,0)$-special}.

\nnew{\begin{remark}
{\em It is important to note here that constructions of $(0,0)$ special signature schemes are abundant in the literature. 
A partial list follows : Lysyanskaya~\cite{Lys02} constructed a deterministic $(0,0)$ special signature scheme using a strong version of the Diffie--Hellman assumption. Hohenberger and Waters~\cite{HW:10} constructed a scheme with a similar guarantee using a variant of the Diffie--Hellman assumption on bilinear groups. In fact, going back  much further, Cramer and Shoup~\cite{CS00, Fischlin03} show that using the Strong RSA assumption, one can get a $(0,0)$ special signature scheme (which however is not deterministic). We remark that the scheme as stated in~\cite{CS00} is not $(0,0)$ special in any obvious sense, but the more efficient version in~\cite{Fischlin03} can be easily verified to be $(0,0)$ special. Throughout this section, for the sake of simplicity, we use the signature scheme in Theorem~\ref{thm:unique-sign}. }
\end{remark}
}
Instantiating Theorem~\ref{thm:signatures} with the signature
scheme from Theorem~\ref{thm:unique-sign}, we obtain the following
corollary:

\begin{corollary}\label{cor:inverse-uniform}
Suppose that Assumption~\ref{ass:1} holds true.  Then  the
following holds : Let $\mathcal{C}$ be a function class such that there
is a polynomial time ($n^k$-time) invertible Levin reduction from
\textsf{CIRCUIT-SAT} to $\mathcal{C}$. Then $\mathcal{C}$ is $(2^{n^c},
1-2^{-n^c}, 1-2^{-n^c})$-hard for inverse approximate uniform generation
for some constant $c>0$ (depending only on the ``$\delta$'' in
Assumption~\ref{ass:1} and on $k$).
\end{corollary}

\subsubsection{Inverse approximate uniform generation
hardness results for specific function classes whose satisfiability
problem is NP-complete.}
In this subsection we use Corollary~\ref{cor:inverse-uniform} to
prove hardness results for inverse approximate uniform generation
for specific function classes $\mathcal{C}$ for which there are
invertible Levin reductions from \textsf{CIRCUIT-SAT} to $\mathcal{C}$.

Recall that a 3-CNF formula is a conjunction of clauses (disjunctions)
of length 3.
The following fact can be easily verified by inspecting the
standard reduction from \textsf{CIRCUIT-SAT} to \textsf{3-CNF-SAT}.

\begin{fact}\label{fact:c-sat} There is a polynomial time invertible
Levin reduction from \textsf{CIRCUIT-SAT} to \textsf{3-CNF-SAT}.
\end{fact}

As a corollary, we have the following result.

\begin{corollary}\label{corr:hardness-3-SAT}
If Assumption~\ref{ass:1} holds true, then
there exists an absolute constant $c>0$ such that the class
\textsf{3-CNF} is $(2^{n^c}, 1-2^{-n^c}, 1-2^{-n^c})$-hard for inverse
approximate uniform generation. \end{corollary}

Corollary~\ref{corr:hardness-3-SAT} is interesting in light of the
well known fact that the class of all 3-CNF formulas is
efficiently PAC learnable from uniform random examples (in fact
under any distribution).

We next observe that the problem of inverse approximate uniform generation
remains hard even for 3-CNF formulas in which each variable occurs a bounded
number of times. To prove this we will use the fact that
polynomial time invertible Levin reductions compose:
\begin{fact}\label{fac:compose} If there is a polynomial time
invertible Levin reduction from \textsf{CIRCUIT-SAT} to $\mathcal{C}$
and a polynomial time Levin reduction from $\mathcal{C}$ to
$\mathcal{C}_1$, then there is a polynomial time invertible Levin
reduction from \textsf{CIRCUIT-SAT} to $\mathcal{C}_1$.
\end{fact}

 The following theorem says that the inverse approximate uniform
generation problem remains hard for the class of all 3-CNF
formulas in which each variable occurs at most 4 times (hereafter
denoted \textsf{3,4-CNF}).

\begin{theorem}\label{thm:hardness-bounded-3-SAT}
If Assumption~\ref{ass:1} holds true, then there
exists an absolute constant $c>0$ such that
\textsf{3,4-CNF-SAT} is $(2^{n^c},
1-2^{-n^c}, 1-2^{-n^c})$-hard for inverse approximate uniform generation.
\end{theorem}

\begin{proof} Tovey~\cite{Tov:84} shows that there is a polynomial time
invertible Levin reduction from \textsf{3-CNF-SAT} to \textsf{3,4-CNF-SAT}.
Using Fact~\ref{fac:compose}, we have a polynomial time Levin reduction
from \textsf{CIRCUIT-SAT} to \textsf{3,4-CNF-SAT}.  Now the result
follows from Corollary~\ref{cor:inverse-uniform}
\end{proof}

The next theorem shows that the class of all intersections
of two halfspaces over $n$ Boolean variables is
hard for inverse approximate uniform generation.

\begin{theorem}\label{thm:degree2hardness}
If Assumption~\ref{ass:1} holds true, then there exists an
absolute constant $c>0$ such that
${\cal C}=\{$all intersections of two halfspaces
over $n$ Boolean variables$\}$
is $(2^{n^c}, 1- 2^{-n^c}, 1-2^{-n^c})$-hard for inverse approximate uniform
generation.
\end{theorem}

\begin{proof} We recall that the \textsf{SUBSET-SUM} problem is defined
as follows : An instance $\Phi$ is defined by positive integers $w_1,
\ldots, w_n, s >0$. A satisfying assignment for this instance is given
by $x \in \{0,1\}^n$ such that $\sum_{i=1}^n w_i x_i =s$. It is
well known that the \textsf{SUBSET-SUM} problem is NP-complete and it is
folklore that there is a invertible Levin reduction from \textsf{3-SAT}
to \textsf{SUBSET-SUM}.  However, since it is somewhat difficult to find this
reduction explicitly in the literature, we outline such a reduction.

To describe the reduction, we first define \textsf{1-in-3-SAT}. An instance
$\Psi$ of \textsf{1-in-3-SAT} is defined over
Boolean variables $x_1,
\ldots, x_n$ with the following constraints : The $i^{th}$ constraint is
defined by a subset of at most three literals over $x_1, \ldots,
x_n$.  An assignment to $x_1, \ldots, x_n$ satisfies $\Psi$ if and only
if for every constraint there is exactly one literal which is set to
true. Schaefer~\cite{Sch78} showed that \textsf{3-SAT} reduces to
\textsf{1-in-3-SAT} in polynomial time, and the reduction can be easily
verified to be an invertible Levin reduction. Now the
standard textbook reduction from \textsf{3-SAT} to \textsf{SUBSET-SUM}
(which can be found e.g. in \cite{Papadimtriou}) applied to instances
of \textsf{1-in-3-SAT}, can be easily seen to be a
polynomial time invertible Levin reduction. By Fact~\ref{fac:compose},
we thus have a polynomial time invertible Levin reduction from
\textsf{3-CNF-SAT} to \textsf{SUBSET-SUM}.

With this reduction in hand, it remains only to observe that
that any instance of \textsf{SUBSET-SUM} is also an instance of
``intersection of two halfspaces,''
simply because $\sum_{i=1}^n w_i x_i =s$
if and only if $s \leq \sum_{i=1}^n w_i \cdot x_i  \leq s.$
Thus, there is a
polynomial time invertible Levin reduction from \textsf{3-CNF-SAT} to
the class of all intersections of two halfspaces. This finishes the proof.
 \end{proof}

\subsubsection{A hardness result where the satisfiability problem is in $P$.}

So far all of our hardness results have been for classes $\mathcal{C}$ 
of NP-complete languages.  As Theorem~\ref{thm:signatures} requires a 
reduction from \textsf{CIRCUIT-SAT} to $\mathcal{C}$, this theorem 
cannot be directly used to prove hardness for classes $\mathcal{C}$ 
which are not NP-hard. We next give an extension of 
Theorem~\ref{thm:signatures} which can apply to classes $\mathcal{C}$ 
for which the satisfiability problem is in $P$. Using this result we 
will show hardness of inverse approximate uniform generation for 
\textsf{MONOTONE-2-CNF-SAT}. (Recall that a monotone 2-CNF formula is a 
conjunction of clauses of the form $x_i \vee x_j$, with no negations; 
such a formula is trivially satisfiable by the all-true assignment.)

We begin by defining by a notion of invertible one-many reductions that 
we will need.

 \begin{definition}\label{def:one-many}

\textsf{CIRCUIT-SAT} is said to have an
\emph{$\eta$-almost invertible one-many reduction to a function class
$\mathcal{C}$} if the following conditions hold:

\begin{itemize}

\item There is a polynomial time computable function $f$ such that
given an instance $\Phi$ of \textsf{CIRCUIT-SAT} (i.e.
$\Phi$ is a satisfiable circuit),
$\Psi = f(\Phi)$ is an instance of $\mathcal{C}$
(i.e. $\Psi \in \mathcal{C}$ and $\Psi$ is satisfiable).

\item Fix any instance $\Phi$ of \textsf{CIRCUIT-SAT} and let
$\mathcal{A}=\Psi^{-1}(1)$ denote the set of satisfying assignments of $\Psi$.
Then $\mathcal{A}$ can be partitioned into sets
$\mathcal{A}_1$ and $\mathcal{A}_2$ such that $|\mathcal{A}_2|/|\mathcal{A}|
\le \eta$ and there is an efficiently computable function
$g: \mathcal{A}_1 \to \Phi^{-1}(1)$
such that
$g(x)$ is a satisfying assignment of $\Phi$
for every $x \in \mathcal{A}_1$.

\item For every $y$ which is a satisfying assignment of $\Phi$, the number
of pre-images of $y$ under $g$ is exactly the same,
and the uniform distribution over $g^{-1}(y)$ is polynomial time samplable.

\end{itemize}

\end{definition}

We next state the following simple claim which will be helpful later.

\begin{claim}\label{clm:stat-dis1}
Suppose there is an $\eta$-almost invertible one-many reduction
from \textsf{CIRCUIT-SAT} to $\mathcal{C}$.
Let $f$ and $g$ be the functions from Definition~\ref{def:one-many}.
Let $\Phi$ be an
instance of \textsf{CIRCUIT-SAT} and let
$\Psi = f(\Phi)$ be the corresponding instance of $\mathcal{C}$.
Define distributions $D_1$ and $D_2$ as follows :

 \begin{itemize}

 \item A draw from $D_1$ is obtained by choosing
$y$ uniformly at random from $\Phi^{-1}(1)$ and then
outputting $z$ uniformly at random from $g^{-1}(y)$.

\item A draw from $D_2$ is obtained by choosing
$z'$ uniformly at random from $\Psi^{-1}(1)$.

\end{itemize}

Then we have $d_{TV}(D_1, D_2) \le \eta$.

 \end{claim}

 \begin{proof}
This is an immediate consequence of the fact that
$D_1$ is uniform over the set $\mathcal{A}_1$ while $D_2$ is
uniform over the set $\mathcal{A}$ (from Definition~\ref{def:one-many}).
\end{proof}

We next have the following extension of
Corollary~\ref{cor:inverse-uniform}.

\begin{theorem}\label{thm:monotone}
Suppose that Assumption~\ref{ass:1} holds true.
Then if $\mathcal{C}$ is a function class such that there is an $\eta$-almost
invertible one-many reduction (for $\eta= 2^{-\Omega(n)}$) from
\textsf{CIRCUIT-SAT} to $\mathcal{C}$, then $\mathcal{C}$ is $(2^{n^c},
1-2^{-n^c}, 1-2^{-n^c})$-hard for inverse approximate uniform generation
for some absolute constant $c>0.$
\end{theorem}

\begin{proof}

The proof is similar to the proof of Corollary~\ref{cor:inverse-uniform}.
Assume towards a contradiction that there is an algorithm for inverse
approximation uniform generation $A_{\inv}$ for $\mathcal{C}$ which runs
in time $t_1$ such that with probability \new{$1-\kappa_2$}, the output
distribution is $\kappa_1$-close to the target distribution. (We will set
$t_1$, $\kappa_1$ and $\kappa_2$ later to $2^{n^c}$, $1-2^{-n^c}$ and
$1-2^{-n^c}$ respectively.) 
\ignore{But for the moment, we will work with $t_1$,
$\kappa_1$ and $\kappa_2$.
}

Let $(G,S,V)$ be the RMA-secure signature scheme constructed in
Theorem~\ref{thm:unique-sign}. Note that $(G,S,V)$ is a 
$(T, \epsilon)$-RMA secure signature scheme where 
$T = 2^{n^{\delta}}$, $\epsilon = 1/T$
and $|\mathcal{M}| = 2^{n^\mu}$ for constant $\delta, \mu>0$.  Let
$(pk,sk)$ be a choice of key pair. We will us $A_{\inv}$ to contradict
the security of $(G,S,V)$.  Towards this, consider the function $V_{pk}
: \mathcal{M} \times \{0,1\}^{\ast} \rightarrow \{0,1\}$ defined as
$V_{pk}(m,\sigma) = V(m,pk,\sigma)$. Clearly, $V_{pk}$ is an instance of
\textsf{CIRCUIT-SAT}.  Consider the $\eta$-invertible one-many reduction
from \textsf{CIRCUIT-SAT} to $\mathcal{C}$. Let $\alpha$ and $\beta$ 
have the
same meaning as in Definition~\ref{def:one-many}. Let $\Psi = \alpha(V_{pk})$
and let $\mathcal{A}$, $\mathcal{A}_1$ and $\mathcal{A}_2$ have the same
meaning as in Definition~\ref{def:one-many}. The adversary receives
message-signature pairs $(m_1, \sigma_1)$ $\ldots$ $(m_{t_1},
\sigma_{t_1})$ where $m_1, \ldots, m_{t_1}$ are chosen independently at
random from $\mathcal{M}$. For any $i$, $(m_i, \sigma_i)$ is a
satisfying assignment of $V_{pk}$. By definition, in time $t_2 = t_1
\cdot \poly(n)$, the adversary can sample $(z_1, \ldots, z_{t_1})$ such
that $z_1, \ldots, z_{t_1}$ are independent and $z_i \sim
\mathcal{U}_{\beta^{-1}(m_i,\sigma_i)}$. Note that this means that each
$z_i$ is an independent sample from $\mathcal{A}_1$ and $|z_i| =
\poly(n)$. Note that $(z_1, \ldots, z_{t_1})$ is a $t_1$-fold product
distribution such that if $D'$ denotes the distribution of $z_i$, then
by Claim~\ref{clm:stat-dis1}, $d_{TV}( D' , \mathcal{U}_{\Psi^{-1}(1)} )
\le \eta$. Hence, if $D$ is the distribution of $(z_1, \ldots,
z_{t_1})$, then $d_{TV}(D, \mathcal{U}^t_{\Psi^{-1}(1)}) \le t_1 \eta$.

Hence, the adversary can now run $A_{rec}$ on the samples $z_1, \ldots,
z_{t_1}$ and as long as $1-\kappa_2 - t_1 \eta \ge (1-\kappa_2)/2$, succeeds
in producing a sampler with probability $(1-\kappa_2)/2$ whose output
distribution (call it $Z$) is $\kappa_1$ close to the distribution
$\mathcal{U}_{\Psi^{-1}(1)}$. Note that as $\eta = 2^{-\Omega(n)}$, for
any $c>0$, $t_1 = 2^{n^c}$ and $\kappa_2 =1- 2^{-n^c}$ satisfies this
condition. Hence, we get that $d_{TV}(Z, D') \le  \kappa_1 + \eta$.
Now, observe that 
$$\Pr_{\rho\in D'} [\beta(\rho) = (m, \sigma)
\textrm{ and } m \not =m_i] = 1 - \frac{t_1}{|\mathcal{M}|}.$$ 
The above
uses the fact that every element in the range of $\beta$ has the same number
of pre-images. This of course implies that

$$ \Pr_{\rho\in Z} [\beta(\rho) = (m,\sigma) \textrm{ and } m \not
=m_i] \ge 1 - \frac{t_1}{|\mathcal{M}|} - ( \kappa_1 + \eta).
$$

Again as long as $\kappa_1 \le 1 - 2 ( \eta + t_1/|\mathcal{M}| )$, the
adversary succeeds in getting a valid message signature pair
$(m,\sigma)$ with $m \not = m_i$ for any $1 \le i \le t_1$ with probability $(1-\kappa_1)/2$. Again, we
can ensure $\kappa_1 \le 1- 2 ( \eta + t_1/|\mathcal{M}| )$ by choosing $c$
sufficiently small compared to $\mu$.
The total probability of success is $(1-\kappa_1 )(1-\kappa_2)/ 4$ and the 
total running time  is $t_1( \poly(n)) + \poly(n)$. Again if $c$ is sufficiently small compared to $\mu$ and $\delta$, then the 
total running time is at most
$t_1( \poly(n)) + \poly(n) < T$ and the success probability is at least
$(1-\kappa_1) (1-\kappa_2) / 4 > \epsilon,$ resulting in a contradiction.
\end{proof}

We now demonstrate a polynomial time $\eta$-invertible one-many
reduction from \textsf{CIRCUIT-SAT} to \textsf{MONOTONE-2-CNF-SAT} for
$\eta= 2^{-\Omega(n)}$. The reduction uses the ``blow-up" idea
used to prove hardness of approximate counting for
\textsf{MONOTONE-2-CNF-SAT} in \cite{JVV86}. We will closely follow the
instantiation of this technique in \cite{Watson:12}.

 \begin{lemma}\label{lem:blow-up}

 There is a polynomial time $\eta$-almost invertible one-many reduction
from \textsf{CIRCUIT-SAT} to \textsf{MONOTONE-2-CNF-SAT} where $\eta
=2^{-\Omega(n)}$.

 \end{lemma}

 \begin{proof}

 We begin by noting the following simple fact.

 \begin{fact}\label{fac:comp1}

 If there is a polynomial time invertible Levin reduction from
\textsf{CIRCUIT-SAT} to a class $\mathcal{C}_1$ and an $\eta$-almost
invertible one-many reduction from $\mathcal{C}_1$ to $\mathcal{C}_2$,
then there is a polynomial time $\eta$-almost invertible one-many
reduction from \textsf{CIRCUIT-SAT} to $\mathcal{C}_2$.

 \end{fact}

 Since there is an invertible Levin reduction from \textsf{CIRCUIT-SAT} to
\textsf{3-CNF-SAT}, by virtue of Fact~\ref{fac:comp1}, it suffices to
demonstrate a polynomial time $\eta$-almost invertible one-many
reduction from \textsf{3-CNF-SAT}
to \textsf{MONOTONE-2-CNF-SAT}. To do this, we
first construct an instance of \textsf{VERTEX-COVER} from the
\textsf{3-CNF-SAT} instance. Let $\Phi = \new{\bigwedge}_{i=1}^m \Phi_i$ be the
instance of \textsf{3-CNF-SAT}. Construct an instance of
\textsf{VERTEX-COVER} by introducing seven vertices for each clause
$\Phi_i$ (one corresponding to every satisfying assignment of $\Phi_i$).
Now, put an edge between any two vertices of this graph if the
corresponding assignments to the variables of $\Phi$ conflict on some
variable. We call this graph $G$.  We observe the following properties
of this graph :

 \begin{itemize}

 \item $G$ has exactly $7m$ vertices.

 \item Every vertex cover of $G$ has size at least $6m$.

 \item There is an efficiently computable and invertible injection
$\ell$ between the satisfying assignments of $\Phi$ and the vertex
covers of $G$ of size $6m$. To get the vertex cover corresponding to a
satisfying assignment, for every clause $\Phi_i$, include the six
vertices in the vertex cover which conflict with the satisfying
assignment.

 \end{itemize}

 We next do the blow-up construction. We create a new graph $G'$ by
replacing every vertex of $G$ with a cloud of $10 m$ vertices, and
for every edge in $G$ we create a complete bipartite graph
between the corresponding clouds in $G'$. Clearly, the size of the
graph $G'$ is polynomial in the size of the \textsf{3-CNF-SAT} formula.
We define a map
$g_1$ between vertex covers of $G'$ and vertex covers of $G$ as follows
: Let $S'$ be a vertex cover of $G'$. We define the set $S = g_1(S')$ in
the following way. For every vertex $v$ in the graph $G$,
if all the vertices in the corresponding cloud in $G'$ are in
$S'$, then include $v \in S$, else do not include $v$ in $S$. It is easy
to observe that $g_1$ maps vertex covers of $G'$ to vertex covers of
$G$. It is also easy to observe that a vertex cover of $G$ of size $s$
has $(2^{10m}-1)^{7m-s}$ pre-images under $g_1$.

 Now, observe that we can construct a \textsf{MONOTONE-2-CNF-SAT} formula
$\Psi$ which has a variable corresponding to every vertex in $G'$ and
every subset $S'$ of $G'$ corresponds to a truth assignment $y_{S'}$ to
$\Psi$ such that $\Psi(y_{S'}) =1$ if and only if $S'$ is a vertex cover
of $G'$. Because of this correspondence, we can construct a map
 $g_1'$ which maps satisfying assignments of $\Psi$ to vertex covers of $G$.  Further, a vertex cover of size $s$ in graph $G$  has $(2^{10m}-1)^{7m-s}$ pre-images under $g_1'$.
 Since the total number of vertex covers of $G$ of size $s$ is at most
$\binom{7m}{s}$, the total number of satisfying assignments of
$\Psi$ which map to vertex covers of $G$ of size more than $6m$ can be
bounded by :
 $$
 \sum_{s=6m+1}^{7m} \binom{7m}{s} \cdot (2^{10m}-1)^{7m-s} \le m \cdot \binom{7m}{6m+1} \cdot (2^{10m}-1)^{m-1} \le (2^{10m}-1)^{m} \cdot \frac{2^{7m}}{2^{10m}-1}
 $$

 On the other hand, since $\Phi$ has at least one satisfying assignment,
hence $G$ has at least one vertex cover of size $6m$ and hence the total
number of satisfying assignments of $\Psi$ which map to vertex covers of
$G$ of size $6m$ is at least $(2^{10m}-1)^{m}$. Thus, if we let
$\mathcal{A}$ denote the set of satisfying assignments of $\Psi$ and
$\mathcal{A}_1$ be the set of satisfying assignment of $\Psi$ which map
to vertex covers of $G$ of size exactly $6m$ (under $g_1$), then
$|\mathcal{A}_1|/|\mathcal{A}| \ge 1- 2^{-\Omega(n)}$. Next, notice that
we can define the map $g$ mapping $\mathcal{A}_1$ to the satisfying
assignments of $\Phi$ in the following manner : $g(x) =
\ell^{-1}(g_1(x))$. It is easy to see that this map satisfies all the
requirements of the map $g$ from Definition~\ref{def:one-many} which
concludes the proof.
 \end{proof}

 Combining Lemma~\ref{lem:blow-up} with Theorem~\ref{thm:monotone}, we
have the following corollary.

 \begin{corollary}\label{corr:monotone-2-sat}
If Assumption~\ref{ass:1} holds true, then \textsf{MONOTONE-2-CNF-SAT} is
$(2^{n^c}, 1-2^{-n^c}, 1-2^{-n^c})$ hard for inverse approximate uniform
generation for some absolute constant $c>0.$
 \end{corollary}

As a consequence of the above result,
we also get hardness for inverse approximate uniform
generation of degree-2 \emph{polynomial threshold functions (PTFs)};
these are functions of the form $\sign(q(x))$ where
$q(x)$ is a degree-2 multilinear polynomial over
$\{0,1\}^n.$

 \begin{corollary}\label{corr:hardness-PTF}
If Assumption~\ref{ass:1} holds true, then the class of all $n$-variable
degree-2 polynomial threshold functions is
$(2^{n^c}, 1-2^{-n^c}, 1-2^{-n^c})$ hard for inverse approximate uniform
generation for some absolute constant $c>0.$
 \end{corollary}

\begin{proof}
This follows immediately from the fact that every
monotone 2-CNF formula can be expressed as a degree-2 PTF.
To see this, note that if $\Phi = \bigwedge_{i=1}^m (x_{i1} \vee x_{i2})$
where each $x_{ij}$ is a 0/1 variable,
then $\Phi(x)$ is true if and only if
$\sum_{i=1}^m x_{i1} + x_{i2} - x_{i1} \cdot x_{i2} \geq m.$
This finishes the proof.
 \end{proof}

\subsection{Hardness results based on Message Authentication Codes.}
All of the previous hardness results intuitively correspond to the case 
when the second step of our 
``standard approach'' is algorithmically hard. Indeed, consider a 
class $\C$ of functions that has an efficient approximate uniform generation
algorithm. Unless $P \not = NP$ there 
cannot be any Karp reduction from \textsf{CIRCUIT-SAT} to $\mathcal{C}$ 
(this would contradict the NP-completeness of \textsf{CIRCUIT-SAT}) and 
hence Theorem~\ref{thm:signatures} is not applicable in this setting. In 
fact, \new{even for $\eta = 1-1/\poly(n)$} there cannot be any 
$\eta$-almost invertible one-many reduction from 
\textsf{CIRCUIT-SAT} to $\mathcal{C}$ unless 
$P \not = NP$. This makes Theorem~\ref{thm:monotone} inapplicable 
in this setting. Thus, to prove hardness results for classes that
have efficient approximate uniform generation algorithms,
we need some other approach.

In this section we show that Message Authentication Codes (MAC) can be 
used to establish hardness of inverse approximate uniform generation
for such classes. We begin by 
defining MACs. (We remark that we use a restricted definition which is 
sufficient for us; for the most general definition, see \cite{gol-fou02}.)

\begin{definition} A \emph{Message Authentication Code (MAC)} is a triple 
$(G,T,V)$ of \new{polynomial-time} algorithms with the following properties :

 \begin{itemize}

\item {\bf (Key generation algorithm)}
$G(\cdot)$ is a randomized algorithm which on input $1^n$ 
produces a secret key $sk$;

\item {\bf (Tagging algorithm)}
$T$ is a randomized algorithm which takes as input message $m$, 
secret key $sk$ and randomness $r$ and outputs $\sigma \leftarrow 
T(m,sk,r)$;

\item {\bf (Verification algorithm)}
$V$ is a deterministic algorithm which takes as input message 
$m$, secret key $sk$ and $\sigma.$  If $\sigma = T(m,sk,r)$
for some $r$ then $V(m,sk,\sigma)=1$.

 \end{itemize}

 \end{definition}

For the purposes of our hardness results we require MACs with some 
special properties. While our hardness results can be derived from 
slightly more general MACs than those we specify below, 
we forsake some generality for the sake of 
clarity. For a MAC $(G,T,V)$ and a choice of secret key $sk$, we say 
$\sigma$ is a \emph{valid tag} for message $m$ if there exists $r$ such that 
$\sigma = T(m,sk,r)$. Likewise, we say that $\sigma$ is a 
\emph{potential tag} for 
message $m$ if $V(m,sk, \sigma)=1$.

 \begin{definition}\label{def:MAC-s}

 A Message Authentication Code $(G,T,V)$ over a message space 
$\mathcal{M}$ is said to be \emph{special} if the following conditions hold : 
For any secret key $sk$,

 \begin{itemize}

 \item For every message $m \in \mathcal{M}$, the set of valid tags is 
identical to the set of potential \new{tags}..

 \item For every two messages
$m_1 \neq m_2$ and every $\sigma_1,\sigma_2$ such that $\sigma_i$ is a valid
tag for $m_i$, we have
$\Pr_r[T(m_1,sk,r)=\sigma_1]
=\Pr_r[T(m_2,sk,r)=\sigma_2].$. In particular, the cardinality of the set of valid tags for $m$ is the same
for all $m$.

 \end{itemize}

 \end{definition}

 We next define the standard notion of security under Random Message 
attacks for MACs.  As before, from now onwards, we will assume implicitly that $\mathcal{M}$ is the message space. 

\begin{definition}\label{def:RMA}
A special MAC $(G,T,V)$ is said to be \emph{$(t,\epsilon)$-RMA secure} 
if the following holds : Let $sk \leftarrow G(1^n)$. Let $(m_1, 
\ldots, m_t)$ be chosen uniformly at random from $\mathcal{M}$. Let 
$\sigma_i \leftarrow T(m_i,sk,r) $. Then for any probabilistic 
algorithm $A$ running in time $t$, $$ \mathop{\Pr}_{sk , (m_1, \ldots, 
m_t) , (\sigma_1, \ldots, \sigma_t)} [A( m_1, \ldots, m_t, \sigma_1, 
\ldots, \sigma_t) = (m', \sigma')] \le \epsilon $$ where 
$V(m',sk, \sigma')=1$ and
$m' \not = m_i$ for all $i=1,\dots,t.$
 \end{definition}

It is known how to construct MACs meeting the requirements 
in Definition~\ref{def:RMA} under standard cryptographic assumptions 
(see \cite{gol-fou02}). 

\subsubsection{A general hardness result based on Message
Authentication Codes.}
The next theorem 
shows that special MACs yield hardness results for
inverse approximate uniform generation.

 \begin{theorem}\label{thm:hardness-MAC}
 Let $c>0$ and $(G,T,V)$ be a $(t,\epsilon)$-RMA secure special MAC for 
some $t= 2^{n^c}$ and $\epsilon = 1/t$ with a message space 
$\mathcal{M}$ of size $2^{\Omega(n)}$. Let $V_{sk}$ denote the 
function $V_{sk} : (m, \sigma) \mapsto V(m,sk,\sigma)$. If 
$V_{sk} \in \mathcal{C}$ for every $s_k$, 
then there exists $\delta>0$ such that 
$\mathcal{C}$ is $(t_1, \kappa,\eta)$-hard for inverse approximate 
uniform generation for $t_1 = 2^{n^{\delta}}$ and 
\new{$\kappa = \eta= 1-2^{-n^{\delta}}$.} \end{theorem}

 \begin{proof}
 Towards a contradiction, let us assume that there is an algorithm 
$A_{\inv}$ for inverse approximate uniform generation of $\mathcal{C}$ 
which runs in time $t_1$ and with probability
$1-\eta$ outputs a sampler whose statistical distance is at most $\kappa$
from the target distribution. (We will set $t_1$, $\kappa$ and 
$\eta$ later in the proof.)  
We will use $A_{\inv}$ to contradict the security of the 
MAC. Let $sk$ be a secret key chosen according to $G(1^n)$. 
Now, the adversary receives message-tag 
pairs $(m_1, \sigma_1), \ldots, (m_{t_1}, \sigma_{t_1})$ where $m_1, 
\ldots, m_{t_1}$ are chosen independently at random from $\mathcal{M}$. 
Because the MAC is special, for each $i$ we have that $\sigma_i$ is a 
uniformly random valid tag for the message $m_i$. Hence each $(m_i, 
\sigma_i)$ is an independent and uniformly random satisfying assignment 
of $V_{sk}$.

We can thus run $A_\inv$ on the samples
$(m_1, \sigma_1), \ldots, (m_{t_1}, \sigma_{t_1})$ 
with its accuracy parameter set to $\kappa$ and its confidence
parameter set to $1-\eta.$  Taking $\kappa = \eta = 1-2^{-n^\delta}$,
we can choose $\delta$ small enough 
compared to $c$, and with $t_1 = 2^{n^{\delta}}$ we get that the total
running time of $A_\inv$ is at most
$2^{n^c}/2$. By the definition of inverse approximate uniform generation,
with probability at least $1-\eta = 2^{-n^\delta}$
the algorithm $A_\inv$ outputs a sampler for a distribution $Z$
that is $\kappa=(1-2^{-n^\delta})$-close to the uniform distribution over 
the satisfying assignments of $V_{sk}$. Now, observe that

$$
\Pr_{(m,\sigma) \sim \mathcal{U}_{V_{sk}^{-1}(1)}} [ 
m_i \not = m \text{~for all~} i \in [t_1]] \ge 1- \frac{t_1}{|\mathcal{M}|}.
 $$

Thus,
 $$
 \Pr_{z \sim Z} [z = (m, \sigma) \textrm{ and }  
m_i \not = m \text{~for all~}i \in [t_1]] \ge 
\new{(1-\kappa)} - \frac{t_1}{|\mathcal{M}|}.
 $$

This means that with probability
\new{ $(1-\eta) \cdot ( (1-\kappa) - 
\frac{t_1}{|\mathcal{M}|} )$,} the adversary can output a forgery. It is 
clear that \new{for a suitable choice of $\delta$ relative to $c$,}
recalling that $\kappa = \eta = 1-2^{-n^{\delta}}$, 
the probability of outputting a forgery is greater than
$2^{-n^c},$ which contradicts the security of the MAC.
  \end{proof}

Unlike signature schemes, which permitted intricate reductions (cf. 
Theorem~\ref{thm:signatures}), in the case of MACs we get a hardness 
result for complexity class $\mathcal{C}$ only if $V_{sk}$ itself
belongs to 
$\mathcal{C}$. While special MACs are known to exist
\nnew{assuming the existence of one-way functions~\cite{gol-fou02}}, 
the constructions are rather involved and  
rely on constructions of pseudorandom functions (PRFs) as an intermediate
step. As a result, the verification algorithm $V$ also involves computing PRFs;
this means that using these standard constructions, one can only get 
hardness results for a class $\mathcal{C}$ if PRFs can be computed in 
$\mathcal{C}$. As a result, the class $\mathcal{C}$ tends to be fairly 
complex, making the corresponding hardness result for inverse 
approximate uniform generation for $\mathcal{C}$ somewhat uninteresting.

  One way to bypass this is to use construction of MACs which do not 
involve use of PRFs as an intermediate step. In recent years there has been 
significant progress in this area~\cite{Jain:11,Dodis:12}.  While both these
papers describe several MACs which do not require PRFs, the one 
most relevant for us is the MAC construction of \cite{Jain:11} based 
on the hardness of the ``Learning Parity with Noise'' (LPN) problem. 

\subsubsection{Some specific hardness assumptions, and a corresponding
specific hardness result.}

We first state a ``decision" version of LPN. To do this, we need the 
following notation:

  \begin{itemize}

  \item Let $\mB_{\tau}$ denote the following distribution over $GF(2)$ 
: If $x \leftarrow \mB_{\tau}$, then $\Pr[x=1] =\tau$.

  \item For $x \in GF(2)^n$, we use $\Lambda(x, \tau,\cdot)$ to denote 
the distribution $(r, x \cdot r \oplus e)$ 
over $GF(2)^n \times GF(2)$ 
where $r \sim GF(2)^n$ and $e 
\sim \mB_{\tau}$ and $x \cdot r = \oplus_{i} x_i r_i \ (\mod \ 2)$.

\end{itemize}

  \begin{assumption}\label{assumption:LPN}
Let $\tau \in (0,1/2)$ and let $\mathcal{O}_{x, \tau}$ be an oracle 
which, each time it is invoked,
returns an independent uniformly random sample from $\Lambda(x, \tau,\cdot)$. 
The LPN assumption states that for any $\poly(n)$-time algorithm 
$\mathcal{A}$,

  $$
  \left| \left[ \Pr_{x \in GF(2)^n} [ \mathcal{A}^{\mathcal{O}_{x,\tau}} 
=1 \right] - \left[ \Pr_{x \in GF(2)^n} [ 
\mathcal{A}^{\mathcal{O}_{x,1/2}} =1 \right] \right| \le \epsilon
  $$

  for some $\epsilon$ which is negligible in $n$.
  \end{assumption}

LPN is a well-studied problem; despite intensive research effort,
the fastest known algorithm 
for this problem takes time $2^{O(n/\log n)}$~\cite{blukalwas03}. For 
our applications, we will need a variant of the above LPN assumption. To 
define the assumption, let $\Lambda(x, \ell, \tau, \cdot)$ 
denote the distribution over $(A, A \cdot x \oplus e)$ where $A$ is 
uniformly random in $GF(2)^{\ell \times n}$ and $e$ is uniformly random 
over the set $\{z \in GF(2)^{\ell} : wt(z) \le \lceil \tau \ell \rceil 
\}$. The vector $e$ is usually referred to as the \emph{noise vector.}

  \begin{assumption}\label{assumption:LPNvar}
  Let $\tau \in (0,1/2)$, $\ell = c \cdot n$ for some $0 < c < 1/2$ and 
let $\mathcal{O}_{x,\ell, \tau}$ be an oracle which returns a uniformly 
random sample from $\Lambda(x, \ell, \tau, \cdot)$. Then the 
\emph{$(t,\epsilon)$ exact LPN assumption} states that for any algorithm 
$\mathcal{A}$ running in time $t$,

  $$
  \left| \Pr_{x \in GF(2)^n} \left[ 
\mathcal{A}^{\mathcal{O}_{x,\ell,\tau}} =1 \right] - \Pr_{x \in GF(2)^n} 
\left[ \mathcal{A}^{\mathcal{O}_{x,\ell,1/2}} =1 \right] \right| \le 
\epsilon
  $$

For the sake of brevity, we henceforth refer to this assumption by
saying ``the exact $(n, \ell, \tau)$ LPN problem is $(t,\epsilon)$-hard.''
  \end{assumption}

The above conjecture seems to be very closely related to 
Assumption~\ref{assumption:LPN}, but it is not known whether
Assumption~\ref{assumption:LPN} formally reduces to 
Assumption~\ref{assumption:LPNvar}. Assumption~\ref{assumption:LPNvar} 
has previously been suggested in the cryptographic literature~\cite{KSS10} 
in the context of getting perfect completeness in LPN-based protocols. 
We note that Arora and Ge~\cite{AG11} have investigated 
the complexity of this problem and gave an algorithm which runs 
in time $n^{O(\ell)}$. We believe that the proximity 
of Assumption~\ref{assumption:LPNvar} to the well-studied 
Assumption~\ref{assumption:LPN}, as well as the failure to find 
algorithms for Assumption~\ref{assumption:LPNvar}, make 
it a plausible conjecture. For the rest of this section we use 
Assumption~\ref{assumption:LPNvar} with $t = 2^{n^{\beta}}$ and 
$\epsilon= 2^{-n^{\beta}}$ for some fixed $\beta>0$.

We next define a seemingly stronger variant of 
Assumption~\ref{assumption:LPNvar} which we call \emph{subset exact LPN}. 
This requires the following definitions:
For $x,v \in GF(2)^n$, $\ell,d \le n$ and $\tau \in (0,1/2)$, we 
define the distribution $\Lambda^a(x,v, \ell, \tau,\new{ \cdot)}$ as follows :

 $$ \Lambda^a(x,v, \ell, \tau,\new{\cdot)} =
\left\{
	\begin{array}{ll}
		\Lambda(x \cdot v, \ell, 1/2,\new{\cdot)}  & \mbox{if } wt(v) < d \\
		\Lambda(x \cdot v, \ell, \tau,\new{\cdot)} & \mbox{if } wt(v) \ge d
	\end{array}
\right.
$$
 where $x \cdot v \in GF(2)^n$ is defined by $(x\cdot v)_i = x_i \cdot 
v_i$.  In other words, if $wt(v) \ge d$, then the distribution 
$\Lambda^a(x,v, \ell, \tau)$ projects $x$ into the non-zero coordinates 
of $v$ and then outputs samples corresponding to exact LPN for the 
projected vector. We define the oracle $\mathcal{O}_{x,\ell, d, 
\tau}^a (\cdot)$ which takes an input $v \in GF(2)^n$ and outputs a 
random sample from $\Lambda^a(x,v, \ell, \tau,\cdot)$. The subset 
exact LPN assumption states the following:

 \begin{assumption}\label{assumption:exactsubset}
  Let $\tau \in (0,1/2)$, $\ell = c \cdot n$ and $d =c' \cdot n$ for 
some $0 < c,c' < 1/2$. The \emph{$(t,\epsilon)$-subset exact LPN 
assumption} says that for any algorithm $\mathcal{A}$ running in time 
$t$,
$$
\left| \Pr_{x \in GF(2)^n}  \left[ \mathcal{A}^{\mathcal{O}_{x,\ell,d,\tau}^a} =1 \right] -  \Pr_{x \in GF(2)^n}  \left[ \mathcal{A}^{\mathcal{O}_{x,\ell,d,1/2}^a} =1 \right] \right| \le \epsilon.
  $$
For the sake of brevity, we henceforth refer to this assumption by
saying ``the subset exact $(n, \ell, d,\tau)$ 
LPN problem is $(t,\epsilon)$-hard.''
 \end{assumption}

 Assumption~\ref{assumption:exactsubset} is very similar to the 
\emph{subset LPN assumption} 
used in \cite{Jain:11} and previously considered in \cite{Pie:12}. 
The subset LPN assumption is the same as 
Assumption~\ref{assumption:exactsubset} but with $\ell=1$ and the 
\new{coordinates of the noise vector $e$ being drawn 
independently from $\mB_{\tau}$.}
Pietrzak~\cite{Pie:12} 
showed that the subset LPN assumption is implied by the standard LPN 
assumption (Assumption~\ref{assumption:LPN}) with a minor change in the 
security parameters.
 Along the same lines, the next lemma shows that 
Assumption~\ref{assumption:LPNvar} implies 
Assumption~\ref{assumption:exactsubset} with a minor change in parameters. 
The proof is identical to the proof of Theorem~1 in \cite{Pie:12} 
and hence we do not repeat it here.

 \begin{lemma}\label{lem:pie}
 If the exact $(n, \ell,\tau)$ LPN problem is $(t, \epsilon)$ hard, \nnew{then for any $g \in \mathbb{N}$}, 
the subset exact $(n', \ell, n + g, \tau)$ LPN problem is $(t', 
\epsilon')$ hard for $n' \ge n+g$, $t' = t/2$ and $\epsilon' = \epsilon 
+ \frac{2t}{2^{g+1}}$.
 \end{lemma}

\begin{proof}
The proof of this lemma follows verbatim from the proof of Theorem~1 in
\cite{Pie:12}. The key observation is that the reduction from subset 
LPN to LPN in Theorem~1 in \cite{Pie:12} is independent of the noise 
distribution. 
\end{proof}

From Lemma~\ref{lem:pie}, we get that Assumption~\ref{assumption:LPNvar} 
implies Assumption~\ref{assumption:exactsubset}. In particular, we can 
set $\ell = n/5$ and $g = n/10$, $n' \ge 11n/10$. Then we get that if 
the exact $(n, \ell,\tau)$ problem is $(2^{n^{\beta}}, 2^{-n^{\beta}})$ 
hard for some $\beta>0$, then the subset exact $(n', \ell, 11n/10, 
\tau)$ is also $(2^{n^{\beta'}}, 2^{-n^{\beta'}})$ hard for some other 
$\beta'>0$.  For the rest of this section, we set the value of $\ell$ 
and $g$ as above and we assume that the subset exact $(n', \ell, 
11n/10, \tau)$ is $(2^{n^{\beta'}}, 2^{-n^{\beta'}})$ hard for some 
$\beta'>0$. 

Now we are ready to define the following Message Authentication Code 
(MAC) $(G,S,V)$, which we refer to as
\textsf{LPN-MAC}:
 
\begin{itemize}

\item The key generation algorithm $G$ chooses a random matrix $X \in 
GF(2)^{\lambda \times n}$ and a string 
$x' \in GF(2)^{\lambda }$, where $\lambda = 2n$.

\item The tagging algorithm samples $R \in GF(2)^{\ell \times \lambda}$ 
and $e \in GF(2)^\ell$ where $e$ is a randomly chosen vector in 
$GF(2)^{\ell}$ with at most $\lceil \tau \ell \rceil$ 
ones. The algorithm outputs $(R, R^T \cdot (X \cdot m + x')+e)$.

\item The verification algorithm, given tag $(R,Z)$ for message $m$, 
computes $y = Z + R^T \cdot (X \cdot m + x')$ and accepts if and only if 
the total number of ones in $y$ is at most $\lceil \tau \ell\rceil $.

\end{itemize}

Note that all arithmetic operations in the description of the above MAC
are done over $GF(2)$. The following theorem shows that under suitable
assumptions the above MAC is special and secure as desired:

\begin{theorem}\label{thm:unforge-mac} Assuming that the exact $(n, 
\ell, \tau)$ problem is $(t, \epsilon)$ hard for $t= 2^{n^{\beta}}$ and 
$\epsilon = 2^{-n^{\beta}}$ for $\beta>0$, \textsf{LPN-MAC} described 
above is a $(t', \epsilon')$-RMA-secure special MAC for $t ' = 
2^{n^{\beta'}}$ and $\epsilon' = 2^{-n^{\beta'}}$ for some $\beta'>0$.  
\end{theorem}

\begin{proof} First, it is trivial to observe that the MAC described 
above is a special MAC. Thus, we are only left with the task of proving 
the security of this construction.  In \cite{Jain:11} (Theorem~5), the 
authors show that the above MAC is secure with the above parameters 
under Assumption~\ref{assumption:LPN} provided the vector $e$ in the 
description of \textsf{LPN-MAC} is drawn from a distribution where every 
coordinate of $e$ is an independent draw from $\mB_{\tau}$. 
(We note that the MAC of Theorem~5 in \cite{Jain:11} is described in a slightly 
different way, but Dodis \etal \cite{Dodis:12} show that the above MAC and 
the MAC of Theorem~5 in \cite{Jain:11} are exactly the same).  
Follow the same proof verbatim except whenever \cite{Jain:11} use the 
subset LPN assumption, we use the subset exact LPN assumption 
(i.e.~Assumption~\ref{assumption:exactsubset}), we obtain a
proof of Theorem~\ref{thm:unforge-mac}. \end{proof}

\subsubsection{A problem for which inverse approximate uniform generation
is hard but approximate uniform generation is easy.}

Given Theorem~\ref{thm:hardness-MAC}, in order to come up with
a problem where inverse approximate uniform generation is hard
but approximate uniform generation is easy,
it remains only to show 
that the verification algorithm for \textsf{LPN-MAC} can be 
implemented in a class of functions for which approximate uniform 
generation is easy. Towards this, we have the following definition.

\begin{definition}\label{def:BI} 
\textsf{BILINEAR-MAJORITY$_{\ell,n,\lambda,\tau}$} is a class of Boolean 
functions such that every 
$f \in$\textsf{BILINEAR-MAJORITY$_{\ell,n,\lambda,\tau}$},
$f : GF(2)^{\ell \times \lambda} \times 
GF(2)^{\ell} \times GF(2)^{n} \rightarrow \{0,1\}$ is parameterized by 
subsets $S_1, \ldots, S_{\lambda} \subseteq [n] $ and $x^0 \in 
GF(2)^{\lambda}$ and is defined as follows : On input $(R,Z,m) \in 
GF(2)^{\ell \times \lambda} \times GF(2)^{\ell} \times GF(2)^{n}$, 
define

$$ y_i = Z_i + \sum_{j=1}^{\lambda} R_{ij} \cdot (\sum_{\ell \in S_j} 
m_{\ell} + x^0_j) $$

where all the additions and multiplications are in $GF(2)$. Then 
$f(R,Z,m)=1$ if and only if at most $\lceil \tau \ell \rceil$ coordinates
$y_1,\dots,y_\ell$ are $1$. \end{definition} 

\begin{claim} For the \textsf{LPN-MAC} 
with parameters $\ell$, $n$, $\lambda$ and $\tau$ 
described earlier, the verification algorithm $V$ can be implemented in 
the class \textsf{BILINEAR-MAJORITY$_{\ell,n,\lambda,\tau}$}. 
\end{claim} 

\begin{proof} Consider the \textsf{LPN-MAC} with parameters 
$\ell$, $n$, $\lambda$ and $\tau$ and secret key $X$ and $x'$. Now
define a function $f$ in 
\textsf{BILINEAR-MAJORITY$_{\ell,n,\lambda,\tau}$} where $x^0 = x'$ and 
the subset $S_j = \{ i : X_{ji} =1 \}$. It is easy to check that 
the corresponding $f(R,Z,m)=1$ if and only if $(R,Z)$ is a valid tag for 
message $m$. \end{proof} 

The next and final claim says that there is an efficient approximate
uniform generation algorithm for
\textsf{BILINEAR-MAJORITY$_{\ell,n,\lambda,\tau}$}:

\begin{claim} 
There is an algorithm which given any $f 
\in$\textsf{BILINEAR-MAJORITY$_{\ell,n,\lambda,\tau}$} (with parameters 
$S_1, \ldots, S_{\lambda} \subseteq [n] $ and $x^0 \in GF(2)^{\lambda}$) 
and an input parameter $\delta>0$, runs in time $\poly(n, \ell, \lambda, 
\log (1/\delta))$ and outputs a distribution which is $\delta$-close to 
being uniform on $f^{-1}(1)$. \end{claim} 

\begin{proof} The crucial observation is that for any $(R,m)$, the set 
$\mathcal{A}_{R,m} = \{z : f(R,Z,m)=1 \}$ has cardinality independent of 
$R$ and $m$. This is because after we fix $R$ and $m$, if we define $b_i 
=\sum_{j=1}^{\lambda} R_{ij} \cdot (\sum_{\ell \in S_j} m_{\ell} + 
x^0_j)$, then $y_i = Z_i + b_i$. Thus, for every fixing of $R$ and $m$, 
since $b_i$ is fixed, the set of those $Z$ such that the number of 
$y_i$'s which are $1$ is bounded by $\tau \ell$ is independent of $R$ 
and $m$. This implies that the following sampling algorithm 
returns a uniformly random element of $f^{-1}(1)$: 

\begin{itemize} 

\item Randomly sample $R$ and $m$. Compute $b_i$ as defined earlier. 

\item Let 
$a = \lceil \tau \ell \rceil$ and consider the halfspace $g(y) = sign 
(a - \sum_{i=1}^\ell y_i )$. Now, we use Theorem~\ref{thm:dyer-gen} to 
sample uniformly at random from $g^{-1}(1)$ and hence draw a uniformly 
random $y$ from the set $ \{ y \in \{0,1\}^{\ell} : \sum_{i=1}^{\ell} 
y_i \le a \}$. \item We set $Z_i = y_i + b_i$. Output $(R,Z,m)$. 
\end{itemize} The guarantee on the running time of the procedure follows 
simply by using the running time of Theorem~\ref{thm:dyer-gen}. 
Similarly, the statistical distance of the output from the uniform 
distribution on $f^{-1}(1)$ is at most $\delta$. 
\end{proof}

\section{Efficient inverse approximate
uniform generation when approximate uniform generation is
infeasible}\label{sec:graph-auto}

In Section~\ref{sec:LTF} we gave an efficient algorithm for the inverse
approximate uniform generation problem for halfspaces, and in
Section~\ref{sec:DNF} we gave a quasi-polynomial time algorithm for the
inverse approximate uniform generation problem for DNFs. Since
both these algorithms follow the standard approach, both crucially use
efficient algorithms for the corresponding uniform generation
problems~\cite{KLM89, MorrisSinclair:04}. In this context, it is
natural to ask the following question: \emph{Is inverse approximate uniform
generation easy only if the corresponding approximate uniform
generation problem is easy?}

In this section we show that the answer to this question is ``no''
(for at least two reasons).
First, we point out that a negative answer follows easily from the
well-known fact that it is computationally hard to ``detect unique
solutions.''
In more detail, we recall the definition of
the \textsf{UNIQUE-SAT} problem.  \textsf{UNIQUE-SAT} is a promise problem
where given a CNF $\Phi$, the task is to distinguish between the
following two cases:

\begin{itemize}

\item $\Phi$ has no satisfying assignment; versus

\item $\Phi$ has \emph{exactly} one satisfying assignment.

\end{itemize}

In a famous result, Valiant and Vazirani~\cite{valvaz} showed the following.

\begin{theorem}\cite{valvaz}
There is a randomized polynomial time reduction from \textsf{CNF-SAT} to
\textsf{UNIQUE-SAT}.

\end{theorem}

Let $\C$ denote the class of all $n$-variable CNF formulas that have exactly
one satisfying assignment.
As an immediate corollary of Theorem~\cite{valvaz} we have the following:

\begin{corollary}\label{cor:valvaz}
There is a constant $c>0$ such that unless \textsf{SAT} $\in$
\textsf{BPTIME$(t(n))$}, there is no approximate uniform generation
algorithm for $\C$ which runs in time \textsf{BPTIME$(t(n^c))$}
\new{even for variation distance $\eps = 1/2$}.
\end{corollary}

\ignore{
for any class of functions $\mathcal{C}$, if for
every $f \in \mathcal{C}$, $|f^{-1}(1)| \le t(n)$, then there is an
obvious algorithm which runs in time $O(t(n) \log (1/\epsilon))$ (Just
observe that after getting $O(t(n) \log (1/\epsilon))$ samples from
$f^{-1}(1)$, with probability $1-\epsilon$, we have every element in the
support of the distribution).  In particular, this means that there is a
}
On the other hand, it is clear that there is a
linear time algorithm for the inverse approximate uniform generation problem
for the class $\C$:  simply draw a single example $x$
and output the trivial distribution supported on that one example.

The above simple argument shows that there indeed exist classes
$\mathcal{C}$ where inverse uniform generation is ``easy" but
approximate uniform generation is ``hard", but this example is somewhat
unsatisfying, as the algorithm for inverse approximate uniform generation
is trivial.  It is natural to ask the following
meta-question: is there a class of functions $\mathcal{C}$ such that
approximation uniform generation is hard, but inverse approximate
generation is easy because of a polynomial-time algorithm
that ``uses its samples in a non-trivial way?''
In the rest of this section we give an example of such a problem.

\medskip

\noindent {\bf Efficient inverse approximate uniform generation for
graph automorphism.}
The following problem is more naturally defined in terms of a
relation over combinatorial objects rather than in terms of
a function and its satisfying assignments.
Let us define $\mathcal{G}_n$ to be the set
of all (simple undirected) graphs over vertex set $[n]$
and $\mathbb{S}_n$ to be the symmetric
group over $[n]$. We define the relation $R_{\mathrm{aut}}(G, \sigma)$ over
$\mathcal{G}_n \times \mathbb{S}_n$ as follows:
$R_{\mathrm{aut}}(G,\sigma)$ holds if
and only if $\sigma$ is an automorphism for the graph $G$. (Recall that
``$\sigma$ is an automorphism for graph $G$'' means that $(x,y)$ is an
edge in $G$ if and only if $(\sigma(x), \sigma(y))$ is also an edge in
$G$.)
The inverse approximate uniform generation problem for the relation
$R_{\mathrm{aut}}$ is then as follows:  There is an unknown $n$-vertex graph $G$.
The algorithm receives uniformly random samples
from the set $\mathrm{Aut}(G) := \{\sigma  \in \mathbb{S}_n \ : \
R_{\mathrm{aut}}(G,\sigma) \textrm{ holds } \}$.  On input
$\eps,\delta$, with probability $1-\delta$ the algorithm must output a
sampler whose output distribution is $\epsilon$-close to the uniform
distribution over $\mathrm{Aut}(G)$.

It is easy to see that $\mathrm{Aut}(G)$ is
a subgroup of $\mathbb{S}_n$, and hence the identity permutation $e_n$
must belong to $\mathrm{Aut}(G)$. To understand the complexity of
this problem we recall the graph isomorphism problem:

\begin{definition}\label{def:auto}
\textsf{GRAPH-ISOMORPHISM} is defined as follows : The input is a pair
of graphs $G_1,G_2 \in \mathcal{G}_n$ and the goal is to determine whether they are isomorphic.
\end{definition}

While it is known that \textsf{GRAPH-ISOMORPHISM} is unlikely to be
NP-complete \cite{Schoning88, BHZ:87}, even after several decades of effort the
fastest known algorithm for
\textsf{ GRAPH-ISOMORPHISM} has a running time
of $2^{\tilde{O}(\sqrt{n})}$~\cite{Babai:1981}. This gives strong
empirical evidence that \textsf{ GRAPH-ISOMORPHISM} is a computationally
hard problem.
The following claim establishes that approximate uniform generation for
$R_{\mathrm{aut}}$ is as hard as \textsf{GRAPH-ISOMORPHISM}:

\begin{claim}
If there is a $t(n)$-time algorithm for approximate uniform generation
for the relation $R_{\mathrm{aut}}$ (with error $1/2$), then
\new{for some absolute constant $c>0$} there is
a \new{$\poly(t(n^c))$-time
randomized} algorithm for \textsf{GRAPH-ISOMORPHISM}.

\end{claim}

\begin{proof}
\new{
Let $A$ be the hypothesized $t(n)$-time algorithm,
so $A$, run on input $(G,1/2)$ \new{where $G$ is an
$n$-node graph}, returns an element $\sigma \in \mathrm{Aut}(G)$
drawn from a distribution $D$ that has $\dtv(D,\U_{\mathrm{Aut}(G)}) \leq 1/2.$
Given such an algorithm $A$, it is easy in $O(t(n))$ time
to determine (with high constant probability of
correctness) whether or not $|\mathrm{Aut}(G)|>1$.
Now the claim follows from the known fact \cite{Hoffmann:82}
that there is a polynomial-time reduction from \textsf{GRAPH-ISOMORPHISM} to
the problem of determining whether an input graph has $|\mathrm{Aut}(G)|>1.$
}
\ignore{
It is known that the problem of
producing a set of generators for $\mathrm{Aut}(G)$ given $G$
is polynomial time Turing reducible to the graph isomorphism problem
\cite{Hoffmann:82}. It is also known that given generators
of a permutation group, it is possible to generate \new{approximately uniform}
random elements of the group in polynomial time~\cite{FHL80}.
Hence, generating random
elements of the permutation group is Turing equivalent to
\textsf{GRAPH-AUTOMORPHISM}. On the other hand, it is obvious that if
one can generate a distribution which is $1/3$ close to
$\mathcal{U}_{\mathrm{Aut}(G)}$, then one can decide if
$|\mathrm{Aut}(G)|>1$.
} \end{proof}

While approximate uniform generation for $R_{\mathrm{aut}}$ is hard,
the next theorem shows that the \emph{inverse}
approximate uniform generation problem for $R_{\mathrm{aut}}$ is in fact easy:

\begin{theorem}\label{thm:graph-auto-easy}
There is a randomized algorithm ${A}_{\inv}^{\mathrm{aut}}$
with the following property:
The algorithm takes as input $\eps,\delta>0$.  Given
access to uniform random samples from $\mathrm{Aut}(G)$
(where $G$ is an unknown $n$-node graph), $A_\inv^{\mathrm{aut}}$
runs in time $\poly(n,\new{\log(1/\eps),\log(1/\delta)})$
and with probability $1-\delta$ outputs
a sampler $C_{\mathrm{aut}}$ with the
following property : The running time of $C_{\mathrm{aut}}$ is $O(n
\log n + \log (1/\epsilon))$ and the output distribution of
$C_{\mathrm{aut}}$ is $\epsilon$-close to the uniform distribution over
$\mathrm{Aut}(G)$.
 \end{theorem}

\begin{proof}
The central tool in the proof is the following theorem of Alon and
Roichman~\cite{AR94}:

\begin{theorem}\label{ALON}\cite{AR94}
Let $H$ be any group and let $h_1, \ldots, h_k$ be chosen uniformly at
random from $H$. Consider the set $S = \cup_{i=1}^k \{ h_i ,
h_i^{-1}\}$. Then, for $k = O(\log |H| + \log (1/\delta))$,
with probability at least $1- \delta$
the Cayley graph $(H,S)$ has its second largest eigenvalue at most $1/2$.
 \end{theorem}

We now describe our algorithm $A^{\mathrm{aut}}_\inv$.
On input $\eps,\delta$ it draws
$k=O(n \log n + \log(1/\delta))$
permutations $g_1, \ldots, g_k$ from $\mathrm{Aut}(G)$.
It computes $g_1^{-1}, \ldots, g_k^{-1}$ and sets $S =
\cup_{i=1}^k \{ g_i , g_i^{-1}\}$.
The sampler $C_{\mathrm{aut}}$ is defined as follows:
It uses its input random bits to perform a random
walk on the Cayley graph $(\mathrm{Aut}(G),S)$, starting at $e_n$, for $T=O(n
\log n + \log(1/\eps))$ steps; it outputs the element of $H$
which it reaches at the end of the walk.  (Note that in order
to perform this random walk it is not necessary to have $\mathrm{Aut}(G)$ explicitly --
it suffices to explicitly have the set $S$.)

The analysis is simple:  we first observe that every graph $G$ has
an automorphism group of size $|\mathrm{Aut}(G)| \leq n!$.
Theorem~\ref{ALON} then guarantees that with probability at least
$1 - \delta$ the Cayley graph $(\mathrm{Aut}(G),S)$ has
its second eigenvalue bounded by $1/2$.  Assuming that the second eigenvalue
is indeed at most $1/2$, standard results in the theory of random
walks on graphs imply that the
distribution of the location reached at the end of the walk
has variation distance at most $\eps$ from the uniform distribution
over $\mathrm{Aut}(G).$ This concludes the proof.
 \end{proof}

\section{Conclusion and future work} \label{sec:conc}

\ifnum\confversion=0
We have considered inverse problems in approximate uniform
generation for a range of interesting and well-studied classes of functions
including LTFs, DNFs, CNFs, polynomial threshold functions,
and more.  While our findings have determined the computational complexity
of inverse approximate uniform generation for these classes,
several interesting questions and directions remain to be pursued.
We outline some of these directions below.

One natural goal is to extend our results (both positive and negative)
to a wider range of function classes; we list several specific
classes that seem particularly worthy of investigation.\ignore{While
we give very conclusive answers for number of classes of functions
which are interesting from the perspective of learning theory, status of
several classes of functions remain open. We list a couple of them here.}
The first of these is the class of intersections
of two monotone LTFs.  We note that
Morris and Sinclair~\cite{MorrisSinclair:04} gave efficient
approximate uniform generation / counting algorithms
for intersections of two monotone LTFs, but on
the other hand, no distribution independent PAC or $\sq$ learning algorithm is
known for this class {(although quasipoly$(n)$-time algorithms
are known if both LTFs have integer weights that are at most
$\poly(n)$ \cite{KOS:04}).}
The second class is that of poly$(n)$-size decision trees.  Our
DNF result gives a quasipoly$(n/\eps)$-time inverse approximate uniform
generation algorithm for this class; can this be improved to
$\poly(n,1/\eps)$?  We note that in order to obtain such a result one would
presumably have to bypass the ``standard approach,''
since decision trees are not known to be PAC learnable
faster than quasipoly$(n/\eps)$-time under the uniform distribution
on $\{-1,1\}^n$.
{(We further note that while \cite{FOS:08} gives a
reduction from learning the uniform distribution over satisfying assignments
of a decision tree to the problem of PAC learning decision trees under
the uniform distribution, this reduction relies crucially on
the assumption --- implicit in the \cite{FOS:08} framework --- that the
probability mass function of the
hypothesis distribution can be efficiently evaluated on any input $x
\in \{-1,1\}^n$.  In our framework this assumption need not hold so the
\cite{FOS:08} reduction does not apply.)
\new{
Still other natural classes to investigate are context free languages (for which
quasi-polynomial time uniform generation algorithms are known
\cite{GJKSM97}) and various classes of branching programs.  It may also
be of interest to consider similar problems
when the underlying measure is (say) Gaussian or log-concave.
}
}
\else
We have determined the computational complexity
of inverse approximate uniform generation for a wide range of well-studied
classes of functions including LTFs, DNFs, CNFs, polynomial threshold functions,
and more.  However, many directions remain for future work.

One natural goal is to extend our results to a wider range of function classes.
In particular, it would be interesting to understand the complexity
of inverse approximate uniform generation for classes such
as decision trees and intersections of two monotone LTFs.
(We note that approximate counting and uniform generation are trivial for the
former, and have polynomial-time algorithms for the latter by a result
of Morris and Sinclair~\cite{MorrisSinclair:04}.)
\new{
Other natural classes to investigate are context free languages (for which 
quasi-polynomial time uniform generation algorithms are known
\cite{GJKSM97}) and various classes of read-once branching programs.  It may also
be of interest to consider similar problems
when the underlying measure is (say) Gaussian or log-concave.
}
\fi

\ignore{ 
class consisting of intersection of two monotone LTFs. In
particular, there are 
complexity of this question, making even a reasonable guess extremely
tricky. 
algorithms for intersections of two monotone LTFs over the
hypercube~\cite{MorrisSinclair:04} as well as SQ-learning algorithms for
intersections of constant number of LTFs over the uniform
distribution on the hypercube~\cite{KOS:04}. On the other hand,
PAC-learning of intersection of two LTFs over arbitrary
distributions is known to be cryptographically
hard~\cite{KlivansSherstov:06}. All in all, we feel a resolution of this
question will be very interesting.
}

Another interesting direction to pursue is to study
inverse approximate uniform generation for combinatorial problems like
matching and coloring as opposed to the ``boolean function
satisfying assignment''--type
problems that have been the main focus of this paper.  We note
that preliminary arguments suggest that there is a simple
efficient algorithm for inverse approximate uniform generation of perfect
matchings in bipartite graphs.  Similarly, preliminary arguments
suggest that for the range of
parameters for which the ``forward''
approximate uniform generation problem for colorings
is known to be easy (namely, the number $q$ of allowable
colors satisfies $q > 11\Delta/6$ where
$\Delta$ is the degree~\cite{Vigoda99}),
the inverse approximate uniform generation
problem also admits an efficient algorithm.  These preliminary
results give rise to the question
of whether there are similar \emph{combinatorial} problems for which
the complexity of the ``forward''
approximate uniform generation problem is not known and yet we can
determine the complexity of inverse approximate 
\ifnum\confversion=0
uniform generation (like
the group theoretic setting of Section~\ref{sec:graph-auto}).
\else
uniform generation.
\fi

Finally, for many combinatorial problems, the approximate uniform
generation algorithm is to run a Markov chain on the state space. In the
regimes where the uniform generation problem is hard, the Markov chain
does not mix rapidly which is in turn equivalent to
the existence of sparse cuts in the
state space. However, an intriguing possibility arises here:  If one can
show that the state space can be partitioned into a small number of
components such that each component has no sparse cuts, then given
access to a small number of random samples from the state space
(with at least one such example belonging to each component), one may
be able to easily perform approximate uniform generation.
Since the inverse approximate uniform generation algorithms that we consider
have access to random samples, this opens
the possibility of efficient approximate uniform generation algorithms
in such cases. To conclude, we give an example of a natural
combinatorial problem (from statistical physics) where it seems that
this is essentially the situation (although we do not have a formal proof).
This is the 2-D Ising model, for which
the natural Glauber dynamics is known to have exponential mixing
time beyond the critical temperature \cite{Mar98-short}. On the other hand,
it was recently shown that even beyond the critical temperature, if one
fixes the boundary to have the same spin {(all positive
or all negative)} then the mixing time comes
down from exponential to quasipolynomial~\cite{Lubetzky:2009}. While
we do not know of a formal reduction, the fact that fixing the boundary
to the same spin brings down the mixing time of the Glauber dynamics
from exponential to quasipolynomial is ``morally equivalent'' to the
existence of only a single sparse cut in the state space of the
graph~\cite{Sinclair:12}. Finding other such natural examples
is an intriguing goal.

\medskip

\ifnum\confversion=0
\noindent {\bf Acknowledgements.}  We thank Alistair Sinclair for
helpful conversations about approximate counting and approximate uniform
generation.  We thank Luca Trevisan for suggesting the use of
``blow-up'' constructions for hardness results. We thank Jonathan Katz,
Tal Malkin, Krzysztof Pietrzak and Salil Vadhan for answering several
questions about MACs and signature schemes. We are thankful to Kousha Etessami, Elchanan
Mossel, Li-Yang Tan and Thomas Watson for helpful conversations about this work. 
\else
\noindent {\bf Acknowledgements.}  We thank Alistair Sinclair,  
Luca Trevisan, Jonathan Katz,
Tal Malkin, Krzysztof Pietrzak, Salil Vadhan, Kousha Etessami, Elchanan
Mossel, Li-Yang Tan and Thomas Watson for helpful conversations.
\fi

\bibliography{allrefs}

\newcommand{\etalchar}[1]{$^{#1}$}
\begin{thebibliography}{DKPW12}

\bibitem[AD98]{AslamDecatur:98}
J.~Aslam and S.~Decatur.
\newblock Specification and simulation of statistical query algorithms for
  efficiency and noise tolerance.
\newblock {\em Journal of Computer and System Sciences}, 56:191--208, 1998.

\bibitem[AG11]{AG11}
Sanjeev Arora and Rong Ge.
\newblock {New Algorithms for Learning in Presence of Errors}.
\newblock In {\em ICALP 2011}, pages 403--415, 2011.

\bibitem[AR94]{AR94}
N.~Alon and Y.~Roichman.
\newblock {Random Cayley Graphs and Expanders}.
\newblock {\em {Random Structures and Algorithms}}, 5:271--284, 1994.

\bibitem[Bab81]{Babai:1981}
L\'{a}szl\'{o} Babai.
\newblock Moderately exponential bound for graph isomorphism.
\newblock In {\em {Proceedings of the 1981 International FCT-Conference on
  Fundamentals of Computation Theory}}, pages 34--50, 1981.

\bibitem[BFKV97]{BFK+:97}
A.~Blum, A.~Frieze, R.~Kannan, and S.~Vempala.
\newblock A polynomial time algorithm for learning noisy linear threshold
  functions.
\newblock {\em Algorithmica}, 22(1/2):35--52, 1997.

\bibitem[BHZ87]{BHZ:87}
Ravi~B. Boppana, Johan Hastad, and Stathis Zachos.
\newblock Does co-np have short interactive proofs?
\newblock {\em Information Processing Letters}, 25(2):127 -- 132, 1987.

\bibitem[BKW03]{blukalwas03}
Avrim Blum, Adam Kalai, and Hal Wasserman.
\newblock Noise-tolerant learning, the parity problem, and the statistical
  query model.
\newblock {\em Journal of the ACM}, 50(4):506--519, July 2003.

\bibitem[BMS08]{BMS08}
G.~Bresler, E.~Mossel, and A.~Sly.
\newblock {Reconstruction of Markov Random Fields from Samples: Some
  Observations and Algorithms}.
\newblock In {\em {APPROX-RANDOM}}, pages 343--356, 2008.

\bibitem[CS00]{CS00}
R.~Cramer and V.~Shoup.
\newblock {Signature schemes based on the strong RSA assumption}.
\newblock {\em {ACM Trans. Inf. Syst. Secur.}}, 3(3):161--185, 2000.

\bibitem[DDS12a]{DDS12soda}
C.~Daskalakis, I.~Diakonikolas, and R.A. Servedio.
\newblock Learning $k$-modal distributions via testing.
\newblock In {\em SODA}, pages 1371--1385, 2012.

\bibitem[DDS12b]{DDS12stoc}
C.~Daskalakis, I.~Diakonikolas, and R.A. Servedio.
\newblock Learning poisson binomial distributions.
\newblock In {\em STOC}, pages 709--728, 2012.

\bibitem[Dec93]{Decatur:93}
S.~Decatur.
\newblock Statistical queries and faulty {PAC} oracles.
\newblock In {\em Proceedings of the Sixth Workshop on Computational Learning
  Theory}, pages 262--268, 1993.

\bibitem[DGL05]{DGL05}
F.~Denis, R.~Gilleron, and F.~Letouzey.
\newblock {Learning from positive and unlabeled examples}.
\newblock {\em {Theoretical Computer Science}}, 348:70--83, 2005.

\bibitem[DKPW12]{Dodis:12}
Yevgeniy Dodis, Eike Kiltz, Krzysztof Pietrzak, and Daniel Wichs.
\newblock {Message Authentication, Revisited}.
\newblock In {\em ECRYPT12}, pages 355--374, 2012.

\bibitem[DL01]{DL:01}
L.~Devroye and G.~Lugosi.
\newblock {\em Combinatorial methods in density estimation}.
\newblock Springer Series in Statistics, Springer, 2001.

\bibitem[DMR06]{DMR06-short}
C.~Daskalakis, E.~Mossel, and S.~Roch.
\newblock Optimal phylogenetic reconstruction.
\newblock In {\em {STOC}}, pages 159--168, 2006.

\bibitem[Dye03]{Dyer03-short}
M.~Dyer.
\newblock Approximate counting by dynamic programming.
\newblock In {\em STOC}, pages 693--699, 2003.

\bibitem[Fis03]{Fischlin03}
M.~Fischlin.
\newblock {The Cramer-Shoup Strong-RSASignature Scheme Revisited}.
\newblock In {\em {Public Key Cryptography - PKC 2003}}, pages 116--129, 2003.

\bibitem[FOS08]{FOS:08}
Jon Feldman, Ryan O'Donnell, and Rocco~A. Servedio.
\newblock Learning mixtures of product distributions over discrete domains.
\newblock {\em SIAM J. Comput.}, 37(5):1536--1564, 2008.

\bibitem[GJK{\etalchar{+}}97]{GJKSM97}
V.~Gore, M.~Jerrum, S.~Kannan, Z.~Sweedyk, and S.~Mahaney.
\newblock A quasi-polynomial-time algorithm for sampling words from a
  context-free language.
\newblock {\em Inf. Comput.}, 134(1):59--74, 1997.

\bibitem[GLS88]{GLS:88}
Martin Gr{\"o}tschel, L{\'a}szlo Lov{\'a}sz, and Alexander Schrijver.
\newblock {\em {Geometric Algorithms and Combinatorial Optimization}},
  volume~2.
\newblock Springer, 1988.

\bibitem[Gol04]{gol-fou02}
Oded Goldreich.
\newblock {\em {Foundations of Cryptography-volume 2}}.
\newblock Cambridge University Press, Cambridge, 2004.

\bibitem[Hof82]{Hoffmann:82}
Christoph~M. Hoffmann.
\newblock {\em {Group-Theoretic Algorithms and Graph Isomorphism}}, volume 136
  of {\em {Lecture Notes in Computer Science}}.
\newblock Springer, 1982.

\bibitem[HV03]{HayesVigoda03}
Thomas~P. Hayes and Eric Vigoda.
\newblock A non-markovian coupling for randomly sampling colorings.
\newblock In {\em FOCS}, pages 618--627, 2003.

\bibitem[HW10]{HW:10}
S.~Hohenberger and B.~Waters.
\newblock {Constructing Verifiable Random Functions with Large Input Spaces}.
\newblock In {\em {EUROCRYPT}}, pages 656--672, 2010.

\bibitem[Jer95]{Jerrum95}
Mark Jerrum.
\newblock A very simple algorithm for estimating the number of k-colorings of a
  low-degree graph.
\newblock {\em Random Struct. Algorithms}, 7(2):157--166, 1995.

\bibitem[JS89]{JS89a}
Mark Jerrum and Alistair Sinclair.
\newblock Approximating the permanent.
\newblock {\em SIAM J. Comput.}, 18(6):1149--1178, 1989.

\bibitem[JSV04]{JSV04}
Mark Jerrum, Alistair Sinclair, and Eric Vigoda.
\newblock A polynomial-time approximation algorithm for the permanent of a
  matrix with nonnegative entries.
\newblock {\em J. ACM}, 51(4):671--697, 2004.

\bibitem[JVV86]{JVV86}
M.~Jerrum, L.~G. Valiant, and V.~V. Vazirani.
\newblock Random generation of combinatorial structures from a uniform
  distribution.
\newblock {\em Theor. Comput. Sci.}, 43:169--188, 1986.

\bibitem[Kea98]{Kearns:98}
M.~Kearns.
\newblock Efficient noise-tolerant learning from statistical queries.
\newblock {\em Journal of the ACM}, 45(6):983--1006, 1998.

\bibitem[Kha80]{Kha:80}
L.G. Khachiyan.
\newblock Polynomial algorithms in linear programming.
\newblock {\em USSR Computational Mathematics and Mathematical Physics},
  20(1):53 -- 72, 1980.

\bibitem[KL83]{KarpLuby83}
R.M. Karp and M.~Luby.
\newblock Monte-carlo algorithms for enumeration and reliability problems.
\newblock In {\em FOCS}, pages 56--64, 1983.

\bibitem[KLM89]{KLM89}
R.~M. Karp, M.~Luby, and N.~Madras.
\newblock {Monte-Carlo Approximation Algorithms for Enumeration Problems}.
\newblock {\em {Journal of Algorithms}}, 10(3):429--448, 1989.

\bibitem[KOS04]{KOS:04}
A.~Klivans, R.~O'Donnell, and R.~Servedio.
\newblock Learning intersections and thresholds of halfspaces.
\newblock {\em Journal of Computer \& System Sciences}, 68(4):808--840, 2004.

\bibitem[KPC{\etalchar{+}}11]{Jain:11}
Eike Kiltz, Krzysztof Pietrzak, David Cash, Abhishek Jain, and Daniele Venturi.
\newblock {Efficient Authentication from Hard Learning Problems}.
\newblock In {\em ECRYPT11}, pages 7--26, 2011.

\bibitem[KS04]{KlivansServedio:04jcss}
A.~Klivans and R.~Servedio.
\newblock {Learning {DNF} in time {$2^{\tilde{O}(n^{1/3})}$}}.
\newblock {\em Journal of Computer \& System Sciences}, 68(2):303--318, 2004.

\bibitem[KSS10]{KSS10}
Jonathan Katz, Ji~Sun Shin, and Adam Smith.
\newblock {Parallel and Concurrent Security of the HB and HB$^{\mbox{+}}$
  Protocols}.
\newblock {\em {Journal of Cryptology}}, 23(3):402--421, 2010.

\bibitem[LMST]{Lubetzky:2009}
E.~Lubetzky, F.~Martinelli, A.~Sly, and F.~L. Toninelli.
\newblock {Quasi-polynomial mixing of the 2D stochastic Ising model with plus
  boundary up to criticality}.
\newblock To appear in Journal of the European Mathematical Society.

\bibitem[Lys02]{Lys02}
Anna Lysyanskaya.
\newblock Unique signatures and verifiable random functions from the {DH-DDH}
  separation.
\newblock In Moti Yung, editor, {\em Advances in Cryptology --- (CRYPTO 2002)},
  volume 2442 of {\em Lecture Notes in Computer Science}, pages 597--612.
  Springer-Verlag, 2002.

\bibitem[Mar98]{Mar98-short}
F.~Martinelli.
\newblock {Lectures on Glauber dynamics for discrete spin models}.
\newblock In {\em {}}, volume 1717 of {\em Lecture Notes in Mathematics}, pages
  93--191. Springer, 1998.

\bibitem[Mos07]{Mos07}
E.~Mossel.
\newblock {Distorted Metrics on Trees and Phylogenetic Forests}.
\newblock {\em {IEEE/ACM Trans. Comput. Biology Bioinform.}}, 4(1), 2007.

\bibitem[MRV99]{MRV99}
Silvio Micali, Michael~O. Rabin, and Salil~P. Vadhan.
\newblock {Verifiable Random Functions}.
\newblock In {\em Proc.\ 40th IEEE Symposium on Foundations of Computer Science
  (FOCS)}, pages 120--130, 1999.

\bibitem[MS04]{MorrisSinclair:04}
Ben Morris and Alistair Sinclair.
\newblock Random walks on truncated cubes and sampling 0-1 knapsack solutions.
\newblock {\em SIAM J. Comput.}, 34(1):195--226, 2004.

\bibitem[MT94]{MT:94}
W.~Maass and G.~Turan.
\newblock How fast can a threshold gate learn?
\newblock In S.~Hanson, G.~Drastal, and R.~Rivest, editors, {\em Computational
  Learning Theory and Natural Learning Systems}, pages 381--414. MIT Press,
  1994.

\bibitem[Mur71]{Muroga:71}
S.~Muroga.
\newblock {\em Threshold logic and its applications}.
\newblock Wiley-Interscience, New York, 1971.

\bibitem[Pap94]{Papadimtriou}
Christos~H. Papadimitriou.
\newblock {\em {Computational Complexity}}.
\newblock Addison-Wesley, 1994.

\bibitem[Pie12]{Pie:12}
Krzysztof Pietrzak.
\newblock {Subspace LWE}.
\newblock In {\em Theory of Cryptography Conference}, pages 548--563, 2012.

\bibitem[Sch78]{Sch78}
Thomas~J. Schaefer.
\newblock {The Complexity of Satisfiability Problems}.
\newblock In {\em STOC}, pages 216--226, 1978.

\bibitem[Sch88]{Schoning88}
Uwe Sch{\"o}ning.
\newblock {Graph Isomorphism is in the Low Hierarchy}.
\newblock {\em J. Comp. Sys. Sci.}, 37(3):312--323, 1988.

\bibitem[Sin12]{Sinclair:12}
A.~Sinclair.
\newblock Personal communication.
\newblock 2012.

\bibitem[Sip83]{Sipser83-short}
M.~Sipser.
\newblock A complexity-theoretic approach to randomness.
\newblock In {\em {STOC}}, pages 330--335, 1983.

\bibitem[SJ89]{JS89b}
Alistair Sinclair and Mark Jerrum.
\newblock Approximate counting, uniform generation and rapidly mixing markov
  chains.
\newblock {\em Inf. Comput.}, 82(1):93--133, 1989.

\bibitem[Sto83]{Stockmeyer83-short}
L.~Stockmeyer.
\newblock The complexity of approximate counting.
\newblock In {\em {STOC}}, pages 118--126, 1983.

\bibitem[Tov84]{Tov:84}
Craig~A. Tovey.
\newblock A simplified {NP}-complete satisfiability problem.
\newblock {\em {Discrete Applied Mathematics}}, 8(1):85--89, 1984.

\bibitem[Vai89]{Vaidya:89}
P.~Vaidya.
\newblock A new algorithm for minimizing convex functions over convex sets.
\newblock In {\em Proceedings of the Thirtheth Symposium on Foundations of
  Computer Science}, pages 338--343, 1989.

\bibitem[Vai96]{Vaidya:96}
P.~M. Vaidya.
\newblock A new algorithm for minimizing convex functions over convex sets.
\newblock {\em Math.\ Prog.}, 73(3):291--341, 1996.

\bibitem[Val12]{gregvaliantfocs12}
G.~Valiant.
\newblock {Finding Correlations in Subquadratic Time, with Applications to
  Learning Parities and Juntas}.
\newblock In {\em {FOCS}}, 2012.

\bibitem[Ver90]{ver90}
Karsten~A. Verbeurgt.
\newblock Learning {DNF} under the uniform distribution in quasi-polynomial
  time.
\newblock In Mark~A. Fulk, editor, {\em Conference on Learning Theory}, pages
  314--326. Morgan Kaufmann, 1990.

\bibitem[Vig99]{Vigoda99}
Eric Vigoda.
\newblock Improved bounds for sampling colorings.
\newblock In {\em FOCS}, pages 51--59, 1999.

\bibitem[VV86]{valvaz}
Leslie~G. Valiant and V.~V. Vazirani.
\newblock {NP} is as easy as detecting unique solutions.
\newblock {\em Theoretical Computer Science}, 47:85--93, 1986.

\bibitem[Wat12]{Watson:12}
Thomas Watson.
\newblock {The complexity of estimating Min-entropy}.
\newblock Technical Report~70, Electronic Colloquium in Computational
  Complexity, 2012.

\end{thebibliography}
\bibliographystyle{alpha}



\end{document}